\pdfoutput=1
\documentclass[entropy,article,accept,oneauthor,pdftex,12pt,a4paper]{mdpi}
\usepackage{graphicx}
\usepackage{amsmath}
\usepackage{enumerate}
\usepackage{amsfonts}
\usepackage{amssymb}
\usepackage{arydshln}
\usepackage{soul}
\usepackage{mathtools}

\DeclareMathAlphabet{\mathbbold}{U}{mathbbol}{m}{n}
\newcounter{theorem}
\newtheorem{theorem}{Theorem}
\newtheorem{proposition}{Proposition}

\def\labelitemi{$\blacklozenge$}
\graphicspath{{./}}
\usepackage{enumitem,booktabs}
\setitemize{parsep=0pt,itemsep=0pt,leftmargin=.4in}
\setenumerate{parsep=0pt,itemsep=0pt,leftmargin=.4in}

\lastpage{\pageref*{LastPage}}
\doinum{10.3390/e17064364}
\pubvolume{17}
\pubyear{2015}
\externaleditor{Academic Editor: Ali E. Abbas}
\history{Received: 30 March 2015 / Accepted: 15 June 2015 / Published: 19 June 2015}
\input{tcilatex}

\Title{Intransitivity in Theory and in the Real World}

\Author{Alexander Y. Klimenko} 

\address[1]{%
School of Mechanical and Mining Engineering, The University of Queensland, Brisbane QLD 4072, Australia; E-Mail:~klimenko@mech.uq.edu.au.\vspace{-12pt}}

\corres{}

\abstract{This work considers reasons for and implications of discarding the assumption of transitivity---the fundamental postulate in the
utility theory of von Neumann and Morgenstern, the adiabatic accessibility principle of Caratheodory and most other theories related to preferences or 
competition. 
The examples of intransitivity are drawn from 
different fields, such as law, biology and economics. 
This work is intended as a common platform that allows us to discuss intransitivity in the context of different 
disciplines. 
The basic concepts and terms that are needed for consistent treatment of intransitivity in various applications are presented 
and analysed in a unified manner. The analysis points out conditions that necessitate appearance of intransitivity, such as multiplicity of 
preference criteria and imperfect (\textit{i.e.}, approximate) discrimination of different cases. 
The present work observes that with increasing presence and  
strength of intransitivity, thermodynamics gradually fades away leaving space for more general kinetic considerations. Intransitivity 
in competitive systems is linked to complex phenomena that would be difficult or impossible to explain on the basis of transitive assumptions. 
Human preferences that seem irrational from the perspective of the conventional utility theory, become perfectly logical 
in the intransitive and relativistic framework suggested here. 
The example of competitive simulations for the risk/benefit dilemma 
demonstrates the significance of intransitivity in cyclic behaviour and abrupt changes in the system.
The evolutionary intransitivity parameter, 
which is introduced in the Appendix, 
is a general measure of intransitivity, which is particularly useful in evolving competitive systems. 
}

\keyword{intransitivity; complex evolving systems; non-conventional thermodynamics; utility; behavioural economics; population dynamics}

\begin{document}

%\maketitle

\section{Introduction}

The strongest case for the existence of methodological similarity between
utility and entropy is represented by two fundamental results: (a) the
utility theory of von Neumann and Morgenstern \cite{VNM1953} and (b)
introduction of entropy through the adiabatic accessibility principle of
Caratheodory \cite{Cara1909}. The latter approach was rigorously formalised
by Lieb and Yngvason \cite{EntOrd2003}. The physical interpretation of this
mathematical theory is linked to the so-called weight process\ previously
suggested by Gyftopoulos and Beretta \cite{Beretta1991}. Both of these
theories link ordering of states to a ranking quantity (utility $U$ or
entropy $S$) and are based on two fundamental principles:

\begin{enumerate}
\item[(1)] Transitivity

\item[(2)]  Linearity (in thermodynamics: extensivity) implying that 
\begin{eqnarray}
U&=&P_{_{\text{A}}}u_{_{\text{A}}}+P_{_{\text{B}}}u_{_{\text{B}}}  \label{iU}
\\
S&=&m_{_{\text{A}}}s_{_{\text{A}}}+m_{_{\text{B}}}s_{_{\text{B}}}  \label{iS}
\end{eqnarray}
\end{enumerate}

In Equation~(\ref{iU}), the overall lottery is a combination of two outcomes
A an B with utilities $u_{_{\text{A}}}$ and $u_{_{\text{B}}}$ and
probabilities $P_{_{\text{A}}}$and $P_{_{\text{B}}}$. In Equation~(\ref{iS}%
), the overall system is a combination of two subsystems A an B with
specific entropies $s_{_{\text{A}}}$ and $s_{_{\text{B}}}$ and masses $m_{_{%
\text{A}}}$and $m_{_{\text{B}}}$.

While the similarity between utility and entropy is obvious, this similarity
remains methodological: theories a and b are generally applied to different
objects taken from different fields of science. There are however some
exceptions, such as competitive systems \cite{K-PS2012,K-PT2013,K_Ent2014a}.
These systems incorporate competition preferences and, at the same time,
permit thermodynamic considerations (here we refer to apparent
thermodynamics---using approaches developed in physical thermodynamics\
and statistical physics to characterise systems not related to heat and
engines).

Further investigations into human decision-making under risk have revealed
substantial disagreements with von Neumann--Morgenstern utility theory,
indicating that preferences\ depend non-linearly on probabilities. One of
the most prominent examples demonstrating non-linearity of human preferences
is known as the Allais paradox \cite{Allais1953}. A spectrum of
generalisations introducing utilities that are non-linear with respect to
probabilities has appeared \cite{pref1955, pref1982T, pref1982, pref1983},
most notably the cumulative prospect theory~\cite{pref1992}. In
thermodynamics, generalisations of conventional entropy have brought new
formulations for non-extensive entropies \cite{Abe1}, most notably Tsallis
entropy \cite{Tsallis1} and its modifications \cite{Thurner1}. Physically,
the definitions of non-extensive entropies correspond to the existence of
substantial stochastic correlations between subsystems. All of these
theories do not violate or question the first fundamental principle listed
above---transitivity.

In this work, we are interested in the phenomenon of intransitivity, \textit{i.e.},
violations of \mbox{transitivity.} A good example of intransitivity has been known
for a long time under the name of the Condorcet \mbox{paradox~\cite{Cond1785}.} The
existence of intransitivity in human preferences has been repeatedly
\mbox{suggested~\cite{Tversky1969,Rubinstein1988,Temkin1996}} and has its advocates
and critics. The main argument against intransitivity is its perceived
irrationality~\cite{Intrans1964}, which was disputed by Anand \cite%
{Intrans1993} from a philosophical perspective. Critics of intransitivity
often argue that \textquotedblleft abolishing\textquotedblright\
transitivity is wrong as we need utility and entropy, while these quantities
are linked to transitivity. The question, however, is not merely in the
replacement of one assumption by its negation: while transitivity is a
reasonable assumption in many good theories, its limitations are a barrier
for explaining the complexity of the surrounding world. Both transitive and
intransitive effects are common and need to be investigated irrespective of
what we tend to call \textquotedblleft rational\textquotedblright\ or
\textquotedblleft irrational\textquotedblright . As discussed in this work,
intransitivity appears under a number of common conditions and, therefore,
must be ubiquitously present in the real world. We have all indications that
intransitivity is a major factor stimulating emergence of complexity in the
competitive world surrounding us \cite{K-PS2012,K-PT2013}. It is interesting
to note that the presence of intransitivity is acknowledged in some
disciplines \linebreak (e.g., population biology) but is largely overlooked in others
(e.g., economics). This work is intended as a common platform for dealing
with intransitivity across different disciplines.

Sections \ref{S1}, \ref{S4} and \ref{S6} and Appendices present a general
analysis. Sections \ref{S1_2}--\ref{S3} and \ref{S5} present
examples from game theory, law, ecology and behavioral economics. Section %
\ref{S7} presents competitive simulations of the risk/benefit dilemma.
Section \ref{S8} discusses thermodynamic aspects of intransitivity.
Concluding remarks are in Section \ref{S9}.

\section{Preference, Ranking and Co-Ranking\label{S1}}

This section introduces main definitions that are used in the rest of the
paper. The basic notion used here is \textit{preference}, which is denoted
by the binary relation A$\prec$B, or equivalently by B$\succ$A, implying
that element B has some advantage over element A. In the context of a
competitive situation, A$\prec$B means that B is the winner in competition
with A. The notation A$\preceq$B indicates that either element B is
preferred over element A (\textit{i.e.}, A$\prec$B), or the elements are equivalent
(\textit{i.e.}, A$\sim$B, although A and B are not necessarily the same A$\neq$B). The
elements are assumed to be comparable to each other (\textit{i.e.}, possess relative
characteristics), while absolute characteristics of the elements may not
exist at all or remain unknown. If equivalence (indifference) A$\sim$B, is
possible only when A$=$B then this preference is called~\textit{strict}.

The preference is \textit{transitive} when 
\begin{equation}
\text{A}\preceq\text{B}\preceq\text{C}\Longrightarrow\text{A}\preceq\text{C}
\end{equation}
for any three elements A, B and C. Otherwise, existence of at least one
intransitive triplet 
\begin{equation}
\text{A}\preceq\text{B}\preceq\text{C}\prec\text{A}  \label{r_trip}
\end{equation}
indicates \textit{intransitivity} of the preference. Generally, we need to
distinguish \textit{current transitivity}---\textit{i.e.}, transitivity of
preference on the current set of elements---from the overall transitivity
of the preference rules (if such rules are specified): the latter requires
the former but not vice versa. Intransitive rules may or may not reveal
intransitivity on a specific set of elements. Intransitivity is called 
\textit{potential} if it can appear under considered conditions but may or
may not be revealed on the current set of elements.

\subsection{Co-Ranking}

The preference of B over A can be equivalently expressed by a \textit{%
co-ranking }function $\rho ($B,A$)$ so that 
\begin{equation}
\text{A}\preceq \text{B }\Longleftrightarrow \rho (\text{A,B})\leq 0
\end{equation}%

This implies the following functional form for co-ranking: 
\begin{equation}
\rho (\text{B,A})=\left\{ 
\begin{array}{c}
\rho (\text{B,A})>0\;\;\text{if\ B}\succ \text{A} \\ 
\rho (\text{B,A})=0\;\;\text{if \ B}\sim \text{A} \\ 
\rho (\text{B,A})<0\;\;\text{if \ B}\prec \text{A}%
\end{array}%
\right.   \label{r_rho}
\end{equation}%

By definition, co-ranking is antisymmetric: 
\begin{equation}
\rho (\text{B,A})=-\rho (\text{A,B})
\end{equation}%

Co-ranking is a relative characteristic that specifies properties of one
element with respect to the other, while absolute characteristics of the
elements may not exist or be unknown. Co-ranking can be \textit{graded},
when the value $\rho ($B,A$)$ is indicative of the strength of our
preference, or \textit{sharp} otherwise. Generally we presume graded
co-ranking. However, a graded co-ranking may or may not be specified (and
exist) for a given preference. If ranking is sharp, only the sign of $\rho ($%
B,A$)$ is of interest while the magnitude $\rho ($B,A$)$\ is an arbitrary
value. The\ following definition of the \textit{indicator co-ranking} 
\begin{equation}
R(\text{B,A})=\func{sign}\left( \rho (\text{B,A})\right) =\left\{ 
\begin{array}{c}
+1\;\;\text{if\ B}\succ \text{A} \\ 
0\;\;\;\;\text{if \ B}\sim \text{A} \\ 
-1\;\;\text{if \ B}\prec \text{A}%
\end{array}%
\right. 
\end{equation}%
corresponds to information available in sharp preferences. The indicator
co-ranking $R$ is a special case of co-ranking $\rho $. The function $R($B,A$%
)$ can be also called the indicator function of the preference.

\subsection{Absolute Ranking\label{S1b}}

If the preference is transitive, it can be expressed with the use of a
numerical function $r($...$)$ called \textit{absolute ranking} so that 
\begin{equation}
\text{A}\preceq \text{B }\Longleftrightarrow r(\text{A})\leq r(\text{B})
\end{equation}%
for any A and B. 
(We consider mainly discrete systems but, in the case of continuous
distributions, the existence of absolute rankings for transitive preferences
is subject to the conditions of the Debreu \mbox{theorem}~\cite{Debreu1964} (\textit{i.e.},
continuity of the preferences), which are presumed to be satisfied in this
work.) Ranking is called \textit{strict} when the corresponding preference
is strict\ (\textit{i.e.}, $r($A$)=r($B$)$ demands A$=$B).\ As with co-rankings, we
distinguish \textit{graded} and \textit{sharp }rankings. In the case of the
sharp ranking, the value $r($B$)-r($A$)>0$ tells us only that B is better
than A but, generally, does not give any indication of the magnitude of our
preference. From a mathematical perspective, a sharp ranking represents a%
\textit{\ total pre-ordering} of the ranked elements, while a sharp strict
ranking represents an \textit{ordering}.\ Any strictly monotonic function $%
f_{m}(...)$ of a sharp ranking is still an equivalent ranking $%
f_{m}(r(...))\sim r(...)$ with the same ordering. In the case of the graded
ranking, the value $r($B$)-r($A$)$ represents the magnitude of our
preference of B over A. In many cases, a graded \ ranking corresponds to a
physical quantity that can be directly determined or measured. In economics,
graded rankings are called \textit{utility}, in biology graded rankings are
called \textit{fitness}, in thermodynamics graded rankings correspond to 
\textit{entropy}. Here, we follow the notation of economics and, when
applicable, refer to graded rankings as utility. Practically, the line
between sharp and graded rankings is blurred. It is often the case that even
a nominally sharp ranking can give some indication of the magnitude of the
preference, for example, in terms of the density of elements. Different
rankings (or co-rankings) are called \textit{equivalent} if they correspond
to the same preference (but might still can have different preference
magnitudes).
%we moved the footnote in the main text.

If a co-ranking specifies a transitive preference, there exists an absolute
ranking for this preference (which is not unique). The absolute ranking
induces a co-ranking, which is linked to the absolute ranking by the
relation 
\begin{equation}
\rho(\text{B,A})=r(\text{B})-r(\text{A}),  \label{r_diff}
\end{equation}
and referred to as the \textit{absolutely} \textit{transitive co-ranking}.
By definition, the original co-ranking is equivalent to co-ranking Equation~(\ref%
{r_diff}) but does not necessarily coincide with it indicating different
magnitude of the preference. The co-rankings defining intransitive
preferences are called intransitive, while co-rankings that define
transitive preferences on a given set of elements but cannot be represented
by Equation~(\ref{r_diff}) are called potentially intransitive as intransitivity can
appear for these co-rankings under conditions specified in Appendix \ref{AA1}%
.

\subsection{Average Rankings}

We are interested in ranking element A with respect to a group of elements
(that may or may not include A), say, group $\mathbb{G}$ represented by a
set $\mathbb{S}$ of elements C$_{i}\in\mathbb{S}$ and the corresponding
weights $g_{i}=g($C$_{i})>0,$\ while $g($C$_{j})=0$ for C$_{j}\notin\mathbb{S%
}$. We can also write C$_{i}\in\mathbb{G}$ implying that C$_{i}\in\mathbb{S}$%
. The average co-ranking of element A and group $\mathbb{G}$, is defined by
the equation 
\begin{equation}
\rho_{_{\mathbb{G}}}(\text{A})=\bar{\rho}(\text{A,}\mathbb{G})=-\bar{\rho }(%
\mathbb{G}\text{,A})=\frac{1}{G}\sum_{\text{C}_{i}\in\mathbb{G}}g(\text{C}%
_{i})\rho(\text{A}\mathbf{,}\text{C}_{i})  \label{r_G}
\end{equation}
where $G$ is the total weight of the group 
\begin{equation}
G=\sum_{\text{C}_{i}\in\mathbb{G}}g(\text{C}_{i})
\end{equation}

Note that $\bar{\rho}($A,$\mathbb{\{}$C$\mathbb{\}})=\rho($A,C$)$. If all
weights are the same $g_{i}=1$ within the group, then $G$ is the number of
elements in the group and the meanings of $\mathbb{S}$ and $\mathbb{G}$ are
essentially the same while the terms ``set" and ``group" become interchangable
(\textit{i.e.}, specification $\mathbb{G}$ as a set implies unit weights for the
elements). The co-ranking $\rho($A$\mathbf{,}$C$_{i})$ and corresponding
preference are referred to as \textit{underlying} \mbox{co-ranking} and preference\
of the ranking $\rho_{_{\mathbb{G}}},$ while the preference 
\begin{equation}
\text{A}\preceq_{_{\mathbb{G}}}\text{B}\Longleftrightarrow\rho_{_{\mathbb{G}%
}}(\text{A})\leq\rho_{_{\mathbb{G}}}(\text{B})\text{ }
\end{equation}
is called the preference \textit{induced by} the ranking $\rho_{_{\mathbb{G}%
}}$. The ranking $\rho_{_{\mathbb{G}}}$ and the preference induced by $%
\rho_{_{\mathbb{G}}}$ are called \textit{conditional} indicating
conditioning of ranking on $\mathbb{G}$. The group $\mathbb{G}$ and weights $%
g_{i}$ are called the \textit{reference group} and \textit{reference weights.%
}

The co-ranking of two groups $\mathbb{G}^{\prime }$ and $\mathbb{G}^{\prime
\prime },$ which is called \textit{group co-ranking}, is defined in the same~way 
\begin{equation}
\bar{\rho}(\mathbb{G}^{\prime },\mathbb{G}^{\prime \prime })=\frac{1}{%
G^{\prime }G^{\prime \prime }}\sum_{\text{C}_{i}\in \mathbb{G}^{\prime
}}\sum_{\text{C}_{j}\in \mathbb{G}^{\prime \prime }}g^{\prime }(\text{C}%
_{i})g^{\prime \prime }(\text{C}_{j})\rho (\text{C}_{i}\mathbf{,}\text{C}%
_{j})  \label{r_G1G2}
\end{equation}%
where $g^{\prime }$ and $g^{\prime }$ are the weights and $G^{\prime }$ and $%
G^{\prime \prime }$ are the total weights associated with the groups $%
\mathbb{G}^{\prime }$ and $\mathbb{G}^{\prime \prime }$. In case of
continuous distributions, the sums are to be replaced by the corresponding
integrals.

If the underlying preference is specified by an absolutely transitive
co-ranking (as represented by Equation~(\ref{r_diff})), then we can write for the
average co-ranking 
\begin{equation}
\rho_{\mathbb{G}}(\text{A})=\bar{\rho}(\text{A,}\mathbb{G})=\frac{1}{G}\sum_{%
\text{C}_{i}\in\mathbb{G}}g(\text{C}_{i})\left( r(\text{A})-r(\text{C}%
_{i})\right) =r(\text{A})-\bar{r}(\mathbb{G}),\;\;  \label{r_Gt}
\end{equation}
and 
\begin{equation}
\bar{\rho}(\mathbb{G}^{\prime},\mathbb{G}^{\prime\prime})=\frac{1}{G^{\prime
}G_{2}^{\prime\prime}}\sum_{\text{C}_{i}\in\mathbb{G}^{\prime}}\sum _{\text{C%
}_{j}\in\mathbb{G}^{\prime\prime}}g^{\prime}(\text{C}_{i})g^{\prime\prime}(%
\text{B}_{j})\left( r(\text{C}_{i})-r(\text{C}_{j})\right) =\bar{r}(\mathbb{G%
}^{\prime})-\bar{r}(\mathbb{G}^{\prime\prime})
\end{equation}
where 
\begin{equation}
\bar{r}(\mathbb{G})=\frac{1}{G}\sum_{\text{C}_{i}\in\mathbb{G}}g(\text{C}%
_{i})r(\text{C}_{i})  \label{r_Gav}
\end{equation}
is the average absolute ranking of group $\mathbb{G}$.

\begin{proposition}
\label{P1}All conditional rankings indicate the same magnitude of
preference, \textit{i.e.}, 
\begin{equation}
\rho_{_{\mathbb{G}^{\prime}}}(\text{A})-\rho_{_{\mathbb{G}^{\prime}}}(\text{B%
})=\rho_{_{\mathbb{G}^{\prime\prime}}}(\text{A})-\rho_{_{\mathbb{G}%
^{\prime\prime}}}(\text{B})  \label{r_ABGG}
\end{equation}
for any A, B$\in\mathbb{G}$ and any $\mathbb{G}^{\prime}$,$\mathbb{G}%
^{\prime\prime}$ $\subset\mathbb{G}$, if and only if the underlying
co-rankings are absolutely transitive.
\end{proposition}

Equation~(\ref{r_Gt}) demonstrates validity of the direct part of the
proposition 
\begin{equation}
\rho_{_{\mathbb{G}^{\prime}}}(\text{A})-\rho_{_{\mathbb{G}^{\prime}}}(\text{B%
})=r(\text{A})-r(\text{B})\text{ }
\end{equation}

The inverse part can be shown by considering groups represented by
one-element sets $\mathbb{G}^{\prime}=\{$C$\}$ and $\mathbb{G}%
^{\prime\prime}=\{$D$\}$. Then Equation~(\ref{r_ABGG}) becomes%
\begin{equation}
\rho(\text{A,C})-\rho(\text{B,C})=\rho(\text{A,D})-\rho(\text{B,D})
\end{equation}
and in, particular, if D $=$ B then $\rho($A,B$)=\rho($A,C$)-\rho($B,C$)$ for
all A, B and C. Hence, we can define absolute ranking by the following
relation: 
\begin{equation}
r(\text{B})=\rho(\text{B,A})+r(\text{A})
\end{equation}
for arbitrary B and fixed A.

\begin{proposition}
\label{P1a}If the preferences induced by all conditional rankings are
equivalent, \textit{i.e.}, 
\begin{equation}
\text{A}\preceq_{_{\mathbb{G}^{\prime}}}\text{B}\Longleftrightarrow \text{A}%
\preceq_{_{\mathbb{G}^{\prime\prime}}}\text{B}  \label{r_GG1}
\end{equation}
for any A,B$\in\mathbb{G}$ and any $\mathbb{G}^{\prime}$,$\mathbb{G}%
^{\prime\prime}$ $\subset\mathbb{G}$, then the underlying preference is
currently transitive.
\end{proposition}

We consider intransitive triplet (\ref{r_trip}) (\textit{i.e.}, A$\preceq$B$\preceq $C$%
\prec$A)---there must be at least one if the preference is intransitive---and demonstrate that at least some of the conditional rankings are
different. For two one-element groups $\mathbb{G}^{\prime}=\{$B$\}$ and $%
\mathbb{G}^{\prime\prime}=\{$C$\},$ we, obviously, have $0\geq\rho _{_{\text{%
\{B\}}}}($A$)\leq\rho_{_{\text{\{B\}}}}($C$)\geq0$ but $0<\rho_{_{\text{\{C\}%
}}}($A$)>\rho_{_{\text{\{C\}}}}($C$)=0.$ Hence, the following conditional
preferences 
\begin{equation}
\text{A}\preceq_{_{\text{\{B\}}}}\text{C \ but\ A}\succ_{_{\text{\{C\}}}}%
\text{C }
\end{equation}
are different. This contradicts (\ref{r_GG1}) implying that the underlying
preference must be transitive.

Hence, if a binary preference is specified, elements in a given set can
always be ordered transitively by conditional ranking of the elements with
respect to a selected reference group or set (which may or may not coincide
with the given set). If the original preference is transitive, it uniquely
determines the ordering irrespective of the reference group. If the original
preference is intransitive, then the relative positions of at least two
elements in this ordering (with respect to each other: say A before B or A
after B) depend on the presence of the other elements in the reference set.

Conditional ranking and group co-ranking can also be introduced for the
indicator co-ranking 
\begin{equation}
R_{\mathbb{G}}(\text{A})=\bar{R}(\text{A,}\mathbb{G})=-\bar{R}(\mathbb{G},%
\text{A})=\frac{1}{G}\sum_{\text{C}_{i}\in \mathbb{G}}g(\text{C}_{i})R(\text{%
A}\mathbf{,}\text{C}_{i})  \label{r_R_G}
\end{equation}%
\begin{equation}
\bar{R}(\mathbb{G}^{\prime },\mathbb{G}^{\prime \prime })=\frac{1}{G^{\prime
}G^{\prime \prime }}\sum_{\text{C}_{i}\in \mathbb{G}^{\prime }}\sum_{\text{C}%
_{j}\in \mathbb{G}^{\prime \prime }}g^{\prime }(\text{C}_{i})g^{\prime
\prime }(\text{B}_{j})R(\text{C}_{i}\mathbf{,}\text{C}_{j})  \label{r_R_GG}
\end{equation}%
which are distinguished by using the word ``indicator''. For example, conditional indicator co-rankings
are further discussed in Appendix \ref{AA2}. Note that the group preferences
induced by the indicator co-rankings ($\mathbb{G}^{\prime }\succeq _{R}%
\mathbb{G}^{\prime \prime }$ $\Longleftrightarrow $ $\bar{R}(\mathbb{G}%
^{\prime },\mathbb{G}^{\prime \prime })\geq 0$) can be intransitive even if
the underlying element preference is currently transitive (see Proposition %
\ref{PA1aaa}). All co-rankings and conditional rankings (such as $\rho $, $%
R, $ $\rho _{_{\mathbb{G}}},$ $R_{\mathbb{G}},$ $\bar{\rho}$, $\bar{R}$) are 
\textit{relative} as opposite to the \textit{absolute} ranking of the
previous subsection.

\subsection{Is Intransitivity Irrational?}

Preferences are always attached to specific conditions and can become
illogical or contradictory if taken out of context and this work endeavors
to use examples to illustrate this point. The propositions of the previous
subsections demonstrate that intransitivity is associated with relativistic
views. The choice between transitive (absolute) and intransitive
(relativistic) models depends on nature of the processes that these models
are expected to reproduce. Many people, however, have psychological
difficulties in accepting a relativistic approach, expecting an absolute
scale of judgments from ``bad'' to ``good'', which can be suitable in some
cases but excessively simplistic in the others.

The argument for irrationality of intransitivity \cite{Intrans1964} is based
on the alleged impossibility of choosing between A, B and C when A$\prec$B, B%
$\prec$C and C$\prec$A. It is often suggested that in this case a decision
maker will circle between these thee options indefinitely, which is
impossible or open to numerous contradictions. In fact, binary preferences
determine the selection of a single element from pairs, but do not tell us
how the choice should be performed when all three elements are
simultaneously present in the current set of $\{$A,B,C$\}$. This can be
achieved consistently on the basis of the conditional ranking $\rho_{_{%
\mathbb{G}}}(...)=\bar{\rho}(...,\mathbb{G}),$ where $\mathbb{G=}\{$A,B,C$%
\}. $\ It might be the case that $\rho_{_{\mathbb{G}}}($A$)=\rho_{_{\mathbb{G%
}}}( $B$)=\rho_{_{\mathbb{G}}}($C$)$ but this case is no more illogical than
selecting from a set with transitive preferences and equivalent elements,
e.g., A$\sim$B$\sim$C$\sim$A. While identifying the reference set with the
current set is most obvious choice in absence of additional information,
deploying alternative reference sets is also possible in practical
situations. For example, a buyer may use information about popularity of
different models instead of considering the set of models currently
available in the store. In a generic consideration of this subsection, we
follow the logical choice of identifying the reference set with the current
set and setting the reference weights to unity.

The preference between elements A and B selected from the current set of
\{A,B\} is based on the conditional ranking 
\begin{equation}
\rho_{_{\{\text{A},\text{B}\}}}(\text{A})=\frac{\rho(\text{A,B})}{2}%
,\;\;\rho_{_{\{\text{A},\text{B}\}}}(\text{B})=\frac{\rho(\text{B,A})}{2}
\end{equation}
as specified by Equation~(\ref{r_G}) with $g_{i}=1$. The conditional co-ranking,
which is a measure of conditional preference of A over B introduced by
similarity with Equation~(\ref{r_diff}), becomes 
\begin{equation}
\rho_{_{\{\text{A},\text{B}\}}}(\text{A,B})=\rho_{_{\{\text{A},\text{B}\}}}(%
\text{A})-\rho_{_{\{\text{A},\text{B}\}}}(\text{B})=\rho(\text{A,B})
\end{equation}
This is expected: $\rho($A,B$)$ indicates a preference between A and B when
selected from the set of \{A,B\}. Let us compare this to preferences between
A and B when the current (reference) set has three elements $\{$A,B,C$\}$.
The conditional ranking is now specified by 
\begin{equation}
\rho_{_{\{\text{A},\text{B,C}\}}}(\text{A})=\frac{\rho(\text{A,B})+\rho(%
\text{A,C})}{3},\;\;\rho_{_{\{\text{A},\text{B,C}\}}}(\text{B})=\frac{\rho(%
\text{B,A})+\rho(\text{B,C})}{3}
\end{equation}
with the corresponding conditional co-ranking 
\begin{equation}
\rho_{_{\{\text{A},\text{B,C}\}}}(\text{A,B})=\rho_{_{\{\text{A},\text{B,C}%
\}}}(\text{A})-\rho_{_{\{\text{A},\text{B,C}\}}}(\text{B})=\frac {2\rho(%
\text{A,B})+\rho_{_{\{\text{C}\}}}(\text{A,B})}{3}=\rho(\text{A,B})+\frac{%
\delta(\text{C,B,A})}{3}
\end{equation}
where 
\begin{equation}
\rho_{_{\{\text{C}\}}}(\text{A},\text{B})=\rho(\text{A},\text{C})-\rho(\text{%
B},\text{C}),\;\;\delta(\text{C,B,A})=\rho(\text{C,B})+\rho(\text{B,A})+\rho(%
\text{A,C})
\end{equation}
are introduced. In absolutely transitive cases, $\rho_{_{\{\text{C}\}}}($A,B$%
)=\rho($A,B$),$ $\delta($C,B,A$)=0$ and $\rho_{_{\{\text{A},\text{B,C}\}}}($%
A,B$)$ is the same as $\rho($A,B$)$. However, $\rho _{_{\{\text{A},\text{B,C}%
\}}}($A,B$)$ and $\rho($A,B$)$ are not necessarily equivalent in
intransitive cases. The combination of intransitivity with a presumption of
invariance of conditional rankings, which is incorrect in intransitive
cases, may result in logical contradictions. In a consistent approach, the
specification of preference between A and B for selection from the set of $%
\{ $A,B,C$\}$ should be $\rho_{_{\{\text{A},\text{B,C}\}}}($A,B$)$ and not $%
\rho($A,B$)$.

If only relative characteristics (\textit{i.e.}, co-ranking but not absolute ranking)
are specified or known, the preferences defined by these characteristics are
most likely to be intransitive or potentially intransitive. The case of
absolutely transitive co-rankings is a very specific case and, in general,
cannot be presumed a priori. The existence of a transitive ordering can
simplify choices but \textquotedblleft simpler\textquotedblright\ does not
necessarily mean \textquotedblleft more accurate\textquotedblright\ or
\textquotedblleft more realistic\textquotedblright . Intransitivity is not
irrational, but considering intransitivity while neglecting relativistic
nature of conditional preferences is illogical and can lead to
contradictions.\ The examples of the following sections show that our
preferences are indeed relativistic and, generally, intransitive but this is
often (and incorrectly) seen as 
``irrationality'', which is not amendable to logical~analysis.

\section{Intransitivity and Game Theory\label{S1_2}}\vspace{-12pt}

\subsection{\protect\nolinebreak Games with Explicit Intransitivity}

Games with explicit intransitivity involve preset rules that a priori
specify intransitive preferences. The best-known example of such games is
the rock(R)-paper(P)-scissors(S) game. The rules of this game are expressed
by the intransitive preference R$\prec$P$\prec$S$\prec$R, which is
illustrated in Figure \ref{fig0}a. The players $\mathfrak{P}^{\prime}$ and $%
\mathfrak{P}^{\prime\prime}$ independently select one of the options (pure
strategies) R, P or S and the winner is determined by the specified
preference. Obviously, intransitive rules are related to intransitive
strategies. Indeed, if a player has to select between two options, while the
remaining option has to be taken by the opponent, the player would obviously
prefer P to R, S to P and R to S. Hence, appearance of intransitive
strategies in explicitly intransitive games is natural and common.
\begin{figure}[H]
\centering

\includegraphics[width=12cm,page=10, clip ]{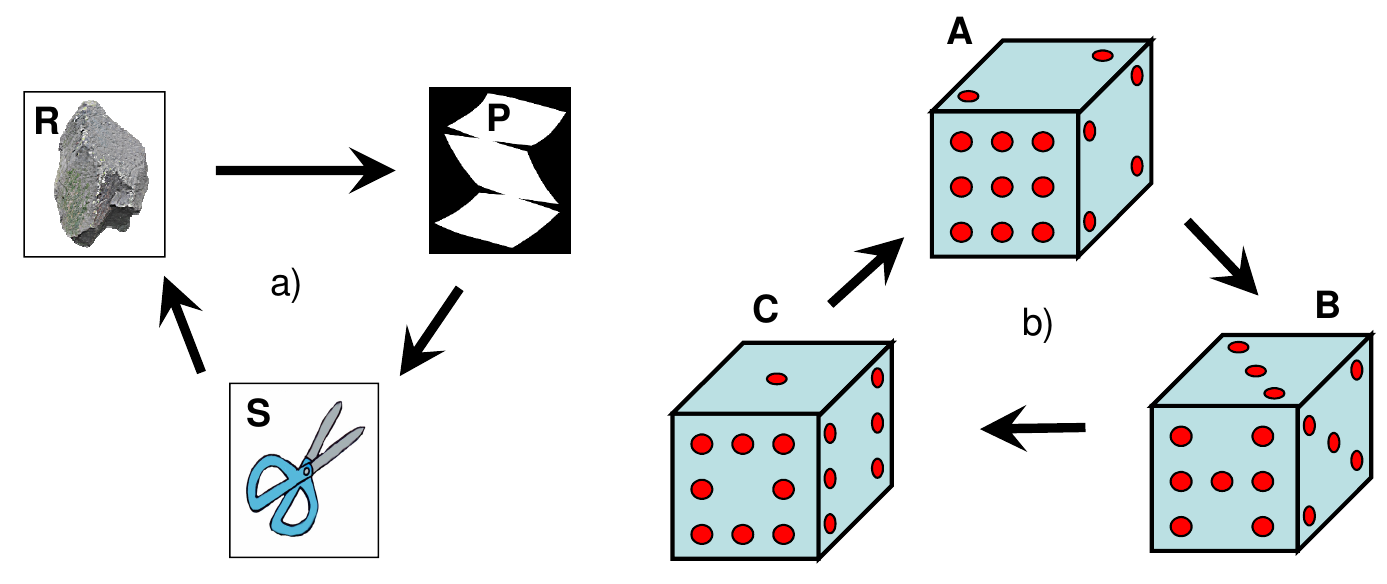}
\caption{ (\textbf{a}) The rock-paper-scissors game is the best-known example of intransitivity in games;
(\textbf{b}) Intransitive dice game where the die thrown by a player wins when it has a higher score
 than the die thrown by the opposition. The opposite sides of the dice have the same numbers.
With the probability of 5/9, die B wins over A, C wins over B, and A wins over C. 
This dice game is a simple version of Efron's dice (see article ``nontransitive dice'' in Wikipedia). 
}
\label{fig0}

\end{figure}

%===============  1
In a more general case, consider players $\mathfrak{P}^{\prime }$ and $%
\mathfrak{P}^{\prime \prime }$ who select options (pure strategies) C$_{1},$C%
$_{2},...$ from the respective subsets $\mathbb{S}^{\prime }$ and $\mathbb{S}%
^{\prime \prime }$ of set $\mathbb{S}$ with respective probabilities $%
g_{i}^{\prime }\geq 0$ and $g_{i}^{\prime \prime }\geq 0$. The sets $\mathbb{%
S}^{\prime }$ and $\mathbb{S}^{\prime \prime }$ may overlap. Hence, the
players' selections are represented by groups $\mathbb{G}^{\prime }$ and $%
\mathbb{G}^{\prime \prime }$. Player $\mathfrak{P}^{\prime }$ chooses group $%
\mathbb{G}^{\prime }\mathbb{\subseteq S}^{\prime }$ while player $\mathfrak{P%
}^{\prime \prime }$ chooses group $\mathbb{G}^{\prime \prime }\mathbb{%
\subseteq S}^{\prime \prime }$ but the players are not allowed to change the
subsets $\mathbb{S}^{\prime }$ and $\mathbb{S}^{\prime \prime }$. Any C$%
_{i}\in \mathbb{S}^{\prime }$\ is called \textit{available} to player $%
\mathfrak{P}^{\prime },$ while any C$_{i}\in \mathbb{G}^{\prime }$ is called 
\textit{selected} by player $\mathfrak{P}^{\prime }$. The relative strength
of pure strategies C$_{i}$ and C$_{j}$, which is called payoff in game
theory, is determined by co-ranking $\rho ($C$_{i},$C$_{j})$. This defines a
general zero-sum game for two players, $\mathfrak{P}^{\prime }$ and $%
\mathfrak{P}^{\prime \prime },$ while the groups $\mathbb{G}^{\prime }$ and $%
\mathbb{G}^{\prime \prime }$ represent mixed strategies of the players. It
is easy to see that the overall payoff of the game is determined by 
\begin{equation}
\bar{\rho}(\mathbb{G}^{\prime },\mathbb{G}^{\prime \prime })=\sum_{\text{C}%
_{i}\in \mathbb{G}^{\prime }}\sum_{\text{C}_{j}\in \mathbb{G}^{\prime \prime
}}g^{\prime }(\text{C}_{i})g^{\prime \prime }(\text{C}_{j})\rho (\text{C}_{i}%
\mathbf{,}\text{C}_{j})
\end{equation}%

Here, the average co-ranking $\bar{\rho}(\mathbb{G}^{\prime },\mathbb{G}%
^{\prime \prime })$ is defined by Equation~(\ref{r_G1G2}) but the total weights are
taken $G^{\prime }=G^{\prime \prime }=1$ since $g_{i}^{\prime }$ and $%
g_{j}^{\prime \prime }$ are interpreted as probabilities. If $\mathbb{S}%
^{\prime }$ and $\mathbb{S}^{\prime \prime }$ are distinct, then the rules
of the game might define only preferences between elements from different
sets $\mathbb{S}^{\prime }$ and $\mathbb{S}^{\prime \prime }$ but not within
each set. If these preferences can be extended transitively to all possible
pairs from $\mathbb{S}$, then the rules of the game are seen as being
transitive (and are intransitive otherwise).

The mixed strategies of this game are known to possess Nash equilibrium \cite%
{Nash1950}, where change of the mixed strategy by each player does not
increase his overall payoff, assuming that mixed strategies of the remaining
player stay the same. This condition can be expressed in terms of
conditional rankings defined by Equation~(\ref{r_G}):

\begin{proposition}
(Nash {\em \cite{Nash1950}}) Nash equilibrium is achieved when and only when all
options in a\ mixed strategy selected by each player have maximal (and the
same within each mixed strategy) ranking conditioned on the mixed strategy
of the opposition: 
\begin{equation}
\begin{array}{cc}
\bar{\rho}(\text{C}_{i},\mathbb{G}^{\prime \prime })=\bar{\rho}_{\max
}^{\prime },\text{ } & \text{if C}_{i}\in \mathbb{G}^{\prime } \\ 
\bar{\rho}(\text{C}_{i},\mathbb{G}^{\prime \prime })\leq \bar{\rho}_{\max
}^{\prime }, & \text{if C}_{i}\notin \mathbb{G}^{\prime }%
\end{array}%
\end{equation}
\end{proposition}

Indeed, if there was A$\in \mathbb{G}^{\prime }$ with $\bar{\rho}($A$,%
\mathbb{G}^{\prime \prime })<\bar{\rho}_{\max }^{\prime }$ then player $%
\mathfrak{P}^{\prime }$ can improve his overall payoff by setting $g^{\prime
}($A$)=0$ and eliminating A from $\mathbb{G}^{\prime }$. In this
proposition, $\mathfrak{P}^{\prime }$ is understood as any of the two
players and $\mathfrak{P}^{\prime \prime }$ represents his opposition. Note
that, generally, $\bar{\rho}_{\max }^{\prime \prime }\neq \bar{\rho}_{\max
}^{\prime }$, where%
\begin{equation}
\bar{\rho}_{\max }^{\prime }=\underset{\text{C}_{i}\in \mathbb{S}^{\prime }}{%
\max }\left( \bar{\rho}(\text{C}_{i},\mathbb{G}^{\prime \prime })\right) ,\
\ \ \bar{\rho}_{\max }^{\prime \prime }=\underset{\text{C}_{j}\in \mathbb{S}%
^{\prime \prime }}{\max }\left( \bar{\rho}(\text{C}_{j},\mathbb{G}^{\prime
})\right)
\end{equation}%

If the preferences between strategies and the associated co-rankings are absolutely transitive, then Nash equilibrium
is achieved when each player selects the option(s) with the highest absolute
ranking in the set available to the player. Transitive games are relatively
simple and are not particularly interesting. Finding Nash equilibrium in
case of a game with intransitive preferences can be more complicated.

For example, in the rock-paper-scissors game with all options available to
all players (\textit{i.e.}, \mbox{$\mathbb{S}^{\prime}=\mathbb{S}^{\prime\prime}=\{$R,P,S$\}$}%
), the Nash equilibrium is specified by 
\begin{equation}
\underset{\mathbb{G}^{\prime}}{\underbrace{g^{\prime}(\text{R})=g^{\prime }(%
\text{P})=g^{\prime}(\text{S})}}=\underset{\mathbb{G}^{\prime\prime }}{%
\underbrace{g^{\prime\prime}(\text{R})=g^{\prime\prime}(\text{P}%
)=g^{\prime\prime}(\text{S})}}=\frac{1}{3}  \label{GT_RPS1}
\end{equation}
with all conditional rankings being the same 
\begin{equation}
\rho_{_{\mathbb{G}^{\prime\prime}}}(\text{R})=\rho_{_{\mathbb{G}^{\prime
\prime}}}(\text{P})=\rho_{_{\mathbb{G}^{\prime\prime}}}(\text{S})=\rho_{_{%
\mathbb{G}^{\prime}}}(\text{R})=\rho_{_{\mathbb{G}^{\prime}}}(\text{P}%
)=\rho_{_{\mathbb{G}^{\prime}}}(\text{S})=0
\end{equation}
and the overall payoff of 
\begin{equation}
\bar{\rho}(\mathbb{G}^{\prime},\mathbb{G}^{\prime\prime})=0
\end{equation}

If player $\mathfrak{P}^{\prime}$ alters his mixed strategy $\mathbb{G}%
^{\prime}$ (while $\mathbb{G}^{\prime\prime}$ does not change), the overall
payoff of the game remains the same. However, any strategy of $\mathfrak{P}%
^{\prime}$ that is different from $\mathbb{G}^{\prime}$ specified by Equation~(\ref%
{GT_RPS1}) can be exploited by player $\mathfrak{P}^{\prime\prime}$ to get a
better payoff for $\mathfrak{P}^{\prime\prime}$.

\subsection{Games with Potential Intransitivity}

Some games have rules that do not explicitly stipulate intransitive
relations but allow for optimal intransitive strategies. Consider a game
where two dice are thrown and the one which shows a greater number wins---there is nothing explicitly intransitive in these rules (and an example of a
transitive set of dice can be easily suggested). However, the dice shown in
Figure \ref{fig0}b are clearly intransitive. In accordance with the
terminology used in this work, we call these games potentially intransitive.
Determining existence of intransitive optimal strategies in a particular
potentially intransitive game can be complicated.

It has been noticed that an ordinary cat tends to prefer fish to meat, meat
to milk, and milk to fish \cite{Qcat2009M}. Makowski and Piotrowski \cite%
{Icat2005PM,Qcat2009M,Ent2015MPS} considered a number of models that can
explain these intransitive preferences by the perfectly rational need of
balancing cat's diet. These models include a classical cat \cite{Icat2005PM}
and a quantum cat \cite{Qcat2009M}, and are sufficiently general to be
applied to other problems such as adversarial/cooperative balanced food
games \cite{Ent2015MPS} or to choosing candidates in elections \cite%
{Qelec2011MP} (quantum preferences are briefly discussed in Appendix \ref
{AQP}). In these models, a cat is offered three types of food in pairs with
some probabilities and is forced to chose between them. The cat strives to
achieve a perfect balanced diet and can follow various strategies
(transitive or intransitive), while intransitivity is not enforced in any
way on the cat by the game rules. The main conclusion drawn by Makowski and
Piotrowski \cite{Ent2015MPS} is consistent with the approach taken in this work:
intransitive strategies can be not only perfectly rational but also the best
under certain conditions (while in other circumstances transitive strategies
are optimal). Choosing between transitive and intransitive strategies is no
more than one of the attributes in selecting the best tactics under given
circumstances---this choice is not linked to upholding rationality of
modern science.

\section{Example: Intransitivity of Justice\label{S2}}

The possibility of intransitivity in legal regulations has been noticed and
discussed in several publications \cite{SLR2010,Katz2014}. In this section,
we offer a different example of intransitivity and suggest some~interpretations.

Consider the following options available to people witnessing a fire:

\begin{itemize}
\item[(A)] Doing nothing; 

\item[(B)] Calling the fire service;

\item[(C)] Trying to rescue people from the fire and/or extinguish the fire
\end{itemize}

\subsection{Popular View}

One can guess that popular (and, possibly, somewhat naive) estimate for
utility of these options would~be 
\begin{equation}
r(\text{A})=-1,\;r(\text{B})=0,\;r(\text{C})=1  \label{J_u}
\end{equation}%
implying that people think of option C as having a higher value for society
than option B 
(As a proof, I \mbox{recall} my childhood experienmce: after extinguishing a faulty
camp stove engulfed by flames my \mbox{father} \mbox{became} a local celebrity among other
holidaymakers. I am sure that his treatment would be less \mbox{favirouble} if,
instead, he limited his actions to hitting a firealarm button). This
popular treatment of the choice between options A, B and C, which is shown
in Figure \ref{fig1} by dashed line is perfectly transitive and~self-consistent 
\begin{equation}
\text{A}\prec \text{B}\prec \text{C}\succ \text{A}
\end{equation}%

The corresponding conditional indicator co-ranking with respect to the
reference set of $\mathbb{G}$ $=\{$A,B,C$\}$ is determined by Equation~(\ref%
{r_R_G}) and given by%
\begin{equation}
\bar{R}(\text{A,\{A,B,C\}})=-\frac{2}{3},\;\bar{R}(\text{B,\{A,B,C\}})=0,\;%
\bar{R}(\text{C,\{A,B,C\}})=+\frac{2}{3}.
\end{equation}
\begin{figure}[H]
\centering
\includegraphics[width=12cm,page=1, clip ]{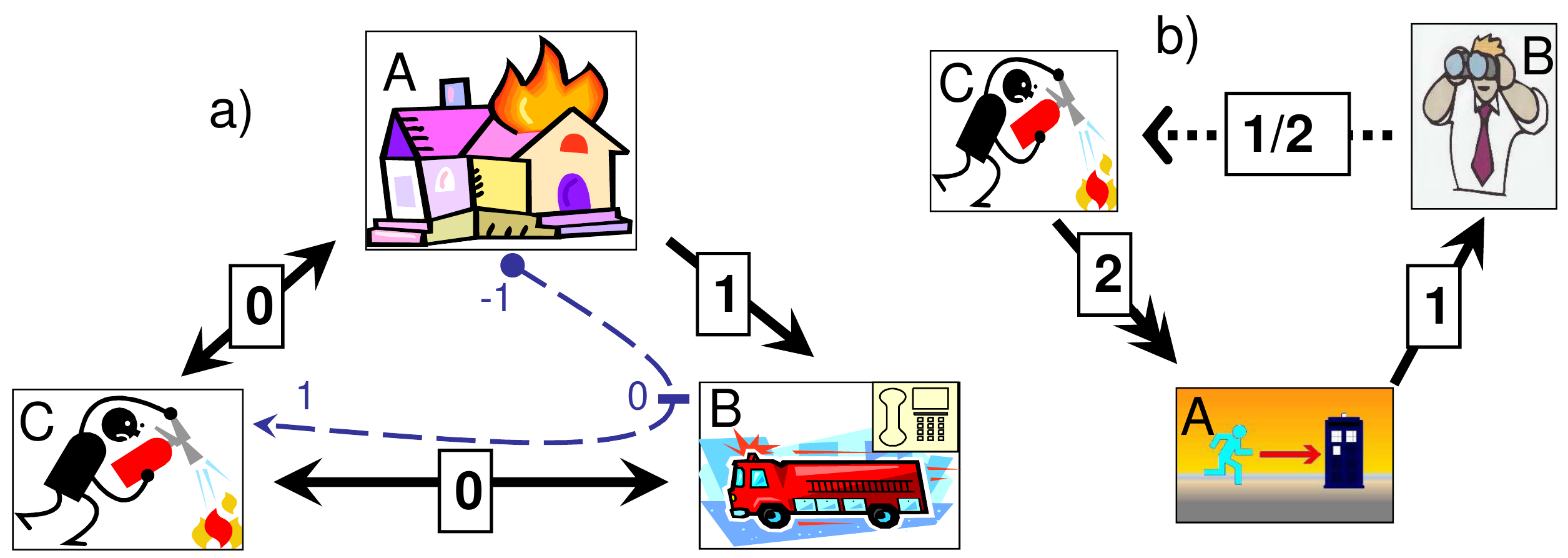}
%\caption{Comparison of choices A, B and C in case of a fire. Dashed line -- naive perspective; solid lines -- legal perspective}
\caption{Comparison of choices A, B and C in case of a fire emergency: 
(\textbf{a}) ``can do more but not less" (dashed line---naive perspective; solid arrows---legal perspective) 
and (\textbf{b})~``manager's choice" (arrows indicate the regulation requirements).
}
\label{fig1}

\end{figure}
\vspace{-24pt}
%===============  2
\subsection{Treatment by Law: Can Do More but not Less}

Let us examine how the same choice, when made by an employee in the case of
a fire in a work environment, is treated by law. The law expects that the
employee either must rush to call the fire service or can do more by trying
to extinguish the fire. Hence A$\prec$B---the employee must chose B over A,
if only two options A and B are available to him. The employee is free to
chose between B and C (\textit{i.e.}, B$\sim$C). The law, however, does not demand
that the employee risks his life if the choice is to be made between A and C
({\em i.e.}, C$\sim $A). One can see that this treatment of the options is intransitive 
\begin{equation}
\text{A}\prec\text{B}\sim\text{C}\sim\text{A}  \label{J_trip}
\end{equation}
and corresponds to the following indicator co-ranking%
\begin{equation}
R(\text{A,B})=-1,\;\;R(\text{B,C})=0,\;\;R(\text{C,A})=0  \label{J_R}
\end{equation}
which is illustrated in Figure \ref{fig1}a. Note that the intransitivity of
Equation~(\ref{J_trip}) is weak---it violates only transitivity of equivalence. If
the choice is to be made between three options, then B and C are legal while
A is not. This implies the following ordering $\left[ \{\text{B,C}\},\text{A}%
\right] $ (here, the square brackets denote ordered~sets).

The indicator co-ranking conditioned on the reference set of $\mathbb{G}$ $%
=\{$A,B,C$\}$ 
\begin{equation}
\bar{R}(\text{A,\{A,B,C\}})=-\frac{1}{3},\;\bar{R}(\text{C,\{A,B,C\}})=0,\;%
\bar{R}(\text{B,\{A,B,C\}})=+\frac{1}{3},  \label{J_R_ABC}
\end{equation}%
indicates something that we might have guessed already: legally, B is the
safest option. Note that the legal Equation~(\ref{J_R_ABC}) and common Equation~(\ref{J_u})
systems of values may differ. The law prefers option B, tolerates option C
and objects to option A so that the three options, when are ordered
according to likely legal advice, are listed as [B,C,A] (although, as noted
above, options B and C are legal while A is not). This is generally correct:
in most cases, the society would benefit if the employee calls fire fighting
professionals instead of undertaking a heroic effort himself. Fire safety
manuals often instruct employees to call the fire service before trying to
do anything else.

When A is selected out of \{A,B,C\}, which is illegal, intransitivity allows
for a line of defence based on making two legal selections instead of one
illegal. The employee selected C out of \{A,B,C\} first but when he
approached the fire (and B was no longer available) he understood that C is
dangerous or impossible and selected A out of \{A,C\}. This line of defence
is not unreasonable, provided the employee can demonstrate that the two
selections were indeed separated in time and space. Choosing A out of
\{A,C\} is not the same as choosing A out of \{A,B,C\}.

\subsection{Strict Intransitivity in Manager's Choice}

Consider a safety regulation that instructs an industrial site manager how
to act in case of an~emergency.

%TCIMACRO{%
%\TeXButton{TEX ===========reg}{\begin{enumerate}[I.]
%\item \textbf{Leadership:} if the manager is on site, he/she is expected to
%lead and organise the site personnel, deploying staff as necessary to
%actively contain or liquidate the cause of emergency.\vspace{-0.3cm}
%
%\item \textbf{Safety:} \vspace{-0.3cm}
%
%
%\begin{enumerate}[a)]
%\item the manager and personnel stay on site during emergency if there is no
%immediate danger to personnel but \vspace{-0.15cm}
%
%\item personnel evacuation must be promptly enacted whenever there is a
%significant danger to personnel.
%\end{enumerate}
%\end{enumerate}
%}}%
%BeginExpansion
\begin{enumerate}
\item[I.] \textbf{Leadership:} if the manager is on site, he/she is expected to
lead and organise the site personnel, deploying staff as necessary to
actively contain or liquidate the cause of emergency.

\item[II.] \textbf{Safety:}

\begin{enumerate}
\item[(a)] the manager and personnel stay on site during emergency if there is no
immediate danger to personnel but

\item[(b)] personnel evacuation must be promptly enacted whenever there is a
significant danger to personnel.
\end{enumerate}
\end{enumerate}
%
%EndExpansion

This regulation seems perfectly reasonable but, in fact, it is prone to
intransitivity. Consider the following options that the site manager can
undertake in case of fire:

\begin{itemize}
\item[(A)] Evacuating personnel and abandoning the site; 

\item[(B)] Organising personnel to monitor the situation on site;

\item[(C)] Organising personnel to contain and extinguish the fire.
\end{itemize}

The regulation (clause II-a) clearly prescribes B out of \{A,B\}, since
monitoring fire is safe and does not endanger personnel. Clause I explicitly
requires selecting C out of \{B,C\}. Combating fire, however, becomes
dangerous for the site personnel, triggering clause II-b: the manager must
select A out of $\{$C,A$\}$. This appears to be a case of strict
intransitivity (see Figure \ref{fig1}b)%
\begin{equation}
\text{A}\prec \text{B}\prec \text{C}\prec \text{A}  \label{J_int}
\end{equation}%

Practically, intransitivity of available options is likely to be sufficient
to create reasonable doubts about incorrectness of the manager's choice. The
question about the best course of action prescribed by the manual
nevertheless remains and needs further analysis. Note that assigning utility
to options A, B and C is impossible, as this would be inconsistent with
intransitivity of the choices Equation~(\ref{J_int}). A co-ranking, however, can
still be deployed. The co-rankings are specified in accordance with
perceived importance the corresponding clauses: I)$\ \rho ($C$,$B$)=1/2,\ $%
II-a) $\rho ($B$,$A$)=1$ and \ II-b) $\rho ($A$,$C$)=2$. Here we take into
account that\ the safety clause (II) has a stronger formulation than the
leadership clause (I) and that the safety of personnel (clause II-b) has the
highest priority in clause II. These priorities are illustrated in Figure %
\ref{fig1}b. The corresponding conditional utilities of these three options
are given by 
\begin{equation}
\bar{\rho}(\text{C,\{A,B,C\}})=-\frac{3}{6},\;\bar{\rho}(\text{B,\{A,B,C\}}%
)=+\frac{1}{6},\;\bar{\rho}(\text{A,\{A,B,C\}})=+\frac{2}{6}
\end{equation}%

When choosing from $\{$A,B,C$\},$ the manager can select the best option out
of this set (which is A) or eliminate the worst option (which is C) and
reassess conditional utilities of the remaining set. In the latter case the
process of elimination continues until only one element is left, which is to
be selected. For the present three options, the \textquotedblleft selecting
the best\textquotedblright\ method yields A while the \textquotedblleft
eliminating the worst\textquotedblright\ method yields B at the end. Indeed,
the best option to be selected from $\left\{ \text{A,B}\right\} $ is B,
which is different from A---the best option selected from $\{$A,B,C$\}$.
This is not a fallacy: our choice simply depends on available information
and the presence of option C adds information about fire danger and affects
our evaluation of the other options. It needs to be understood that
different methods do provide the best choices but in different
circumstances. The \textquotedblleft eliminating the
worst\textquotedblright\ method corresponds to the case when the fire danger
is completely eliminated with elimination of option C (for example when C
requires moving personnel to a different location). In a more mixed
situation where the distinction between different options is more blurred,
option C remains a potential danger to personnel even if it is not
specifically selected. For example, someone might try to extinguish the fire
or perhaps there is a probability that switching to option C will be forced
by developing circumstances. According to Propositions \ref{P1} and~\ref{P1a}, dependence
of values of the options on perspective (\textit{i.e.}, dependence of conditional
rankings on the reference set) is a property of intransitive systems. As
considered in Section \ref{S4}, intransitivity is common when multiple
selection criteria (in this case clauses I and II) are in place.

\section{Potential Intransitivity of the Original Lotka-Volterra Model\label%
{S3}}

In ecology and biological population studies, intransitivity and locality of
competition have been recognised as the key factors that maintain
biodiversity \cite{LV2002,LV2007} 
(It is interesting that the same factors---\mbox{intransitivity} and localisation---have been nominated as conditions for complex behavior in generic
competitive systems \cite{K-PS2012,K-PT2013,K_Ent2014a}). A few concepts
that are commonly used in this field indicate the inherent presence of
intransitivity. For example, some invasions and occupying niches are signs
of intransitivity: species in ecological systems are competitive against
existing competitors but may be vulnerable against unfamiliar threats.
Hence, this competitiveness is not absolute---a new invader, which does
not necessarily hold the highest competitive rank in its home environment,
might be very successful, this clearly demonstrates vulnerability of the
system. After the Isthmus of Panama connecting two Americas was formed,
North American fauna was more successful in invading the other continent and
diversifying there \cite{GAI1982}. Hence, it is reasonable to conclude that
North American fauna was more competitive than the fauna of South America 
\cite{K-PT2013}. This statement refers to higher absolute competitiveness
and, hence, is transitive. However, a more detailed consideration reveals
that some of the South American species (such as armadillos and sloths,
which generally do not seem particularly competitive) were quite successful
in invading North America and occupying niches there. This indicates
intransitivity of competitiveness. Indeed, while being highly competitive in
general, North American fauna was not resistant with respect to invasion of
sloths. Propositions \ref{P1}, \ref{P1a}, \ref{PA1a} and \ref{PA1aa} explicitly link relativity of
competitiveness (preferences) to intransitivity.

While the existence of intransitivity in competitive Lotka--Volterra models
that generalise the original version for multiple species is well-known \cite%
{LV2014}, the cyclic behaviour of the original version hints at possible
intransitivity (according to our terminology, it can be called potentially
intransitive).\ Here, we show that such intransitivity is indeed implicitly
present in the original version of the Lotka--Volterra model \cite{LV1920,
LV1939}, which is given by the following system%
\begin{equation}
\frac{dN_{_{\text{R}}}}{dt}=a_{_{\text{R}}}N_{_{\text{R}}}-bN_{_{\text{R}%
}}N_{_{\text{F}}}
\end{equation}%
\begin{equation}
\frac{dN_{_{\text{F}}}}{dt}=b^{\prime}N_{_{\text{R}}}N_{_{\text{F}}}-a_{_{%
\text{F}}}N_{_{\text{F}}}
\end{equation}
where $N_{_{\text{R}}}$ represents the population of rabbits (R) and $N_{_{%
\text{F}}}$ is the population of foxes (F). The model coefficients $a_{_{%
\text{R}}}$, $a_{_{\text{R}}}$, $b$ and $b^{\prime}$ are assumed to be
positive indicating that the foxes win resources from the rabbits due to
presence of the term $bN_{_{\text{R}}}N_{_{\text{F}}}$ ($b$ and $b^{\prime}$
are generally different since $N_{_{\text{F}}}$ and $N_{_{\text{R}}}$ are
measured in different resource units: in foxes and in rabbits). The
intransitivity of this model is not apparent since this model does not
explicitly mention the environment (E), but the presence of the environment
is important. Indeed if $N_{_{\text{R}}}=1$ and $N_{_{\text{F}}}=0$ at $t=0$%
, then $N_{_{\text{R}}}\rightarrow\infty$ as $t\rightarrow\infty$, which is
physically impossible. In real life, this growth of rabbit population would
be terminated by environmental restrictions: $a_{_{\text{R}}}=a_{_{\text{R}%
}}(N_{_{\text{E}}})\rightarrow0$ and $N_{_{\text{E}}}\rightarrow0$ as $N_{_{%
\text{R}}}\rightarrow(N_{_{\text{R}}})_{\text{max}}$. Here, $N_{_{\text{E}}}$
represents free environmental resource, which is approximately treated as
constant in the original Lotka-Volterra model. Figure \ref{fig2} illustrates
the intransitivity of this case: rabbits win from the environment, while
foxes lose to the environment.
\begin{figure}[H]
\centering
\includegraphics[width=8cm,page=2, clip ]{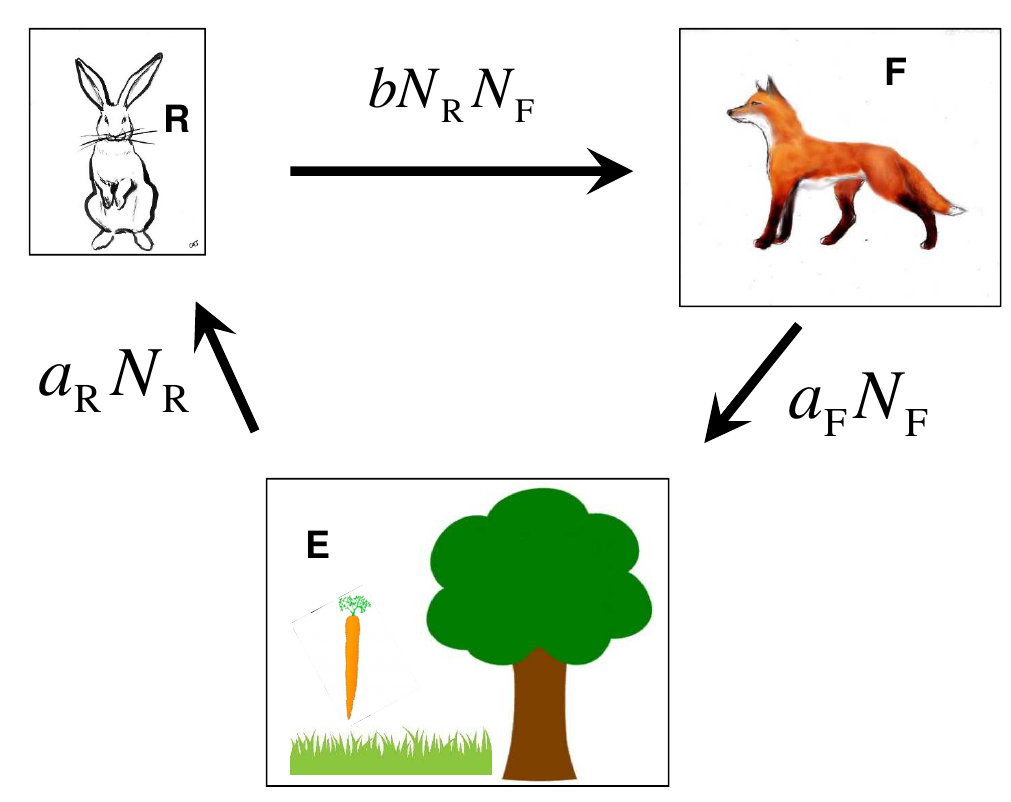}
\caption{Presence of intransitivity in the original Lotka--Volterra model.}
\label{fig2}

\end{figure}
\vspace{-24pt}

%===============  3

\section{Fractional Ranking\label{S4}}

Different options are often characterised by a set of criteria, say, indexed
by $\alpha=1,...,K$.\ A ranking is then specified for each of these criteria 
$r^{(\alpha)}($A$).$ The notations which are used here are similar to those
used in the previous sections. For example all of the following statements 
\begin{equation}
\text{A}\prec^{(\alpha)}\text{B},\text{\ \ }\rho^{^{(\alpha)}}(\text{A,B}%
)<0,\;\;R^{^{(\alpha)}}(\text{A,B})=\func{sign}\left( \rho ^{^{(\alpha)}}(%
\text{A,B})\right) =-1,\;\;r^{^{(\alpha)}}(\text{A})<r^{^{(\alpha)}}(\text{B}%
)
\end{equation}
indicate that B is preferred to A with respect to criterion $\alpha.$ The
word \textit{fractional} (or \textit{partial}) is used to indicate that the
comparison is performed only with respect to a single criterion: $%
r^{^{(\alpha)}}$ is a \textit{fractional ranking}, $\rho^{^{(\alpha)}}$ is a 
\textit{fractional co-ranking}, \textit{etc.} 
(Another equivalent term is ``partial''.
This term is commonly used in calculus but this might be in conflict with
``partial orders'', where
``partial'' is interpreted as
``incomplete''---in partial orders,
preferences are not necessarily defined for all pairs of elements). Of
course, fractional ranking $r^{^{(\alpha)}}$ exists only if elements can be
transitively ordered with respect to criterion $\alpha.$ No criterion alone
determines the \textit{overall preference}: for example, we might have $%
r^{(1)}($A$)>r^{(1)}($B$)$ but $r^{(2)}($A$)<r^{(2)}($B$)$. The fractional
ranking can be either graded or sharp; the former can be called \textit{%
fractional utility}. In this section, we deploy the results of the social
choice theory, where different criteria represent judgment of different
individuals.

\subsection{Commensurable Fractional Rankings}

When fractional rankings are graded they are interpreted as fractional
utilities, \textit{i.e.}, they reflect the degrees of satisfaction with respect to
specific criteria. It is reasonable to expect that these $K$ degrees of
satisfaction can be compared to each other; this means that the fractional
utilities can be rescaled to be measured in the same common ``units of
satisfaction''. \ Hence, when different criteria are commensurable, absolute
utility can be easily introduced according to equations 
\begin{equation}
r(\text{A})=\frac{1}{W}\sum_{\alpha=1}^{K}w^{^{(\alpha)}}r^{^{(\alpha)}}(%
\text{A}),\;\;\;\;W=\sum_{\alpha=1}^{K}w^{^{(\alpha)}}  \label{p_u}
\end{equation}
where the criterion weights $w^{(\alpha)}$ are used to rescale units as
needed. In principle, it is possible to consider cases when fractional
utilities are combined in a non-linear manner but this would not change the
main conclusion of this subsection:

\begin{proposition}
\label{P2}Commensurable fractional utilities correspond to an overall
preference that can be expressed by the absolute utility $r($...$)$ and,
hence, is transitive.
\end{proposition}

\subsection{Commensurable Fractional Co-Rankings}

There are two cases when utility must be replaced by co-ranking: (1) absolute
fractional rankings do not exist or are unknown and (2) absolute fractional
rankings exist but are incommensurable 
(That is we can compare the magnitudes of partial improvements, say, $%
r^{^{(1)}}($A$)-r^{^{(1)}}($B$)$ and $r^{^{(2)}}($A$)-r^{^{(2)}}($B$)$ but
cannot compare the absolute magnitude, say, $r^{^{(2)}}($A$)$ and $%
r^{^{(2)}}($A$)$. This is a common situation since, as discussed in the next
chapter, human preferences are inherently relativistic). The fractional
preferences can always be expressed by fractional co-rankings, which are
treated in this subsection as graded and commensurable. 
(The relative character of real-world preferences, which is reflected by
co-rankings, is discussed further in the paper. The case of completely
incomparable partial preferences is considered in the following subsection).
The overall co-ranking is expressed in terms of the fractional co-rankings
by the equation 
\begin{equation}
\rho(\text{A,B})=\frac{1}{W}\sum_{\alpha=1}^{K}w^{^{(\alpha)}}\rho
^{^{(\alpha)}}(\text{A,B})  \label{p_ro}
\end{equation}
where the weights $w^{^{(\alpha)}}$ represent the scaling coefficients,
whose physical meaning is similar to that of the weights in Equation~(\ref%
{p_u}). Depending on the functional form of the fractional co-rankings,
three cases are possible (1) fractional and overall co-rankings are
transitive (in this case the fractional and overall utilities exist); (2)
fractional co-rankings are transitive but the overall co-ranking is
intransitive and \mbox{(3)~all~co-rankings} are intransitive. As discussed further
in the paper,\ the second case is common when fractional co-rankings have
non-linear functional forms, which can appear due to imperfect
discrimination or for other reasons.

\subsection{Incommensurable Fractional Preferences}

In this subsection, the case of incommensurable fractional preferences is
considered (irrespective of transitivity of fractional preferences). For
example, it could be the case that $\rho^{(1)}($A$,$B$)$ and $\rho^{(2)}($A$%
, $B$)$ can not be rescaled to produce commensurable quantities. Grading of
fractional rankings or co-rankings becomes useless if different gradings are
incommensurable. The information that can be used in this case is limited to
(1) sharp fractional ranking, if the fractional preferences are transitive or
(2) indicator co-ranking (or sharp fractional co-ranking), if the fractional
preferences are intransitive. The first case in considered first. If sharp
(or incommensurable) fractional absolute rankings are strict, they represent
an \textit{ordering} as discussed in Section \ref{S1b}.

Arrow's theorem \cite{Arrow} states that fractional orderings cannot be
converted into overall ordering in a consistent manner. The conditions of
being consistent are stipulated in the formulation of the theorem, which is
given below and uses notations of the present work.

\begin{theorem}
(Arrow {\em \cite{Arrow}}) For more than two elements, a set of $K$ fractional
orderings cannot be universally converted into an overall ordering in a way
that is:

%TCIMACRO{%
%\TeXButton{TeX =========== Arr}{\begin{enumerate}[a)]
%\item Non-trivial (non-dictatorial): absolute ranking does not simply
%replicate one of the fractional rankings: $r(...)\nsim$ $r^{(\alpha)}(...)$
%for all $\alpha;$\vspace{-0.25cm}
%
%\item Pairwise independent: preference between any two elements does not
%depend on fractional rankings of the other elements, \textit{i.e.}, $R($A$,$B$)$
%depends only on all $R^{(\alpha)}($A$,$B$),$ $\alpha=1,...,K;$\vspace{-0.25cm}
%
%\item Pareto-efficient: A$\succ$B when $r^{(\alpha)}($A$)>r^{(\alpha)}($B$)$
%for all $\alpha.$
%\end{enumerate}}}%
%BeginExpansion
\begin{enumerate}
\item[(a)] Non-trivial (non-dictatorial): absolute ranking does not simply
replicate one of the fractional rankings: $r(...)\nsim$ $r^{(\alpha)}(...)$
for all $\alpha;$

\item[(b)] Pairwise independent: preference between any two elements does not
depend on fractional rankings of the other elements, \textit{i.e.}, $R($A$,$B$)$
depends only on all $R^{(\alpha)}($A$,$B$),$ $\alpha=1,...,K;$

\item[(c)] Pareto-efficient: A$\succ$B when $r^{(\alpha)}($A$)>r^{(\alpha)}($B$)$
for all $\alpha.$
\end{enumerate}%
%EndExpansion
\end{theorem}

The underlying reason for the impossibility of Arrow-consistent conversion
(\textit{i.e.}, complying with the three conditions of the Arrow's theorem) of
fractional-criteria orderings into overall ordering is intransitivity. This
point is discussed further below with the use of the following proposition:

\begin{proposition}
\label{P3}Strict fractional preferences (represented by fractional rankings
if transitive or by fractional co-rankings if intransitive) can always be
converted into an overall strict preference in an Arrow-consistent way,
which is (1) non-trivial (for $K>2$), (2) pairwise independent and (3)~Pareto-efficient.
\end{proposition}

The proof is straight-forward: the overall co-ranking defined by 
\begin{equation}
\rho(\text{A,B})=\frac{1}{W}\sum_{\alpha=1}^{K}w^{^{(\alpha)}}R^{^{(%
\alpha)}}(\text{A,B})  \label{p_R}
\end{equation}
is non-trivial, pairwise independent and Pareto-efficient. Indeed, (1) $\func{%
sign}(\rho($A,B$))\nsim$ $R^{(\alpha)}($A,B$)$ for any $\alpha$ (we assume $%
K>2$), (2) the formula for $\rho($A,B$)$ does not involve any characteristics
of any third element (say C) and (3) $\rho($A,B$)=1$ when all $R^{(\alpha)}($%
A,B$)=1$. Here we put $w^{(\alpha)}=1+\varepsilon^{(\alpha)},$ where $%
\varepsilon^{(1)},...,\varepsilon^{(K)}$ are small random values, which
ensure that $\rho($A,B$)=0$ only when A$=$B.

The existence of fractional rankings is not essential for this proposition
and Equation~(\ref{p_R}) can be used when the fractional preferences are
intransitive. If fractional co-rankings are transitive and fractional
rankings exist, this, as shown below, cannot ensure transitivity of the
overall co-ranking by Equation~(\ref{p_R}).

\begin{proposition}
\label{P4} Any Arrow-consistent conversion of fractional orderings into an
overall strict preference is potentially intransitive.
\end{proposition}

Indeed, if the overall Arrow-consistent strict preference of Proposition \ref%
{P3} was necessarily transitive, this preference could be always converted
into an Arrow-consistent ordering, which is impossible according to Arrow's
theorem. Here we refer only to potential intransitivity since the overall
preference might be currently transitive under some specific conditions. For
example, Pareto efficiency requires that A$\succ$B$\succ$C$\prec$A when $%
r^{(\alpha)}($A$)>r^{(\alpha)}($B$)>r^{(\alpha)}($C$)$ for all $\alpha$.

We conclude that intransitivity necessarily appears in overall preference
rules that are consistently derived from a set of incommensurable fractional
criteria. Despite intransitive co-ranking, the elements can still be
transitively ordered by conditional ranking with respect to a selected
reference set. This transitive ordering, however, is in conflict with the
second condition of Arrow's theorem (pairwise independence) due to
dependence of conditional ranking on the reference set in intransitive
systems (see Propositions \ref{P1}, \ref{P1a}, \ref{PA1a} and \ref{PA1aa}).\ The conditions of Arrow theorem
require that the overall preferences cannot have absolutely transitive
co-rankings and, practically, violate at least one of the properties:
transitivity or pairwise independence.

\section{The Subscription Example\label{S5}}

This section is dedicated to detailed analysis of the example ``\textit{The
Economist's} subscription'' used by Dan Ariely \cite{Dan2008} as an
excellent demonstration of the relativity of human preferences: in the real
world, we can express our preferences only in comparison with the other
options available. Classic economic theory sees this kind of behaviour as
irrational. While many examples from Ariely's book ``Predictably
irrational'' \cite{Dan2008} are indeed linked to irrationality of human
behaviour, relativity of our preferences in general and in choosing the 
\textit{The} \textit{Economist's} subscription in particular is perfectly
rational.

Sometimes we can make a choice in absolute terms without resorting to
relative comparison. For example, considering \textit{The Economist's}
subscription, I would reject an annual print subscription for \textit{The
Economist }priced at \$1000 and agree to have this subscription for \$10
without much thought or any further comparisons. One can see that these
statements based on the perceived absolute values of the subscription and
money are quite approximate and, probably, not suitable for the real world.
Realistically, I do not know offhand whether I want \textit{The Economist }%
print\textit{\ }subscription priced at \$120. To make this judgment, I need
to estimate the value of money in terms of utility of published media and
see if the offered price is reasonable or not. I would probably look at
subscriptions for other magazines to make up my mind. In the end, I might
decide that \textit{The Economist} provides good value for the money, or
that I can get a colourful magazine to read at much lower cost. My relative
preference is perfectly rational; in fact, it would be irrational for me to
make up my mind on the basis of the absolute value of money and the absolute
utility of enjoying \textit{The Economist,} without knowing the subscription
market and undertaking relative comparisons.
\newpage
\subsection{Ariely's Subscription Example}

Consider the following options for subscription to \textit{The Economist}

\begin{description}
\item (A) Web (W) subscription, \$60;\vspace{-0.3cm}

\item (B) Print \& Web (P+W) subscription, \$120;\vspace{-0.3cm}

\item (C) Print (P) subscription, \$120
\end{description}

The prices have been slightly adjusted from original \$59 for W and \$125
for P and P+W reported by Ariely \cite{Dan2008} to make evaluation of this
example more simple and transparent.

Ariely \cite{Dan2008} determined that when only two options, A and B, are
given, people tend to make their choices with the following frequencies: 
\begin{equation}
\text{(A) 68\%,\ \ \ (B) 32\%}
\end{equation}

However, when all three options are available, preferences become very
different 
\begin{equation}
\text{(A) 16\%,\ \ \ (B) 84\%,\ \ \ (C) 0\%}
\end{equation}

Although option C is not chosen, it affects the choice between options A and
B. While Ariely\ believes that this is irrational, we argue here that this
is a perfectly rational choice conducted in line with a reasonable
relativistic analysis of the offers. Note that the dependence of the choice
on C cannot be explained within the conventional framework of absolute
preferences (\textit{i.e.}, by any set of absolute utilities assigned to options A, B
and C).

\subsection{Evaluating Co-Rankings}

Since no additional information is given (for example, we have no idea about
realistic costs of the options offered), the choice between the subscription
options can be made rationally only on the basis of comparing these options
to each other. We thus compare the options with respect to the two
fractional criteria, prices $p$ and values $v,$ using three different
grades: the same \textquotedblleft $\sim $\textquotedblright\ \ ($\rho =0$),
better \textquotedblleft $\succ $\textquotedblright\ ($\rho =1$) and clearly
better \textquotedblleft $\succ \succ $\textquotedblright\ ($\rho =2$). The
following estimate of our relative preferences 
\begin{equation}
\$60\succ \succ \$120,\;\;\;\text{P}\succ \text{W, \ \ \ P+W}\succ \succ 
\text{W, \ \ \ P+W}\succ \succ \text{P }  \label{A_pref1}
\end{equation}%
seems reasonable. Note that the first relation is not a mistake: the symbol
\textquotedblleft $\succ $\textquotedblright\ means \textquotedblleft
preferred to\textquotedblright\ and not \textquotedblleft greater
than\textquotedblright . Obviously, we prefer a lower price and $p=\$60$ is
clearly better than $p=\$120.$ Figure \ref{fig3}a shows that our assessment $%
\$60\succ \succ \$120$ corresponds to the price utility 
\begin{equation}
r^{(p)}(\text{A})=3,\;\;r^{(p)}(\text{B})=r^{(p)}(\text{C})=1  \label{A_rp}
\end{equation}%
which can be evaluated from the equation 
\begin{equation}
r^{(p)}=\frac{\$150-p}{\$30}.  \label{A_up}
\end{equation}%

Higher utility corresponds to lower price and any price $\geq \$150$ is not
seen as reasonable.
\begin{figure}[H]
\centering
\includegraphics[width=14cm,page=3, clip ]{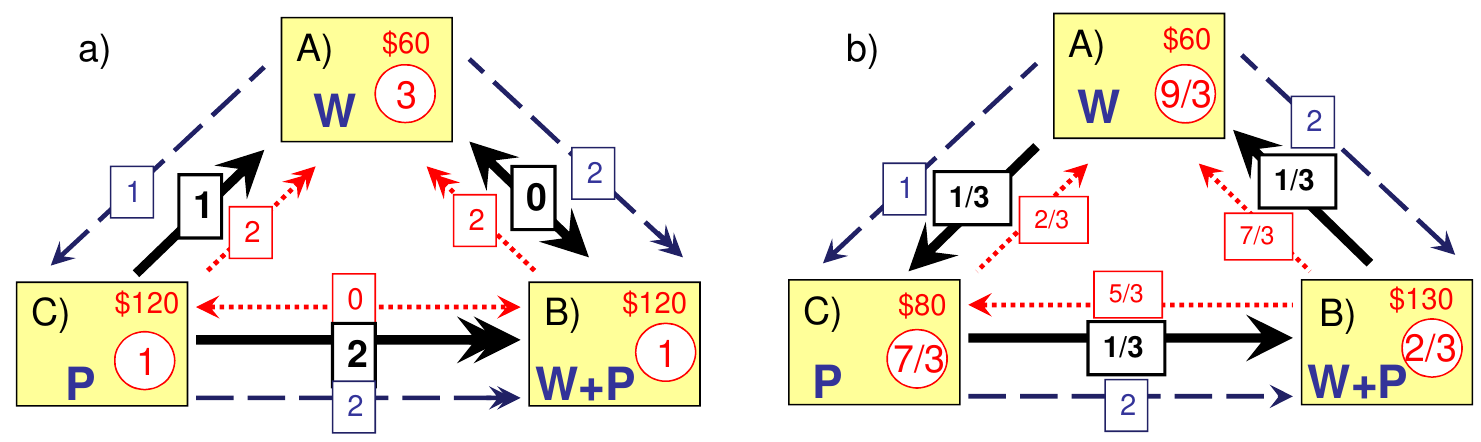}
\caption{Choices A, B and C for  {\it The Economist} subscription \cite{Dan2008}. 
  Dashed blue arrows---value preferences, dotted red arrows---price preference, 
  thick black arrows---overall preferences (shown as the sum of the fractional preferences); the price utility shown in the circles. 
  (\textbf{a}) Original prices; (\textbf{b}) adjusted prices leading to explicit intransitivity}
\label{fig3}

\end{figure}
%===============  4
We concede that the print edition is more convenient than web edition P$%
\succ $W (reading \textit{The Economist} while sitting in an armchair and
having a cup of coffee does have some advantage), although this convenience
is not overwhelming. Having both print and web subscription is clearly
better than any of these subscriptions alone. The result of comparison is
shown in Figure \ref{fig3}a so that the overall co-ranking is given by 
\begin{equation}
2\rho(\text{A,B})=0,\;\;2\rho(\text{A,C})=1,\;\;2\rho(\text{B,C})=2,
\label{A_rho}
\end{equation}
and 
\begin{equation}
\text{A}\sim\text{B}\succ\succ\text{C}\prec\text{A}  \label{A_ABCA}
\end{equation}

As specified by Equation~(\ref{p_ro}), this co-ranking is obtained by summing
fractional co-rankings $\rho^{(v)}$ and $\rho^{(p)}$ with equal weights and $%
W=2.$ When only two options, A and B, are available the conditional ranking
with the reference set of $\mathbb{G}=\{$A$,$B$\}$ is given by 
\begin{equation}
\bar{\rho}(\text{A,\{A,B\}})=0,\;\;\bar{\rho}(\text{B,\{A,B\}})=0
\end{equation}
but when the choice is to be made by selecting from A, B and C, the
conditional ranking using the reference set of $\mathbb{G=}\{$A$,$B$,$C$\}$
becomes 
\begin{equation}
6\bar{\rho}(\text{A,\{A,B,C\}})=1,\;\;6\bar{\rho}(\text{B,\{A,B,C\}})=2,\;\;6%
\bar{\rho}(\text{C,\{A,B,C\}})=-3
\end{equation}
Here, we use Equation~(\ref{r_G}) with all weights $g_{i}$ set to unity.
Hence, we would chose B from $\{$A$,$B$,$C$\}$ but will have difficulty of
selecting between A and B from $\{$A$,$B$\}$. 
(For the sake of our argument, it is sufficient to put $\rho($A,B$)=0$ and
treat the preference between A and B \ as being close to 50\% each. Ariely 
\cite{Dan2008} indicates a marginal preference of A over B, which can be
accommodated by introducing another grade of a preference---``marginally
better'' quantified by, say, 1/3 or 1/2. \ The co-ranking $\rho($A,B$)$ is
thus redefined while the remaining co-rankings in Equation~(\ref{A_rho}) are kept the
without change. If \mbox{$2\rho($A,B$)=1/3$,} then \ A$\prec_{_{\mathbb{G}}}$B
since $2\rho_{_{\mathbb{G}}}($A$)=4/3\;$and \ $2\rho_{_{\mathbb{G}}}($B$%
)=5/3.$ If $2\rho($A,B$)=1/2$, then \ A$\sim _{_{\mathbb{G}}}$B since \mbox{$%
2\rho_{_{\mathbb{G}}}($A$)=2\rho_{_{\mathbb{G}}}($B$)=3/2.$} \ Here, $\mathbb{%
G}=$\{A,B,C\}.
The author has repeated Ariely's experiment in class of 60 students with
half of the class selecting between A and B, while the other half \
selecting between A, B and C. The results \{85\%, 15\%\} and \{35\%, 62\%,
3\%\} clearly confirm the effect discovered by Ariely, although indicate a
higher level of acceptance of electronic communications than a decade~ago).

\subsection{Potential Intransitivity of the Subscription Values}

As suggested by Propositions \ref{P1}, \ref{P1a}, \ref{PA1a} and \ref{PA1aa}, the dependence of the conditional
rankings of A and B on the reference set $\mathbb{G}$ indicates potential
intransitivity of our preference. This intransitivity is not clearly visible
since our preferences Equation~(\ref{A_ABCA}) do not form a strictly intransitive
triplet, but current intransitivity may appear when conditions are altered.
Here we consider a specific example, while a general case is treated in
Appendix \ref{AA1}. The subscription case shown in Figure \ref{fig3}a has a
potentially intransitive value co-ranking and absolutely transitive price
co-ranking. Indeed, current intransitivity can easily appear if we adjust
the prices. Figure \ref{fig3}b indicate that the prices $p_{_{\text{A}%
}}=\$60,$ $p_{_{\text{B}}}=\$130$ and $p_{_{\text{C}}}=\$80$ correspond to
the utilities of 3, 2/3 and 7/3 as specified by Equation~(\ref{A_up}). Our
assessment of the subscription values remains the same as in Figure \ref%
{fig3}a. With the new price utilities, the overall co-ranking becomes 
\begin{equation}
2\rho(\text{A,B})=-1/3,\;\;2\rho(\text{B,C})=-1/3,\;\;2\rho(\text{C,A})=-1/3,
\end{equation}
and our overall preferences given by 
\begin{equation}
\text{A}\prec\text{B}\prec\text{C}\prec\text{A}  \label{A_trip}
\end{equation}
are currently intransitive as shown in Figure \ref{fig3}b.

If we need to get rid of intransitivity, the values of subscriptions have to
be adjusted so that fractional utility $r^{(v)}$ can be introduced and then
the overall utility $r=r^{(v)}+r^{(p)}$ ensures transitivity of our
preferences. Let $r^{(v)}($A$)=1,$ $r^{(v)}($B$)=3$ and $r^{(v)}($C$)=2$.
This corresponds to replacing the last preference in Equation~(\ref{A_pref1}) by P+W$%
\succ $P. This transitive correction does not necessarily represent human
preferences better (in fact P+W$\succ$P is not accurate for me, since I
think that P+W is clearly better than P) but it removes potential
intransitivity.

\subsection{Discussion of the Choices}

Our knowledge of the subscription market and, consequently, our analysis of
the available options given above may be imperfect, but it is not
irrational. It is based on a system of values and on a systematic comparison
between these options---but why does our choice depend on the presence of
option C? This seems to be illogical. We compare A and B and, with the
information available to us, we have difficulties of making a choice between
these options. Option C provides us with additional information that makes
option B more attractive: the web subscription is given to us at no
extra-cost, while the printed version of the magazine has a high cost and,
presumably, high quality and high aesthetic value.

It can be argued that a buyer should care about the value of the product and
the price but not about getting a good deal from a seller. This could be
rational only if the buyer had a complete knowledge of the product and its
future use. In the real world, a responsible buyer checks that he is getting
a reasonable deal even if, in principle, he is prepared to pay more. A buyer
who does not want a good deal is, in fact, irrational---in the real world,
this buyer will be overcharged much too often.

An overzealous deal-seeker is, however, prone to manipulations and to buying
goods and product features that he does not needed. Therefore, the fact that
sellers can be manipulative should not be overlooked. The web subscription
is given for free in option B because its web delivery has a very low cost.
Perhaps, but there could be other reasons. For example let us assume that
the realistic pricing of subscriptions is similar to the prices shown in
Figure \ref{fig3}b. The subscription seller may then lift the price of $p_{_{%
\text{C}}}$ from $\$80$ to $\$130$ to lure his customers into subscribing
for option B. In this case the presence of C in the subscription list is not
information but disinformation. How can the buyers protect themselves
against such manipulations?

In transitive systems, preferences are absolute and independent of
perspective. Proposition \ref{P1} and \ref{P1a} show that, in intransitive systems,
preferences depend on perspective: whether A$\succ_{_{\mathbb{G}}}$B or not
depends on reference set $\mathbb{G}$ and on reference weights $g_{i}$.
Hence, a reasonable choice relies on a good selection of the perspective.
Artificially or unscrupulously selected elements may distort the picture. In
the subscription example, it might be desirable to weight the options by
their estimated market shares. In this case the seller's manipulations with
option C would not have a significant effect on our choice.

Economic theory sees the inherent relativity of our preference as being
irrational. While humans can make irrational choices at times,  relative comparisons are very common and perfectly rational despite being inherently prone to intransitivity.   
In fact, in many cases relative comparisons are the only ones that are practically possible and avoiding them would be irrational.  
Enforcing transitivity does not necessarily make our assessments or theories more accurate but it does make our choices more stable, more
immune from manipulations and easier to predict---transitive preferences are absolute and do not depend on third options.
However, as demonstrated by Ariely's subscription example, a real buyer in
the real world is likely to have a (potentially) intransitive set of
preferences.

\section{Intransitivity Due to Imperfect Discrimination\label{S6}\protect%
\nolinebreak}

Since it is often the case that exact values of the fractional utilities are
not known or, maybe known but, to some extent ignored by decision-makers, we
need to deal with approximate values of the parameters. This, as
demonstrated below, leads to intransitivity.

\subsection{Discrimination Threshold}

In the real world, preferences are typically not revealed whenever
difference in utility values are small, say, smaller than a given threshold $%
\varepsilon $. This corresponds to a \textit{coarse co-ranking} $%
\rho^{^{(\alpha)}}$ defined by 
\begin{equation}
\rho^{^{(\alpha)}}(\text{A,B})=\left\{ 
\begin{array}{cc}
r^{(\alpha)}(\text{A})-r^{(\alpha)}(\text{B}), & \text{if \ }\left|
r^{(\alpha)}(\text{A})-r^{(\alpha)}(\text{B})\right| >\varepsilon^{(\alpha )}
\\ 
0, & \text{if \ }\left| r^{(\alpha)}(\text{A})-r^{(\alpha)}(\text{B})\right|
\leq\varepsilon^{(\alpha)}%
\end{array}
\right.  \label{p_rho_c}
\end{equation}

The corresponding coarsened fractional equivalence is understood as 
\begin{equation}
\begin{array}{cc}
\text{A}\succ^{(\alpha)}\text{B}, & \text{if \ }r^{(\alpha)}(\text{A}%
)-r^{(\alpha)}(\text{B})>\varepsilon^{(\alpha)} \\ 
\text{A}\sim^{(\alpha)}\text{B}, & \text{if \ }\left| r^{(\alpha)}(\text{A}%
)-r^{(\alpha)}(\text{B})\right| \leq\varepsilon^{(\alpha)} \\ 
\text{A}\prec^{(\alpha)}\text{B}, & \text{if \ }r^{(\alpha)}(\text{B}%
)-r^{(\alpha)}(\text{A})>\varepsilon^{(\alpha)}%
\end{array}%
\end{equation}

Fractional co-rankings determine the overall preference according to Equation~(\ref%
{r_rho}). Despite the existence of commensurable fractional utilities and
the overall utility Equation~(\ref{p_u}) for the \textit{fine preference} (\textit{i.e.},
original preference with perfect discrimination), the coarse preference
specified by Equation~(\ref{p_rho_c}) is intransitive and does not have an overall
utility
(Coarsening of partial co-rankings corresponds to coarsening partial
preferences, understood according to coarsening of preferences as defined in
Appendix \ref{AA2}. The overall preferences, however, do not represent
coarsening of the original (fine) overall preferences, at least because the
former can be intransitive while the latter are transitive). The
intransitive properties of coarsening are characterised by the following
proposition due to Yew-Kwang Ng \cite{Ng1977}.

\begin{proposition}
\label{P5}(Ng {\em \cite{Ng1977}}) The overall preferences that correspond to
threshold coarsening of $K$ independent fractional utilities are

%TCIMACRO{%
%\TeXButton{TeX field ============ weak int}{\begin{enumerate}[(a)]
%
%\item weakly intransitive (existence of A$\sim$B$\sim$C$\succ$A) if $K=1,$\vspace{-0.25cm}
%
%\item semi-weakly intransitive (existence of A$\succ$B$\sim$C$\succ$A) if $K=2$ and $w^{^{(1)}}\varepsilon^{^{(1)}}=w^{^{(2)}}\varepsilon^{^{(2)}},$\vspace{-0.25cm}
%
%\item strictly intransitive (existence of A$\succ$B$\succ$C$\succ$A) if $K=2$
%and $w^{^{(1)}}\varepsilon^{^{(1)}}\neq w^{^{(2)}}\varepsilon^{^{(2)}}$ or $K\geq3$.
%\end{enumerate}}}%
%BeginExpansion
\begin{enumerate}

\item[(a)] weakly intransitive (existence of A$\sim$B$\sim$C$\succ$A) if $K=1,$

\item[(b)] semi-weakly intransitive (existence of A$\succ$B$\sim$C$\succ$A) if $K=2$ and $w^{^{(1)}}\varepsilon^{^{(1)}}=w^{^{(2)}}\varepsilon^{^{(2)}},$

\item[(c)] strictly intransitive (existence of A$\succ$B$\succ$C$\succ$A) if $K=2$
and $w^{^{(1)}}\varepsilon^{^{(1)}}\neq w^{^{(2)}}\varepsilon^{^{(2)}}$ or $K\geq3$.
\end{enumerate}%
%EndExpansion
\end{proposition}

The details of the definitions characterising intransitivity can be found in
Appendix \ref{AA4}. The proof is illustrated by Figure \ref{fig5}. The case A%
$\sim$B$^{\prime}\sim$C$^{\prime}\succ$A can be found in Figure \ref{fig5}a
(the second coordinate, $r^{(2)},$ can be ignored in this case). The
semi-weakly intransitive triplet A$\succ$B$\sim$C$^{\prime\prime}\succ$A in
the same figure does not depend on the direction of the red line of constant
fine utility. Finally, the points A$\succ$B$\succ$C$\succ$A form a strictly
intransitive triplet (note that C must be above the red line). In the
three-dimensional case shown in Figure \ref{fig5}b, we select $\varepsilon
^{^{(\alpha)}}=1.5$, \ $\alpha=1,2,3$, hence $1\sim^{(\alpha)}2\sim^{(\alpha
)}3\succ^{(\alpha)}1.$ Most values in the table are equivalent, while the
three strict preferences of 3 over 1 are shown by the arrows. It is easy to
see that the listed points form a strictly intransitive triplet A$\succ $B$%
\succ$C$\succ$A.
\begin{figure}[H]
\centering
\includegraphics[width=14cm,page=4, clip ]{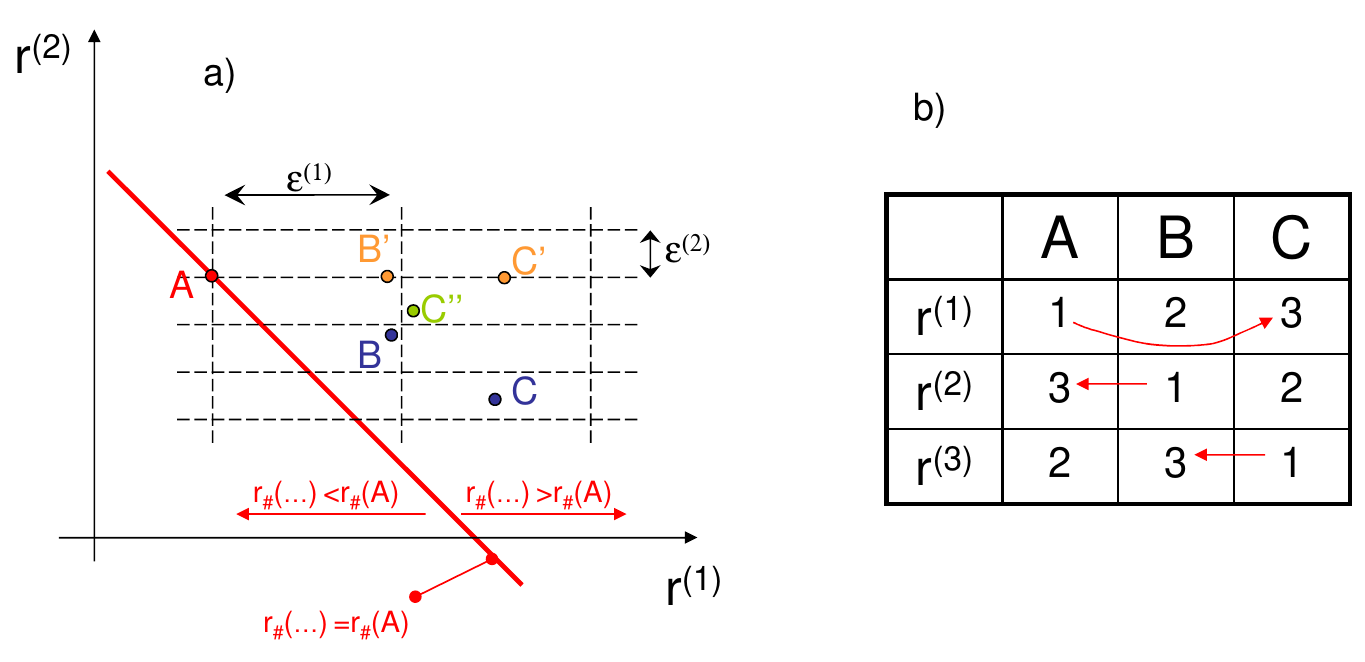}
\caption{Intransitivity of fractional (partial) selection criteria when these criteria are known approximately: (\textbf{a}) case of two criteria (\textbf{b}) case of three criteria. }
\label{fig5}

\end{figure}

Practically, coarsening in multidimensional cases becomes strictly
intransitive. The cases without strict intransitivity are degenerate: either
dimensions are redundant or coarsening is performed after merging the
fractional variables into the overall utility (instead of independent
coarsening for all or some of the criteria). From a\ philosophical
perspective, this statement can be presented as a continuum argument for
intransitivity \cite{Temkin1996}: small alterations are commonly overlooked
for secondary parameters but can be accumulated into critical differences.

\subsection{Imperfect Discrimination Due to the Presence of Noise}

If the exact value of utility $r$ is not known (which is often the case in
the real life), this can be expressed by adding a random variable $\xi ,$
which is assumed to be Gaussian, to the utility 
\begin{equation}
y=r+\xi /2^{1/2},\;\;P_{\xi }(\xi )=\frac{\exp \left( -\frac{\xi ^{2}}{%
2\sigma ^{2}}\right) }{\sigma \left( 2\pi \right) ^{1/2}}
\end{equation}%

The value $y$ is a measured, perceived or known estimate of the unknown
value $r$. When comparing A and B we need to estimate $\Delta r=r($A$)-r($B$%
) $ from known $\Delta y=y($A$)-y($B$).$ We can write $\Delta y=\Delta r+\xi 
$, where a difference of two independent Gaussian random values is shown as
a Gaussian random value, $\xi $. While obviously $\Delta r=\left\langle
\Delta y\right\rangle ,$ averages cannot be evaluated from a single
measurement---the preference must be evaluated deterministically on the
basis of known value $\Delta y.$ It is clear that small changes of $\Delta y$
should be ignored (\textit{i.e.}, $\Delta r\approx 0$) since $\left\vert \Delta
y\right\vert \ll \sigma $ is likely to be induced by the random noise $\xi .$%
\ In this case the sign of $\Delta y$\ does not tell us much about the sign
of $\Delta r$. If, however, $\left\vert \Delta y\right\vert \gg \sigma ,$
then $\Delta r\approx \Delta y$ with a high degree of certainty so that $%
\Delta y$ and $\Delta r$ are most likely to have the same sign. The
magnitude $\sigma $ of the noise is presumed to be known. While coarsening Equation~(%
\ref{p_rho_c}) implements these ideas abruptly (all or nothing), it is clear
that our confidence in estimating $r$ increases gradually as $y$ increases.

The estimate $\Delta y$ of the value $\Delta r$ is considered to be
satisfactory if 
\begin{equation}
\left\vert \Delta r-\Delta y\right\vert \leq\beta\;\Delta y
\end{equation}
with a sufficiently small $\beta$. In this case our preference is modelled
with the use of the following function 
\begin{equation*}
F(\Delta r,\Delta y)=\left\{ 
\begin{array}{cc}
0, & \text{if \ }\left\vert \Delta r-\Delta y\right\vert >\beta y \\ 
\Delta r, & \text{if \ }\left\vert \Delta r-\Delta y\right\vert \leq\beta y%
\end{array}
\right.
\end{equation*}

This function, which is illustrated in Figure \ref{fig6}, coincides with graded co-ranking $\Delta r=r($A$)-r($B$)$ when
our estimate is satisfactory and is set to zero otherwise, \textit{i.e.},
unsatisfactory estimates are ignored. The average of this function is 
\begin{equation}
\bar{F}(\Delta y)=\left\langle F(\Delta r,\Delta y)\right\rangle
=\int\limits_{\Delta y(1-\beta)}^{\Delta y(1+\beta)}\Delta rP_{\xi}\left( \Delta y-\Delta
r\right) d\Delta r=\Delta y\func{erf}\left( \frac{\left\vert \Delta
y\right\vert }{\varepsilon}\right)
\end{equation}
where $\varepsilon=2^{1/2}\sigma/\beta$. The function $\bar{F}(\Delta y)$
represents $\Delta y$ multiplied by a factor representing reliability of $%
\Delta y$ giving a satisfactory estimate for $\Delta r$, \textit{i.e.}, $\bar{F}%
(\Delta y)$ is the reliable fraction of $\Delta y$. This models our
inclination to ignore small $\Delta y=y($A$)-y($B$)$ and accept large $%
\Delta y$ while comparing A and B.

Assuming that the measured values $y^{(\alpha)}$ of the fractional utilities 
$r^{(\alpha)}$ are different from the true values due to presence of some
random noise, we are now compelled to define the fractional co-ranking by 
\begin{equation}
\rho^{^{(\alpha)}}(\text{A,B})=\bar{F}\left( y^{(\alpha)}(\text{A}%
)-y^{(\alpha)}(\text{B})\right) =\rho_{0}^{^{(\alpha)}}(\text{A,B})\func{erf}%
\left( \frac{\left\vert \rho_{0}^{^{(\alpha)}}(\text{A,B})\right\vert }{%
\varepsilon^{(\alpha)}}\right)  \label{c_rho}
\end{equation}
where $\rho_{0}^{^{(\alpha)}}($A,B$)=y^{(\alpha)}($A$)-y^{(\alpha)}($B$)$,
is the fine-graded co-ranking, which does not take into account the presence
of the noise (\textit{i.e.}, is based upon $\Delta r=\Delta y$), and $\varepsilon
^{(\alpha)}$ is determined by the intensity of noise in direction $\alpha$.
The fine co-ranking and two coarse co-rankings that correspond to threshold
coarsening Equation~(\ref{p_rho_c}) and Gaussian coarsening Equation~(\ref{c_rho}) with $%
\varepsilon^{(\alpha)}=1$ are shown in Figure \ref{fig1m}a. The effect of
smooth coarse grading on intransitivity is qualitatively similar to the
threshold case:

\begin{figure}[H]
\centering
\includegraphics[width=10cm,page=5, clip ]{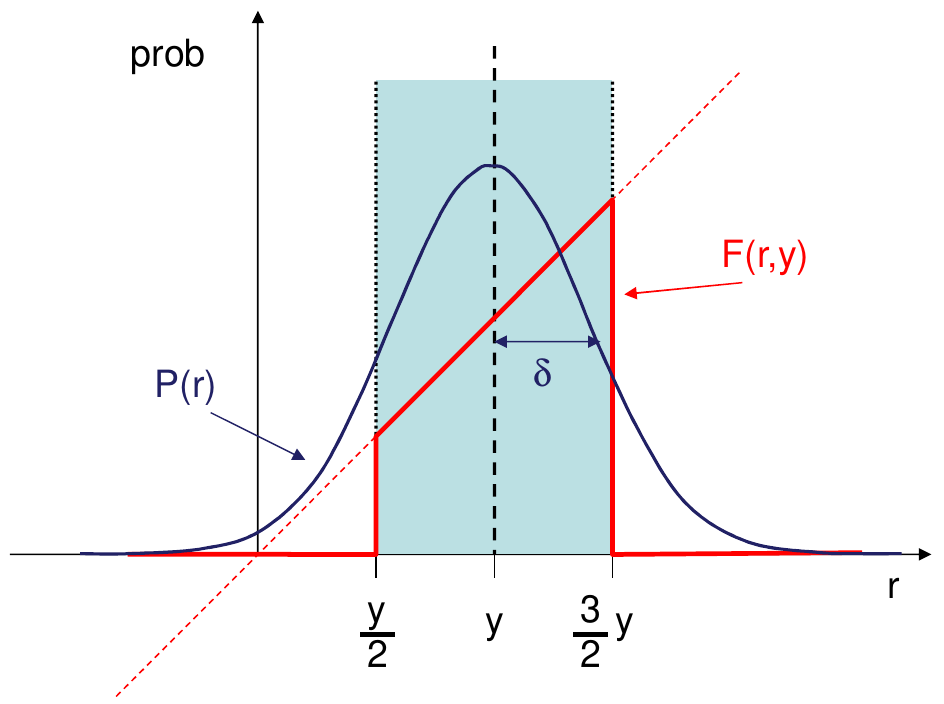}
\caption{Estimating the reliable fraction of a utility criterion in presence of Gaussian noise.} 
\label{fig6}
\end{figure}

\begin{proposition}
Gaussian coarsening in multiple dimensions $K>1$ leads to strict
intransitivity provided that not all $w^{^{(\alpha)}}\varepsilon^{^{(%
\alpha)}}$ are the same.
\end{proposition}

Figure \ref{fig1m}b demonstrates intransitivity A$\succ$B$\succ$C$\succ$A
for the case when coarsening occurs along the first direction (that is $%
\varepsilon^{(1)}=1$ and $\varepsilon^{(2)}=1/10$).
%===============  7
\begin{figure}[H]
\centering
\includegraphics[width=\textwidth,page=1, clip]{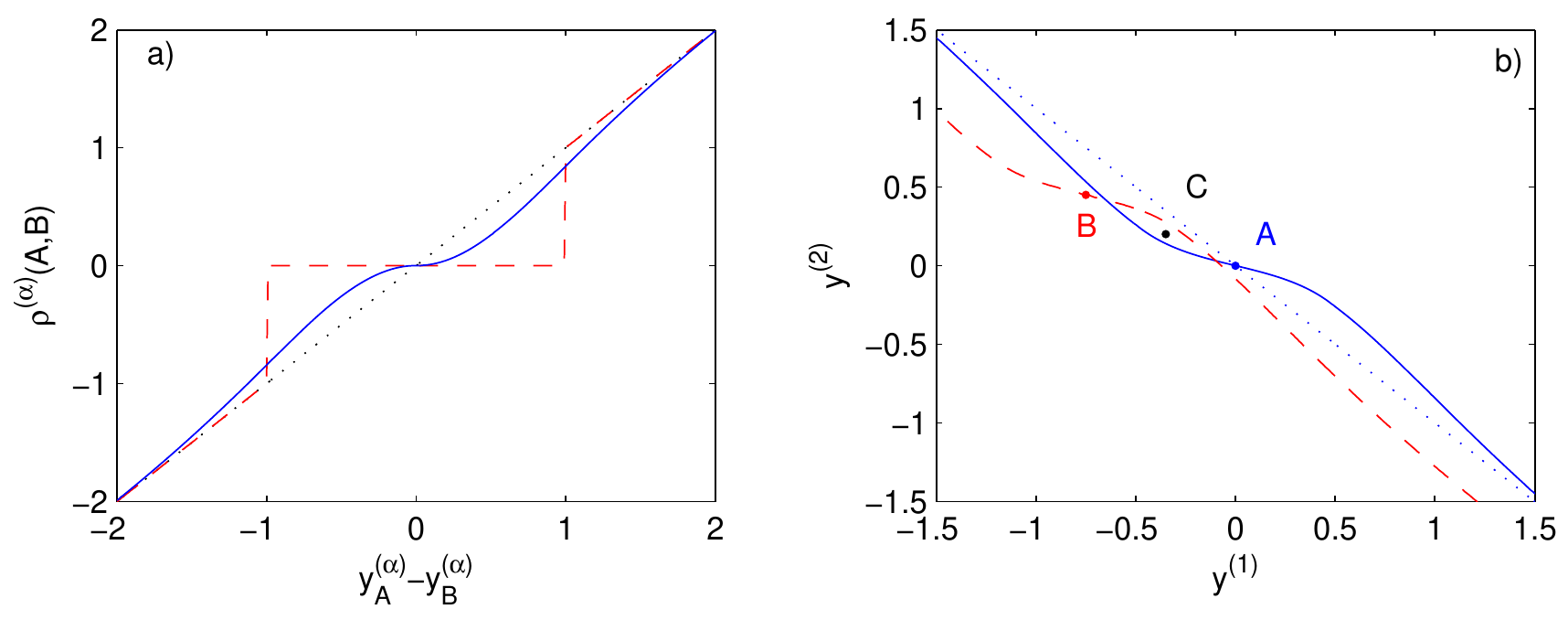}
\caption{(\textbf{a}) Co-ranking functions \textit{\textit{vs.}} fractional (partial) utility: dotted \mbox{line---original} \mbox{fine-graded;} dashed line---threshold-coarsened; 
solid line---Gauss-coarsened; \mbox{(\textbf{b}) Intransitivity} due to coarsening in two dimensions:  
solid line (blue) represents elements X that X$\sim$A and dashed line (red) represents elements X that X$\sim$B for  
Gauss-coarsened co-rankings with $\varepsilon^{(1)}=1$ and $\varepsilon^{(2)}=1/10$. 
Dotted (blue) line  represents elements X that X$\sim$A for \mbox{fine-graded~co-ranking. }
}
\label{fig1m}
\end{figure}
\vspace{-24pt}

\section{Risks and Benefits\label{S7}}

The problem considered in this section has two main parameters that
represent different quantities and are not trivially combinable into a
single value. One of these parameters has a greater uncertainty or is
contaminated by random noise. While this consideration is generic, we
interpret these parameters as risk and benefit. This is determined by three
factors. First, as discussed in the introduction, balancing risk and benefit
has been investigated in various contexts (decision-making under
uncertainty, personal preferences, portfolio management, \textit{etc.}). Second, risk
and benefit are not commensurable (at least not in a trivial or obvious
manner). The benefit is defined as the average pay-off so that increasing or
decreasing risk does not affect it directly. Given the benefit (which is
presumed to always have a positive utility), in real life people can be
risk-adverse or risk-seeking depending on the situation but, in this work,
we treat risk as a detrimental factor having negative utility. Third,
benefits (which are linked to mean values) are typically known with less
uncertainty than the associated risks (which are linked to stochastic
variances).

The fact that people tend to ignore small increments in risk has been
noticed in many publications. 
(Here we refer only to small increments of risk---most people are
over-sensitive to small risks in \mbox{comparison} to absence of any risk.
Typically people are risk-seeking for small probabilities of gains and
substantial probability of losses but the same people are risk-aversive for
small probabilities of losses and substantial probability of gains \cite%
{risk2013}). Rubinstein \cite{Rubinstein1988} suggested that the Allais
paradox is linked to common treatment of close probabilities as being
equivalent and noted that intransitivities are likely to appear in this
case, undermining the existence of utility. Leland \cite{Leland2002} noted
limited abilities of individuals to discriminate close probabilities.
Lorentziadis \cite{Lorentz2013} introduced division of probabilities into
ranges to reflect coarsened treatment of probabilities. This approach
requires an individual to discriminate very close probabilities located of
different sides of the range divides (this does not seem realistic but
preserves transitivity). Here, we follow these works and assume that the
risk is known to us with a substantial degree of uncertainty.

\subsection{Hidden Degradation}

It is often the case that seemingly positive incremental developments are
accumulated to create problems or malfunctions. This is impossible in
transitive systems (due to the absolute nature of transitive improvements)
but, if intransitivity is present, then an obvious improvement in one
parameter (e.g., higher benefit) may be accompanied by a tacit decrease in
performance with respect to another parameter (e.g., increasing risk). The
problem occurs when the risk becomes too high and the system malfunctions or
collapses. Tacit loss of competitiveness is called hidden degradation. In general, 
evolution of an intransitive system may result in competitive escalation or
competitive degradation. The degradation can be explicit or hidden \cite
{K-PS2012,K-PT2013,K_Ent2014a}.

Consider the following example: in order to improve the performance of
industrial turbines, the manufacturers commonly cut technical margins for
operational conditions of the components. This does not make the turbines
unsafe and does improve their efficiency. Competition between manufacturers
forces each of the competitors to cut margins further and further to reach
higher and higher efficiency producing turbines that become more and more
sensitive to fuel, servicing and other requirements. As discussed above, we
immediately acknowledge the increase in performance but hardly notice any
tiny increments in risk associated with reducing the margins. However, these
increments can gradually accumulate into vulnerability of the product. An
unexpected change in conditions (which can be very small in magnitude---a
different fuel, for example) can cause a malfunction or even make the
technology impractical. In intransitive conditions, competition and gradual
apparent improvements may result in a collapse due to accumulated negative
effects, which are collectively referred to as ``risk''. This is
intransitivity in action: each new design is better than the previous one
and, yet, one day the latest and seemingly best design fails miserably and
gives way to alternative technology.

It is interesting that knowledge of the treacherous nature of intransitive
competition does not always allow us to avoid its unwanted consequences. For
example a cautious turbine manufacturer deciding not to improve the
efficiency of its turbine is likely to be forced out of business well before
any intransitive effects will come into play.

Similar effects can be found in biology, economics and other disciplines.
For example, as the capacity for economic growth becomes more and more
saturated, investors have to undertake higher and higher risks to uphold
their profits. Accumulation of invisible risks makes the market unstable,
and, one day, the market collapses. There might be an external factor that
triggers the collapse, but the fundamental reason that makes this collapse
possible is the intransitivity of economic competition.

\subsection{Competitive Simulations for Risk-Benefit Dilemma}

This dilemma has two independent variables: (undesirable) risk $y^{(1)}$ and
(desirable) benefit $y^{(2)}$. In a simple transitive model, there exists a
1:1 trade off between the risk and the benefit according to co-ranking
defined by 
\begin{equation}
\rho_{0}(\text{A,B})=\left( y_{_{\text{A}}}^{(2)}-y_{_{\text{B}%
}}^{(2)}\right) -\left( y_{_{\text{A}}}^{(1)}-y_{_{\text{B}}}^{(1)}\right)
\label{s_rho0}
\end{equation}

It is obvious that, due to its linear form, co-ranking Equation~(\ref{s_rho0})
implies the absolute utility, which can be written~as 
\begin{equation}
r_{0}(\text{A})=y_{_{\text{A}}}^{(2)}-y_{_{\text{A}}}^{(1)}  \label{s_u0}
\end{equation}

Existence of the absolute utility $r_{0}($A$)$ indicates the transitivity of
this case.

If there is some uncertainty in evaluating the risk, the co-ranking takes
another form in accordance with Equation~(\ref{c_rho}) 
\begin{equation}
\rho(\text{A,B})=\left( y_{_{\text{A}}}^{(2)}-y_{_{\text{B}}}^{(2)}\right)
-\left( y_{_{\text{A}}}^{(1)}-y_{_{\text{B}}}^{(1)}\right) \func{erf}\left( 
\frac{\left| \left( y_{_{\text{A}}}^{(1)}-y_{_{\text{B}}}^{(1)}\right)
\right| }{\varepsilon}\right) ^{k}  \label{s_rho1}
\end{equation}
with $\varepsilon=1$ and $k=1$. Here we put $k=1$ but, in principle, $k$ can
be set to other values, thus changing the degree of coarsening. If $k=0$,
then Equation~(\ref{s_rho1}) coincides with Equation~(\ref{s_rho0}). If $k>0$, than the
co-ranking becomes intransitive and the corresponding absolute utility does
not exist.

In the simulations, the competing elements are represented by particles (the
same as particles used in modelling of reacting flows). The particles
compete according to the preferences specified by the co-ranking functions.
The losers are assigned the properties of the winners subject to mutations,
which are predominately directed towards $y^{(1)}=y^{(2)}=0$. The details
can be found in previous \mbox{publications~\cite{K-PT2013, K-PS2012, K_Ent2014a,
K_JSSSE}.} The gray area in Figure \ref{fig9} indicates prohibited space. The
boundary is Pareto-efficient: it is impossible to increase the benefit
without increasing the associated risk.
%===============  8
\begin{figure}[H]
\centering
\includegraphics[width=0.7\textwidth,page=8, clip]{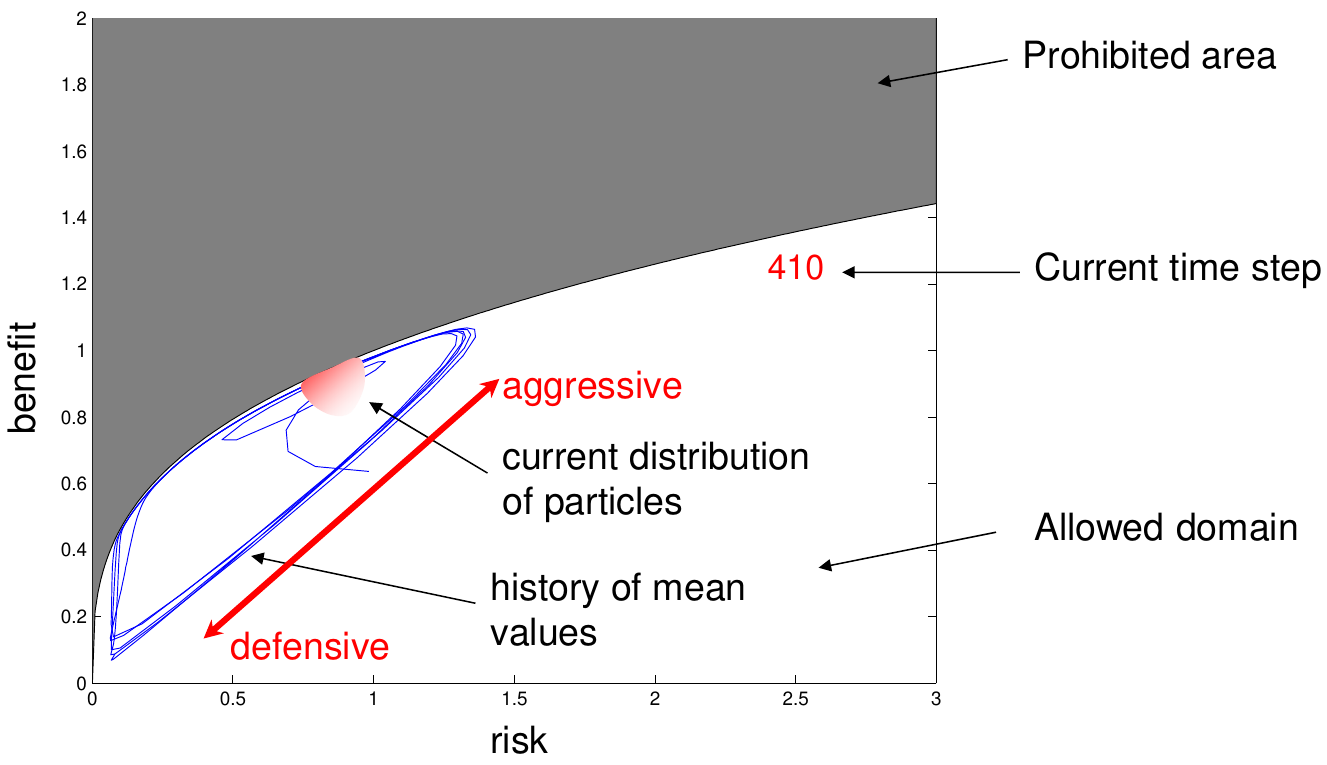}
\caption{Schematic of the simulation domain for the risk-benefit dilemma.}
\label{fig9}

\end{figure}

%===============  9
\begin{figure}[H]
\centering
\includegraphics[width=14cm,page=2, clip]{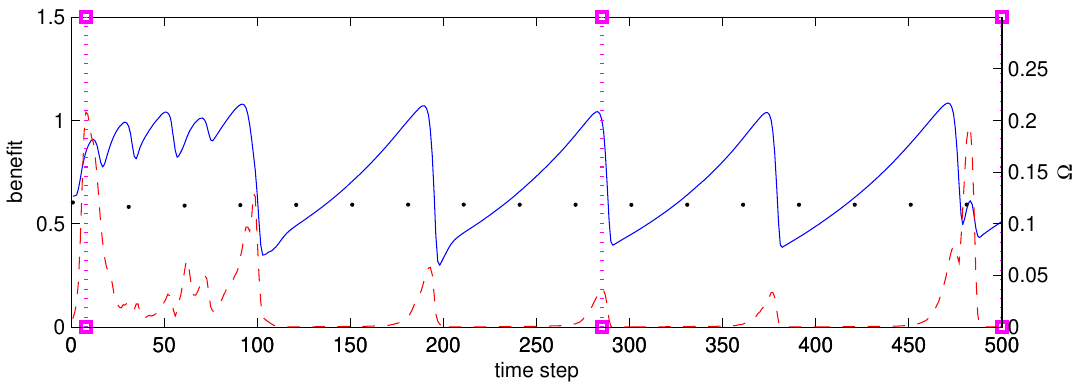}
\caption{Simulations of the risk-benefit dilemma. 
 Scale on the left: solid line \mbox{(blue)---average} benefit,  dots (black)---equilibrium state in the corresponding transitive competition. 
 Scale on the right: dashed line (red)---evolutionary intrasitivity parameter $\Omega$ (see Appendix \ref{AA3}).
 Vertical lines: cases shown in Figure \ref{fig3m}. }
\label{fig2m}

\end{figure}

The simulations are qualitatively similar to previous simulations \cite%
{K_Ent2014a, K_JSSSE} with a power-law \mbox{representation} of the co-ranking
functions. The simulations start from arbitrary conditions but promptly
(within $\sim$20 time steps) relax to quasi-equilibrium states that may
continue to evolve. These initial evolutions are similar in transitive and
intransitive cases and, as shown in \cite{K_Ent2014a}, these cases have
similar quasi-equilibrium distributions (although the present simulations
show more variations). The transitive cases Equation~(\ref{s_rho0}) quickly reach
equilibrium and stay in the equilibrium forever. The intransitive cases Equation~(\ref%
{s_rho1}) continue to evolve cyclically: the benefit and risk grow until the
system collapses into a defensive state of low benefit and low risk. Figure %
\ref{fig2m} indicates the existence of two periods of around 100 steps and
30 steps respectively in intransitive simulations with co-ranking defined by
Equation~(\ref{s_rho1}). There is a switch to the dominant mode at 100 steps. The
dots indicate the equilibrium state, which is necessarily achieved in
transitive simulations with co-ranking defined by Equation~(\ref{s_rho0}). The dashed
lines in Figures~\ref{fig2m} and \ref{fig3m} demonstrate the apparent
similarity of intransitive simulations to a transitive preference (although
not given by Equation~(\ref{s_rho0})). This similarity disappears prior to
collapses of the high-benefit states of the system. The details are given in
Appendix \ref{AA3}.

\begin{figure}[H]
\centering
\includegraphics[width=16cm,page=3]{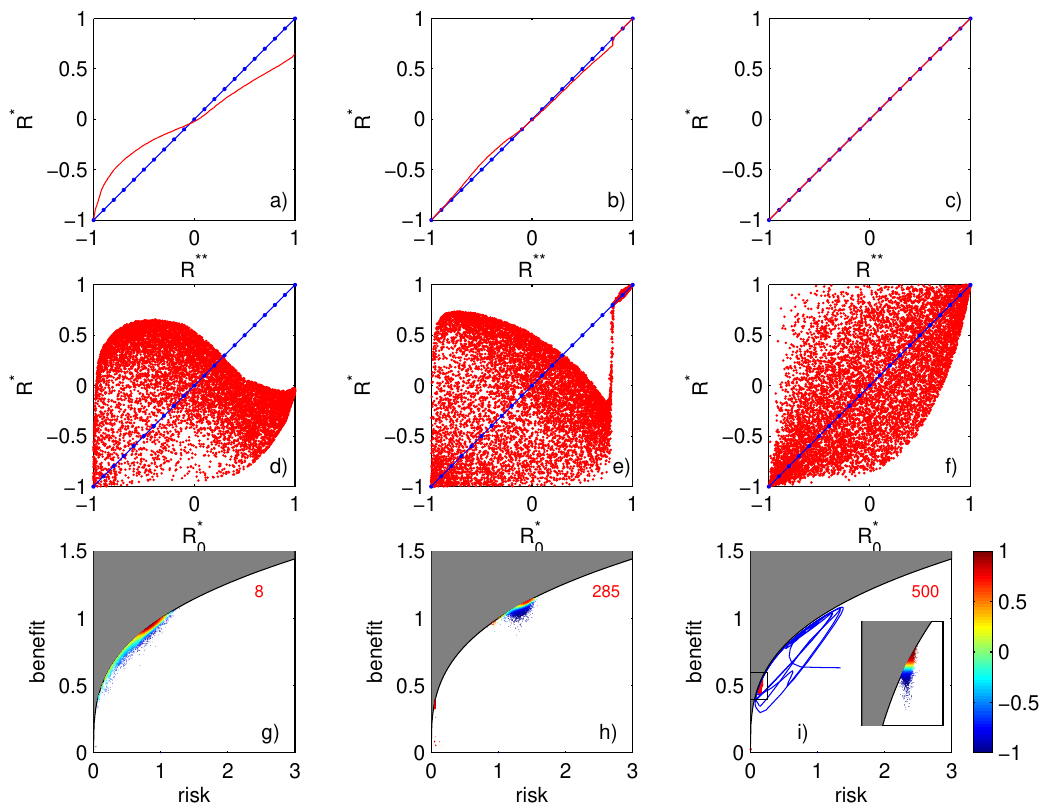}
\caption{ The case as in Figure \ref{fig2m}: (\textbf{a},\textbf{d},\textbf{g}) at 8 time steps; (\textbf{b},\textbf{e},\textbf{h}) at 285 time steps; (\textbf{c},\textbf{f},\textbf{i}) at 500 time steps.
Plots a,b,c: Primary current ranking $R^\ast$ \textit{vs.} secondary current ranking $R^{^{\ast\ast}}$ (solid red line). The cases shown are the same as in Figure \ref{fig3m}.
The solid blue line with dots corresponds to transitive or effectively transitive case where primary and secondary rankings coincide (see Appendix \ref{AA2a}).
Plots d,e,f: Intransitive primary current ranking $R^\ast$ evaluated for co-ranking $\rho$ specified by Equation~(\ref{s_rho1}) \textit{vs.} 
transitive primary current ranking $R_0^\ast$ evaluated for co-ranking $\rho_0$ specified by Equation~(\ref{s_rho0}). 
Deviations from the solid line with dots indicate differences 
between current rankings based on $\rho$ and on $\rho_0$.
Plots g,h,i: Domain snapshots. Colour shows the primary current ranking $R^\ast$ for each particle according to the colour bar. Plot i: blue line is a 500-step history of mean values;
enlarged box is shown as~insert.    
}
\label{fig3m}

\end{figure}
\vspace{-24pt}
\section{Thermodynamics and Intransitivity\label{S8}}

Physical thermodynamics is fundamentally transitive due to restrictions
imposed by the second and zeroth laws. For example, $T_{\text{A}}>T_{\text{B}%
}$ and $T_{\text{B}}>T_{\text{C}}$ require that $T_{\text{A}}>T_{\text{C}}.$
The concept of negative temperatures is compliant with the laws of
thermodynamics and does not alter transitivity. 
Negative temperatures are placed above positive temperatures (for example $%
T=-300$ K is hotter than \mbox{$T=300$~K}) but all temperatures are still linearly
ordered according to $-1/T.$ That is $T=+0$ K is the lowest possible
temperature, while $T=-0$ K is the highest possible temperature (see \cite%
{K-OTJ2012}). If +0 K were identical to $-$0 K (which is not the case), then the
thermodynamic temperatures would be intransitive). Hence, the constraints of
physical thermodynamics allow for cyclic or intransitive behaviour only in
thermodynamically open systems. There is no evidence of any kind that the
laws of thermodynamics are violated in complex evolutionary processes
(biological or social). Increase of order in a system is always compensated
by dispersing much larger quantities of entropy. The question that is often
discussed in relevant publications \cite{EIE1988} is not the letter but the
spirit of the thermodynamic laws---the possibility of explaining complex
evolutions using thermodynamic principles. Such explanations can be referred
to as apparent thermodynamics (\textit{i.e.}, thermodynamics-like behaviour explained
with the use of the theoretical machinery of thermodynamics).
%===============  10

\subsection{Transitive Competitive Thermodynamics}

As noted above, complex evolutionary systems are closer to thermodynamic
systems placed in an environment than to isolated thermodynamic systems. The
latter are characterised by maximisation of entropy while the former are
characterised by minimisation of Gibbs free energy or by maximisation of
free entropy, which, effectively, is Gibbs free energy taken with the
negative sign. Typical free energy and free entropy have two terms:
configurational and potential. In competitive systems, the configurational
terms reflect existence of chaotic mutations, while the potential terms
reflect the ordered competitiveness (or fitness, or utility), which
represents the propensity of elements to survive (theory and examples can be
found in previous publications \cite{K-PS2012,K-PT2013,K_Ent2014a}). Here we
refer to effective potentials that reflect competitiveness of elements
placed into certain conditions (or environment). As in conventional
thermodynamics, equilibrium is the balance between chaos (configurational
terms) and order (potential terms). This approach explains transitive
effects such as evolution directed towards increase of fitness, which is
characterised by increase in apparent entropy. The process of competition is
also characterised by production of physical entropy at least because
information is continuously destroyed when the properties of losers are
discarded.

While physical thermodynamics is always transitive, the same cannot be
stated \textit{a priori} about apparent~thermodynamics.

\subsection{Nearly Transitive Systems}

Here we refer to the systems as intransitive but their evolution remains
close, in one way or another, to evolution of transitive systems. These
cases are shown in Figure \ref{fig7} as: (a) locally intransitive systems,
where large-scale preference is transitive or (b) globally intransitive
systems, where transitivity is preserved locally. In case (a), we can treat
evolution as transitive if local details of the evolution can be neglected.
In this approximate consideration, we treat elements A$_{1}$ B$_{1}$ and C$%
_{1}$ as equivalent and distinguish only transitive preferences 
\begin{equation}
\text{A}_{1}\sim\text{B}_{1}\sim\text{C}_{1}\succ\text{A}_{2}\sim\text{B}%
_{2}\sim\text{C}_{2}\succ...
\end{equation}

This corresponds to the \textit{transitive closure} of the original
preference \cite{K-PS2012}.
%===============  11
\begin{figure}[H]
\centering
\includegraphics[width=15cm,page=6, clip ]{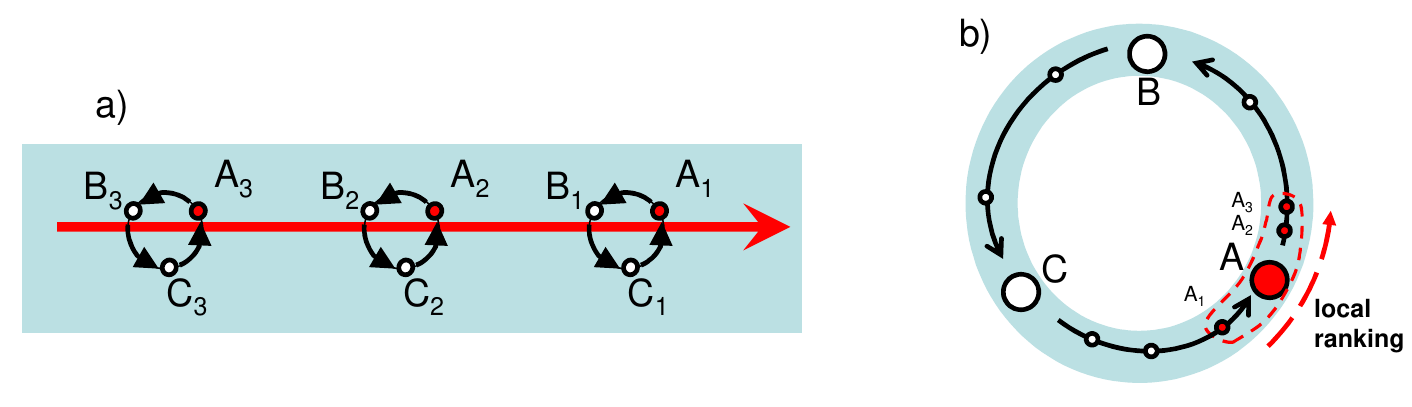}
\caption{Intransitive systems that display a degree of similarity with transitive systems: 
 (\textbf{a}) locally intransitive globally transitive system and (\textbf{b}) locally transitive globally intransitive~system.}
\label{fig7}

\end{figure}
In case (b), the properties of conventional thermodynamics are preserved as
long as our consideration is confined to regions where conditions are
transitive, but intransitive cycles can appear on larger scales. Consider
competitive system shown in Figure \ref{fig7}b where the competition between
A, A$_{1},...,$A$_{3}$ is transitive and unique ranking is possible within
the region around element A. Under conditions specified in \cite%
{K-PS2012}, this local competition can be characterised by competitive
potentials $\chi$, which are analogous to chemical potentials of
conventional thermodynamics taken with a negative sign. However, the
competition is intransitive if we look at larger scales: A$\prec$B$\prec$C$%
\prec$A. If three identical systems have elements A, B and C as leading
particles, the competitive potentials are related to each other by 
\begin{equation}
\chi_{_{\text{A}}}\prec\chi_{_{\text{B}}}\prec\chi_{_{\text{C}}}\prec
\chi_{_{\text{A}}}
\end{equation}

Note that we cannot write $\chi_{_{\text{A}}}<\chi_{_{\text{B}}}<\chi_{_{%
\text{C}}}<\chi_{_{\text{A}}}$ since this does not make mathematical sense.
Competitive potentials are numbers in transitive systems but they must be
treated as more complex objects in intransitive systems.

\subsection{Strong Intransitivity}

If strictly intransitive triplets are both local (or dense---can be found
in any small vicinity of every point) and global, we identify this
intransitivity as strong. The properties of such systems differ from
properties of conventional thermodynamic systems, and the benefits of using
thermodynamical analogy for stochastic theories of strongly intransitive
systems are not clear~\cite{K-PS2012,K-PT2013}. In these systems, the
similarity with conventional thermodynamics becomes quite remote and the
term ``kinetics'' seems to be more~suitable.

Complex kinetics of strongly intransitive competitive systems must at least
account for

\begin{enumerate}
\item rapid relaxation to a quasi-steady state (which, possibly, can be
approximately treated as transitive) with subsequent slow evolution;

\item the possibility of alternative directions of evolution (\textit{i.e.},
competitive degradation or \mbox{competitive~escalation);}

\item violation of Boltzmann's \textit{Stosszahlansatz} (the hypothesis of
stochastic independence of the system elements) and formation of cooperative
structures.
\end{enumerate}

Further examples and discussion of competitive thermodynamics and its
limitations can be found in~\cite{K-PS2012,K-PT2013,K_Ent2014a}.

\section{Discussion and Conclusions\label{S9}}

This work endeavours to introduce a consistent approach to treat the problem
of intransitivity, which is relevant to many different fields of knowledge
that have to deal with preferences and/or competition. In general, a
theoretical treatment of competition and preferences allows for two
alternative self-consistent frameworks: absolute and relativistic. The first
framework presupposes the existence of absolute characteristics that are
commensurable and known exactly. This framework is fundamentally transitive
and more simple from the perspective of theoretical analysis. The second
framework is based on relative characteristics. In this framework,
assessment of performance of the elements depends on perspective and
intransitivity is common. Many real world effects that cannot be explained
within the absolute framework (and thus are commonly viewed as irrational or
abnormal), become perfectly logical in the relativistic framework.

The focus of this work is the existence of key links between intransitivity
and dependence of conditional preferences on the reference group. Our main
conclusion is that intransitivity must be typical for systems involving
preferences and competition since intransitivity appears under any of the
following conditions, all of which are common in the real world:

\begin{itemize}
\item relative comparison criteria or

\item multiple comparison criteria that are incommensurable or 

\item multiple comparison criteria that are known approximately or 

\item comparisons of groups of comparable elements.
\end{itemize}

In some disciplines (e.g., social choice, population ecology),
intransitivities have been acknowledged and studied but in other disciplines
(law, economics) they tend to be ignored. While intransitivity represents
rather special cases in law that legal professionals need to be aware of,
intransitive behaviours seem to be ubiquitously present in economic systems.
This does not mean, of course, that assumptions of transitivity are always
erroneous, but large economic systems tend to behave in a typically
intransitive manner (e.g., with cycles and collapses). It is difficult to say
if intransitivity of personal preferences is the main factor causing this
behaviour 
(There could be other reasons. For example, intransitivity of competition,
which is not related to intransitivity of personal preferences of the
consumers, has been detected in the American car industry \cite{K_JSSSE}). 
It is worthwhile to note that intransitive systems can behave so that their
short-term evolution seems perfectly transitive---this effect corresponds
to transitive theories. This point is well illustrated by the competitive
simulations of the risk/benefit dilemma: segments of this evolution seem to
be transitive, until a point is reached where the system collapses and then
behaves cyclically, which is typical of intransitive competitions and
impossible in transitive competitions. Therefore the main purpose of this
work is not in criticising assumptions of transitivity, which can be quite
reasonable in many circumstances, but in pointing out that these assumptions
are excessively restrictive to deal with more complex phenomena. We have all
indications that emergence of complexity in competitive systems is linked to
intransitivity. Disciplines that dismiss intransitivity outright as being
``irrational" or ``non-scientific" are likely to experience an increasing
pressure of a paradigm shift in coming years.

Transitive systems are characterised by absolute quantities (e.g., absolute
rankings). Absolute quantities do not exist in intransitive systems, which
have to operate with relative quantities, such as \mbox{co-rankings.}
Thermodynamics (here we refer to both physical thermodynamics and
\mbox{thermodynamics-like} theories) is primarily equipped to deal with transitive
phenomena and the physical thermodynamics is fundamentally transitive.
Although apparent thermodynamics may allow for a limited presence of
intransitivity, it seems that the strongly intransitive effects should be
studied by the relevant complex kinetics, which is at the very beginning of
its development. The main question is whether common features observed in a
number of complex systems of different physical nature can be analysed at a
generic level while neglecting less significant case-specific details. The
author of this work believes that the answer to this difficult question is
positive and that it requires consideration of evolving systems with
competition based on generally intransitive preferences. The following
patterns of behaviour are indicative of presence of intransitivity in
competitive systems:

\renewcommand{\labelitemi}{$\bullet$}

\begin{itemize}
\item relativity of strength and dependence of preference on perspective

\item cyclic behaviour instead of relaxation to a unique equilibrium 

\item relatively slow evolutions punctuated by sudden collapses and changes%

\item complex patterns of behaviour (e.g., cooperative structures).
\end{itemize}

It must be noted that there exist psychological barriers in accepting
intransitivity as a realistic phenomenon, at least in some fields. This has
delayed investigation of complex effects, although reasons for this are
understandable. The transitive world is black and white: fit versus unfit,
strong versus weak and good versus bad---this is a simplification that can
be useful as it is more easily amenable to theoretical analysis or
explanation. The intransitive world is much more colourful: what is good or
bad depends on perspective and on the current situation, and possible
situations have endless varieties. While theoretical advances are more
difficult in intransitive systems, there lays the explanation of the
complexities of the real world.

%\begin{figure}[h!]
%\begin{center}
%\includegraphics[width=10cm,page=4]{figmtl.pdf}
%\caption{ Primary current ranking \textit{vs.} secondary current ranking (solid red line). The cases shown are the same as in Figure \ref{fig3m}.
%The solid blue line with dots corresponds to transitive or effectively transitive case where primary and secondary rankings coincide (see Appendix).}
%\label{fig4m}
%\end{center}
%\end{figure}

\acknowledgements{Acknowledgments} \ The author thanks D.A. Klimenko for
many fruitful discussions and a number of useful corrections. The author
acknowledges funding by the Australian Research Council.

\appendix
\section*{\noindent Appendix}
\setcounter{figure}{0}
\renewcommand\thefigure{A\arabic{figure}}
\setcounter{equation}{0}
\renewcommand\theequation{A\arabic{equation}}

\setcounter{proposition}{0}
\renewcommand\theproposition{A\arabic{proposition}}
\section{Potential Intransitivity\label{AA1}}

It is possible to have a \textit{potentially intransitive} co-ranking $%
\rho(...,...)$ that defines a transitive preference on a given set of
elements but cannot be represented by transitive relation Equation~(\ref{r_diff}). In
this case there always exists an absolutely transitive \mbox{co-ranking} that
defines an equivalent preferences on this set. Indeed, transitivity of the
preference allows for the introduction of an\ absolute ranking and then the
transitive \mbox{co-ranking} is determined by Equation~(\ref{r_diff}). While co-rankings
that cannot be represented by Equation~(\ref{r_diff}) may define a transitive
preference on the current set of elements (\textit{current transitivity}),
the underlying rules for these preferences can be expected to be
intransitive and intransitivity can be revealed when these co-rankings are
used in a different context or conditions. This is illustrated by the
following~example.

Consider a transitive fractional preference between a set three elements 
\begin{equation}
\text{A}\preceq ^{(1)}\text{B}\preceq ^{(1)}\text{C}\succeq ^{(1)}\text{A}
\label{tc_trip}
\end{equation}%
which corresponds to the co-ranking%
\begin{equation}
\rho ^{(1)}(\text{B,A})\geq 0,\;\rho ^{(1)}(\text{C,B})\geq 0,\text{\ }\rho
^{(1)}(\text{C,A})\geq 0  \label{tc_co}
\end{equation}%

If $\delta =0$ in the relation 
\begin{equation}
\delta _{1}=\delta ^{(1)}(\text{C,B,A})=\rho ^{(1)}(\text{C,B})+\rho ^{(1)}(%
\text{B,A})+\rho ^{(1)}(\text{A,C})=\rho _{t}^{(1)}(\text{C,A})-\rho ^{(1)}(%
\text{C,A})  \label{tc_del}
\end{equation}%
where 
\begin{equation}
\rho _{t}^{(1)}(\text{C,A})=\rho ^{(1)}(\text{C,B})+\rho ^{(1)}(\text{B,A})
\end{equation}%
is absolutely transitive approximation for $\rho ^{(1)}($C,A$)$, then
fractional co-ranking Equation~(\ref{tc_co}) can be represented by Equation~(\ref{r_diff}) in
terms of the following absolute fractional ranking 
\begin{equation}
r^{(1)}(\text{A})=0,\;r^{(1)}(\text{B})=\rho ^{(1)}(\text{B,A}),\;r^{(1)}(%
\text{C})=\rho _{t}^{(1)}(\text{C,A})
\end{equation}%

However, if $\delta _{1}\neq 0$, this representation is impossible. In this
case setting $\rho ^{(1)}($C,A$)\geq 0$ to $\rho _{t}^{(1)}($C,A$)\geq 0$
and ensuring that $\delta _{1}=0$ would remove potential intransitivity but
alter the magnitude of our preference. Note that a triplet Equation~(\ref{tc_trip})
with $\delta _{1}\neq 0$ can always be found in an arbitrary set with a
potentially intransitive co-ranking (otherwise we may define $r($B$)=r($A$%
)+\rho ($B,A$)$ for fixed A and arbitrary B).

While intransitivity is not currently visible in this example, it can appear
when co-ranking Equation~(\ref{tc_co}) is used in a different way. The preference Equation~(%
\ref{tc_trip}) does not have any rules specified for other elements, but we
can investigate how this preference is combined with other fractional
preferences that are strictly transitive. Consider a fractional ranking for
the second criterion defined by%
\begin{equation}
r^{(2)}(\text{A})=-\epsilon-\rho^{(1)}(\text{A,B}),\;r^{(2)}(\text{B}%
)=0,\;r^{(2)}(\text{C})=\epsilon-\rho^{(1)}(\text{C,B})
\end{equation}

This ranking corresponds to the fractional co-ranking 
\begin{equation}
\rho^{(2)}(\text{B,A})=\epsilon-\rho^{(1)}(\text{B,A}),\;\rho^{(2)}(\text{C,A%
})=-\rho^{(1)}(\text{B,A})-\rho^{(1)}(\text{C,B})+2\epsilon ,\;\rho^{(2)}(%
\text{C,B})=\epsilon-\rho^{(1)}(\text{C,B})  \label{tc_22}
\end{equation}

If the overall co-ranking is defined as the average $\rho=(\rho^{(1)}+%
\rho^{(2)})/2$ then we obtain 
\begin{equation}
\rho(\text{B,A})=\epsilon/2,\;\;\rho(\text{C,B})=\epsilon/2,\;\;\rho (\text{%
A,C})=\delta_{1}/2-\epsilon  \label{tc_int}
\end{equation}

If $\delta_{1}>0,$ then selecting $\delta_{1}/2>\epsilon>0$ determines an
overall preference that is currently intransitive. 
\begin{equation*}
\text{A}\prec\text{B}\prec\text{C}\prec\text{A}
\end{equation*}

The case $\delta_{1}<0$ is treated in a similar manner. If $\delta_{1}=0$,
intransitivity cannot appear. This discussion can alternatively be
summarised in form of a short proposition:

\begin{proposition}
\label{PAA1}A potentially intransitive co-ranking can always be represented
as a superposition of a currently intransitive co-ranking and an absolutely
transitive co-ranking.
\end{proposition}

Three elements with original co-ranking $\rho _{0},$ where $\rho _{0}=\rho
^{(1)}$ in Equation~(\ref{tc_del}) with $\delta _{1}\neq 0,$ can always be selected
from a set with potentially intransitive co-rankings. Consider $\rho
_{0}=\rho _{1}+\rho _{2}$ where, $\rho _{1}=-\rho ^{(2)}$ defined by Equation~(\ref%
{tc_22}) is transitive, and $\rho _{2}=2\rho $ specified by Equation~(\ref{tc_int})
is currently intransitive with a proper choice of $\varepsilon $. This
decomposition is extended to the remaining elements so that $\rho _{1}$
remains absolutely transitive. The representation $\rho _{0}=\rho _{1}+\rho
_{2}$ is, obviously, not unique.

\section{Preference Properties and Indicator Co-Ranking\label{AA2}}\vspace{-12pt}
\setcounter{figure}{0}
\renewcommand\thefigure{B\arabic{figure}}
\setcounter{equation}{0}
\renewcommand\theequation{B\arabic{equation}}
\setcounter{proposition}{0}
\renewcommand\theproposition{B\arabic{proposition}}
\setcounter{theorem}{0}
\renewcommand\thetheorem{B\arabic{theorem}}

\subsection{Coarsenings and Refinements}

Consider two different preference rules specified on the same set of
particles: $\prec_{a}$ and $\prec_{b}$. Rule $b$ represents a \textit{%
coarsening} of rule $a$ if A$\prec_{b}$B demands that A$\prec_{a}$B for any
A and B (although A$\sim_{b}$B may correspond to any of A$\sim_{a}$B$,\;\ $A$%
\prec_{a}$B or A$\succ_{a}$B). If rule $b$ represents a \textit{coarsening}
of rule $a$ then\ the same property is expressed by saying that rule $a$
represents a \textit{refinement} of rule $b,$ implying that A$\preceq_{a}$B
demands that A$\preceq_{b}$B for any A and B (if A$\succ_{b}$B was correct,
then this would demand A$\succ_{a}$B and thus contradict A$\preceq_{a}$B). %
(Statement ``relation $\preceq_{b}$ \ \textit{contains}
relation $\preceq_{a}$''\ is another equivalent, which is
often used in literature but can be confusing in the context of the present
work). If rules $a$ and $b$ are refinements of one another, then these rules
are obviously equivalent.

\subsection{Transitive Preferences in Intransitive Systems}

The \textit{transitive closure} of a preference (generally intransitive) is
the minimal (\textit{i.e.}, most refined) transitive ordering that represents a
coarsening of the original preference (see the Appendix in~\cite%
{K-PS2012} for details). The ordering produced by the transitive closure is
indicated by the subscript \textquotedblleft $t$\textquotedblright\ (\textit{i.e.}, $%
\prec _{t},$ $\sim _{t},\preceq _{t},$ \textit{etc.}) or\ by $\prec \prec $, which is
equivalent to $\prec _{t}$, and by $\approx $, which is equivalent to $\sim
_{t}$. Hence the relation A$\prec \prec $B, which is expressed as
\textquotedblleft A is \textit{transitively preferred} to
B\textquotedblright , indicates that A$\prec $B and A$\prec _{t}$B (where
preference $\prec $ is generally intransitive). The relation A$\prec \prec $%
B implies that there should not exist any set \{C$_{1}$,...,C$_{k}$\} such
that A$\prec $B$\preceq $C$_{1}\preceq ...\preceq $C$_{k}\preceq $A and, in
particular, there is no such C that can form an intransitive triplet C$%
\preceq $A$\prec $B$\preceq $C.

The symbol $\preccurlyeq,$ which is used to indicate $\prec$ but not $%
\prec\prec$ , should be distinguished from $\preceq,$ which implies $\prec$
or $\sim$. \ In general,\ comparison of A and B belongs to one of the five
mutually exclusive possibilities listed in the second row of the following
table 
\begin{equation}
\begin{tabular}{|c|c|c|c|c|}
\hline
\multicolumn{3}{|c|}{$\text{A}\preceq\text{B}$} & \multicolumn{2}{|c|}{$%
\text{A}\succ\text{B}$} \\ \hline
$\text{A}\prec\prec\text{B}$ & $\text{A}\preccurlyeq\text{B}$ & $\text{A}\sim%
\text{B}$ & $\text{A}\succcurlyeq\text{B}$ & $\text{A}\succ\succ\text{B}$ \\ 
\hline
$\text{A}\prec\prec\text{B}$ & \multicolumn{3}{|c|}{$\text{A}\approx\text{B}$%
} & $\text{A}\succ\succ\text{B}$ \\ \hline
\multicolumn{2}{|c|}{$\text{A}\prec\text{B}$} & \multicolumn{3}{|c|}{$\text{A%
}\succeq\text{B}$} \\ \hline
\end{tabular}%
\end{equation}

\newpage
\subsection{Transitivity and the Indicator Co-Ranking}

Different versions of Propositions \ref{P1} and \ref{P1a} are formulated
below for conditional indicator rankings, \textit{i.e.}, for $R_{_{\mathbb{G}}}($A$)$
defined by Equation~(\ref{r_R_G})---the average conditional rankings based on the
indicator co-ranking $R($A,B$)$.

\begin{proposition}
\label{PA1a}If the underlying preference is currently transitive, then the
preference induced by the conditional indicator ranking (\textit{i.e.}, A$\preceq _{_{%
\mathbb{G}}}$B $\Longleftrightarrow R_{_{\mathbb{G}}}($A$)\leq R_{_{\mathbb{G%
}}}($B$)$) represents a\ coarsening of the underlying preference for any
reference group $\mathbb{G}$.
\end{proposition}

Demonstrating the validity of this statement is not difficult. Due to
transitivity, the preference A$\preceq$B demands that $R($A$\mathbf{,}$C$%
)\leq R($B$\mathbf{,}$C$)$ for any C$.$ Note that A$\prec$B but $R($A$%
\mathbf{,}$C$)=R($B$\mathbf{,}$C$)$ is possible for some C. We compare
definitions of $R_{_{\mathbb{G}}}($A$)$ and $R_{_{\mathbb{G}}}($B$)$ by Equation~(\ref%
{r_R_G}) and conclude that A$\preceq$B demands $R_{_{\mathbb{G}}}($A$)\leq
R_{_{\mathbb{G}}}($B$),$ since $R($A$\mathbf{,}$C$_{i})\leq R($B$\mathbf{,}$C%
$_{i})$\ for all C$_{i}\in\mathbb{G}$.\ Note that the combination of \mbox{$R_{_{%
\mathbb{G}}}($A$)=R_{_{\mathbb{G}}}($B$)$} and A$\prec$B is possible for some 
$\mathbb{G}$.

\begin{proposition}
\label{PA1aa}If conditional indicator rankings based on different reference
groups are strictly non-equivalent (\textit{i.e.}, there exist at least two elements A
and B and at least two groups $\mathbb{G}^{\prime}$ and $\mathbb{G}%
^{\prime\prime}$ so that A$\prec_{_{\mathbb{G}^{\prime}}}$B and \mbox{A$\succ _{_{%
\mathbb{G}^{\prime\prime}}}$B),} then the underlying preference is currently
intransitive.
\end{proposition}

Indeed, if the underlying preference were currently transitive, then,
according to Proposition \ref{PA1a}, preferences induced by the conditional
indicator ranking would be coarsenings of the underlying preference. Therefore, A$%
\prec _{_{\mathbb{G}^{\prime }}}$B demands A$\prec $B and A$\succ _{_{%
\mathbb{G}^{\prime \prime }}}$B demands A$\succ $B, which are contradictory.
Hence the underlying preference must be intransitive.

If an underlying preference is defined for the elements, any two groups of
elements $\mathbb{G}^{\prime }$ and $\mathbb{G}^{\prime \prime }$ can be
compared on the basis of the group indicator co-ranking defined by Equation~(\ref%
{r_R_GG}), that is $\mathbb{G}^{\prime }\succeq _{R}\mathbb{G}^{\prime
\prime }$ iff $\bar{R}(\mathbb{G}^{\prime },\mathbb{G}^{\prime \prime })\geq
0$. We note that

\begin{proposition}
\label{PA1aaa}Group preference based on the indicator co-ranking is not
necessarily transitive even if the underlying element preference is
currently transitive.
\end{proposition}

Consider three sets $\mathbb{G}_{A}=\{2,4,9\},$ $\mathbb{G}_{B}=\{3,5,7\}$
and $\mathbb{G}_{C}=\{1,6,8\},$ where the numbers indicate the absolute
rankings of nine distinct elements. It is easy to see that $\mathbb{G}%
_{A}\prec _{R}\mathbb{G}_{B}\prec _{R}\mathbb{G}_{C}\prec _{R}\mathbb{G}_{A}$
since $\bar{R}(\mathbb{G}_{B},\mathbb{G}_{A})=\bar{R}(\mathbb{G}_{C},\mathbb{%
G}_{B})=\bar{R}(\mathbb{G}_{A},\mathbb{G}_{C})=1/9$. This example
corresponds to the dice shown in Figure~\ref{fig0}b.

\section{Primary and Secondary Rankings \label{AA2a}}
\setcounter{figure}{0}
\renewcommand\thefigure{C\arabic{figure}}
\setcounter{equation}{0}
\renewcommand\theequation{C\arabic{equation}}
\setcounter{proposition}{0}
\renewcommand\theproposition{C\arabic{proposition}}
\setcounter{theorem}{0}
\renewcommand\thetheorem{C\arabic{theorem}}
%\appendix{Transitivity and secondary ranking}

This appendix presents definitions and statements related to the indicator
co-ranking $R(...,...)$ that are intended for characterisation of the level
of intransitivity in large and, possibly, evolving systems. The main
assumption of this section is that all elements in system $\mathfrak{S}_{0}$
are connected and form an overall group $\mathbb{G}_{0}$, implying that
there is a positive weight $g_{i}=g($C$_{i})>0$ specified for every element C%
$_{i}\in \mathbb{G}_{0}$. In this section, weights $g_{i}$ are associated
with the whole system $\mathbb{G}_{0}$ and thus is the same for all groups $%
\mathbb{G}_{q}\mathbb{\subseteq G}_{0}$, which are referred to in this
section as sets $\mathbb{S}_{q}=\mathbb{G}_{q}$, $q=1,2,...$ to emphasise
that the element weights are specified only for the whole system.

\subsection{Current Rankings}

Consider a system $\mathbb{G}_{0}$ of connected elements. The current
ranking is defined by 
\begin{equation}
R^{\ast }(\text{A})=\bar{R}(\text{A},\mathbb{G}_{0})=\frac{1}{G_{0}}\sum_{%
\text{C}_{i}\in \mathbb{G}_{0}}g(\text{C}_{i})R(\text{A}\mathbf{,}\text{C}%
_{i})  \label{A1_R}
\end{equation}%
that is $R^{\ast }($A$)$ is the conditional ranking Equation~(\ref{r_R_G}) of
element A with respect to all other elements in the system. In the same way
the current ranking can be introduced for an arbitrary set $\mathbb{G}_{q}$ 
\begin{equation}
\bar{R}^{\ast }(\mathbb{G}_{q})=\bar{R}(\mathbb{G}_{q},\mathbb{G}_{0})=\frac{%
1}{G_{q}}\sum_{\text{C}_{j}\in \mathbb{G}_{q}}g(\text{C}_{j})R^{\ast }(\text{%
C}_{j})=\frac{1}{G_{0}G_{q}}\sum_{\text{C}_{j}\in \mathbb{G}_{q}}\sum_{\text{%
C}_{i}\in \mathbb{G}_{0}}g(\text{C}_{j})g(\text{C}_{i})R(\text{C}_{j}\mathbf{%
,}\text{C}_{i})  \label{A1_RG}
\end{equation}%
where $G_{0}$ and $G_{q}$ are total weights of the system and of the set $%
\mathbb{G}_{q}$. The group co-ranking $\bar{R}(\mathbb{G}_{q},\mathbb{G}%
_{0}) $ is introduced according to Equation~(\ref{r_R_GG}). The current
ranking $R^{\ast }(...),$ the underlying preference $\prec ,$ and the
corresponding\ indicator co-ranking $R($A$\mathbf{,}$C$)$ are referred to as 
\textit{primary }when we need to distinguish them from secondary
characteristics.

The preference induced by the primary current co-ranking is referred to as 
\textit{secondary} as it is now considered as the underlying preference for
the secondary characteristics. The \textit{secondary current ranking} is
defined analogously to the primary current ranking 
\begin{equation}
R^{^{^{\ast \ast }}}(\text{A})=\bar{R}^{\prime \prime }(\text{A}\mathbf{,}%
\mathbb{G}_{0})=\frac{1}{G_{0}}\sum_{\text{C}_{i}\in \mathbb{G}_{0}}g(\text{C%
}_{i})R^{\prime \prime }(\text{A}\mathbf{,}\text{C}_{i})  \label{A1_RR}
\end{equation}%
\begin{equation}
\bar{R}^{^{\ast \ast }}(\mathbb{G}_{q})=\bar{R}^{\prime \prime }(\mathbb{G}%
_{q},\mathbb{G}_{0})=\frac{1}{G_{q}}\sum_{\text{C}_{j}\in \mathbb{G}_{q}}g(%
\text{C}_{j})R^{^{\ast \ast }}(\text{C}_{j})=\frac{1}{G_{0}G_{q}}\sum_{\text{%
C}_{j}\in \mathbb{G}_{q}}\sum_{\text{C}_{i}\in \mathbb{G}_{0}}g(\text{C}%
_{j})g(\text{C}_{i})R^{\prime \prime }(\text{C}_{j}\mathbf{,}\text{C}_{i})
\label{A1_RRG}
\end{equation}%
but with the use of the secondary co-ranking function $R^{\prime \prime }($A$%
\mathbf{,}$B$),$ which is introduced on the basis of the \textit{secondary
preference} (denoted by $\prec ^{\prime \prime }$) and determined by the
primary current ranking: 
\begin{equation}
R^{\prime \prime }(\text{A}\mathbf{,}\text{B})=\left\{ 
\begin{array}{ccc}
+1, & R^{\ast }(\text{A})>R^{\ast }(\text{B}) & \text{\textit{i.e.}, A}\succ ^{\prime
\prime }\text{B} \\ 
0, & R^{\ast }(\text{A})=R^{\ast }(\text{B}) & \text{\textit{i.e.}, A}\sim ^{\prime
\prime }\text{B} \\ 
-1, & R^{\ast }(\text{A})<R^{\ast }(\text{B}) & \text{\textit{i.e.}, A}\prec ^{\prime
\prime }\text{B}%
\end{array}%
\right.  \label{A3_RAB2}
\end{equation}

\subsection{Properties of Current Rankings}

The current ranking is a special case of conditional indicator-ranking, \textit{i.e.}, 
$R^{\ast }($A$)$ is $R_{_{\mathbb{G}_{q}}}($A$)$ with set $\mathbb{G}_{q}$
expanded to the whole system $\mathbb{G}_{0}$. The following propositions
characterise properties of

\begin{proposition}
\label{PA1b}If the primary preference is currently transitive, the secondary
preference is equivalent to the primary preference: A$\preceq^{\prime\prime}$%
B $\Longleftrightarrow$ $R^{\ast}($A$)\leq R^{\ast}($B$)$ $%
\Longleftrightarrow$ A$\preceq$B for any A and B.
\end{proposition}

The proof is similar to the proof of Proposition \ref{PA1a}, but the case of 
$R^{\ast}($A$)=R^{\ast}($B$)$ (\textit{i.e.}, A$\sim^{\prime\prime}$B) is now
impossible when A$\prec$B. Indeed, all elements are presumed to be connected
and present in the reference set in definition of the current ranking.
Hence, the terms $R($A$\mathbf{,}$B$)=-R($B$\mathbf{,}$A$)<0$ are present in
the sums Equation~(\ref{A1_R}) evaluated for $R^{\ast}($A$)$ while including C$_{i}=$%
B, and $R^{\ast}($B$)$ while including C$_{i}=$A$.$

\begin{proposition}
\label{PA1c}If B is transitively preferred to A in a generally intransitive
preference, then the preference of B over A is preserved by the current
ranking: A$\prec\prec$B$\;\Longrightarrow R^{\ast}($A$)<R^{\ast}($B$%
)\Longleftrightarrow$A$\prec^{\prime\prime}$B
\end{proposition}

Although this proposition refers to a transitive preference A$\prec\prec$B
in a preference $\prec$ that is generally intransitive and is thus different
from the previous statements, its proof is similar.\ The relation A$%
\prec\prec$B requires that A$\prec$B and there are no C$_{i}$ that satisfies
C$_{i}\preceq $A$\prec$B$\preceq$C$_{i}$. Hence, $R($A$\mathbf{,}$C$%
_{i})\leq R($B$\mathbf{,}$C$_{i})$ and $R($A$\mathbf{,}$B$)=-R($B$\mathbf{,}$%
A$)<0$ so that $R^{\ast}($A$)<R^{\ast}($B$)$ as defined by Equation~(\ref{A1_R}).
Note that the inverse statements $R^{\ast}($A$)<R^{\ast}($B$%
)\Longrightarrow\;$A$\prec\prec $B and $R^{\ast}($A$)<R^{\ast}($B$%
)\Longrightarrow\;$A$\prec$B are incorrect in intransitive~systems.

Since the preference specified by the primary current ranking is transitive,
the preferences specified by the primary and secondary current rankings must
be equivalent according to Proposition \ref{PA1b}: 
\begin{equation}
R^{\ast }(\text{A})<R^{\ast }(\text{B})\Longleftrightarrow R^{^{\ast \ast }}(%
\text{A})<R^{^{\ast \ast }}(\text{B})
\end{equation}%

Here, Proposition \ref{PA1b} is applied to secondary (as primary) and
tertiary (as secondary) preferences. The conclusion is that there is no
independent tertiary preferences as they coincide with the secondary
preferences. Equivalence of the primary and secondary rankings, however,
does not imply that these rankings are identical: generally $R^{\ast }($A$%
)\neq R^{^{\ast \ast }}($A$)$. Deviations of secondary ranking from primary
ranking indicate intransitivity in evolutions of competitive systems as
determined by the following~theorem:

\begin{theorem}
\label{PA2}The following statements are correct for a system $\mathbb{G}_{0}$
of connected elements (\textit{i.e.}, \mbox{$g($C$_{i})>0$} for any C$_{i}\in \mathbb{G}_{0}$%
):

%TCIMACRO{%
%\TeXButton{TeX field =============== myTH}{\begin{enumerate}[a)]
%
%\item If the primary preference is currently transitive, the primary and
%secondary current rankings coincide (\textit{i.e.}, $R^{\ast}($C$_{i})=R^{^{\ast\ast}}( $C$_{i})$ for all C$_{i}\in\mathbb{G}_0$).\vspace{-0.25cm}
%
%\item If the primary and secondary rankings coincide (\textit{i.e.}, $R^{\ast}($C$_{i})=R^{^{\ast\ast}}($C$_{i})$ for all elements C$_{i}\in\mathbb{G}_0$), then
%the secondary preference is a coarsening of the primary preference (\textit{i.e.}, C$_{i}\prec^{\prime\prime}$C$_{j}$ $\Longrightarrow$ C$_{i}\prec$C$_{j}$). 
%\vspace{-0.25cm}
%
%\item In particular, if the\ primary and secondary rankings coincide\ and
%are strict (\textit{i.e.}, current rankings of different elements are different: $R^{\ast}($C$_{i})\neq R^{\ast}($C$_{i})$ for any C$_{i}\nsim$C$_{j}$), then
%the primary preference is currently transitive.
%\end{enumerate}}}%
%BeginExpansion
\begin{enumerate}

\item[(a)] If the primary preference is currently transitive, the primary and
secondary current rankings coincide (\textit{i.e.}, $R^{\ast}($C$_{i})=R^{^{\ast\ast}}( $C$_{i})$ for all C$_{i}\in\mathbb{G}_0$).

\item[(b)] If the primary and secondary rankings coincide (\textit{i.e.}, $R^{\ast}($C$_{i})=R^{^{\ast\ast}}($C$_{i})$ for all elements C$_{i}\in\mathbb{G}_0$), then
the secondary preference is a coarsening of the primary preference (\textit{i.e.}, C$_{i}\prec^{\prime\prime}$C$_{j}$ $\Longrightarrow$ C$_{i}\prec$C$_{j}$).

\item[(c)] In particular, if the\ primary and secondary rankings coincide\ and
are strict (\textit{i.e.}, current rankings of different elements are different: $R^{\ast}($C$_{i})\neq R^{\ast}($C$_{i})$ for any C$_{i}\nsim$C$_{j}$), then
the primary preference is currently transitive.
\end{enumerate}%
%EndExpansion
\end{theorem}

\begin{proof}
Statement (a) immediately follows from the equivalence of primary and
secondary preferences, as stated in Proposition \ref{PA1b}, while statements
(b) and (c) require detailed consideration.

Let us sort all elements into $k$ sets of decreasing primary ranking $%
R_{1}^{\ast }>R_{2}^{\ast }>...>R_{k}^{\ast }$ ; \textit{i.e.},$\ q^{\text{th}}$ set $%
\mathbb{G}_{q}$ contains $n_{q}\geq 1$ elements that have primary current
ranking $R_{q}^{\ast }.$ The elements C$_{1},...,$C$_{i},...,$C$_{n}$ are
thus ordered according to decreasing primary ranking. Figure \ref{fig8}\
shows the structure of the \mbox{$n\times n$} matrices $R_{ij}=R($C$_{i},$C$_{j})$
and $R_{ij}^{\prime \prime }=R^{\prime \prime }($C$_{i},$C$_{j}),$ which
correspond to the primary and secondary \mbox{co-rankings.} The co-ranking of
different sets is denoted by $\bar{R}_{qp}=\bar{R}(\mathbb{G}_{q},\mathbb{G}%
_{p})$ for primary preferences and by $\bar{R}_{qp}^{\prime \prime }=\bar{R}%
^{\prime \prime }(\mathbb{G}_{q},\mathbb{G}_{p})$ for the secondary
preferences. The average primary current ranking of a set is denoted by $%
\bar{R}_{q}^{\ast }=\bar{R}^{\ast }(\mathbb{G}_{q})=\bar{R}(\mathbb{G}_{q},%
\mathbb{G}_{0}),$ while the average secondary current ranking is denoted by $%
\bar{R}_{q}^{^{\ast \ast }}=\bar{R}^{^{\ast \ast }}(\mathbb{G}_{q})=\bar{R}%
^{\prime \prime }(\mathbb{G}_{q},\mathbb{G}_{0})$. These quantities are
specified by Equations~(\ref{r_R_G}), (\ref{r_R_GG}) and (\ref{A1_R})--(\ref%
{A1_RRG}). Obviously $\bar{R}_{q}^{\ast }=\bar{R}_{q}^{^{\ast \ast }}$ for
all $q$ due to equivalence of the primary and secondary current rankings as
stated in the theorem (\textit{i.e.}, $R^{\ast }($C$_{i}$)$=R^{^{\ast \ast }}($C$_{i}$%
) for all C$_{i}\in \mathbb{G}_{0}$).
%===============  12
\begin{figure}[H]
\centering
\includegraphics[width=6.5cm,page=7, clip ]{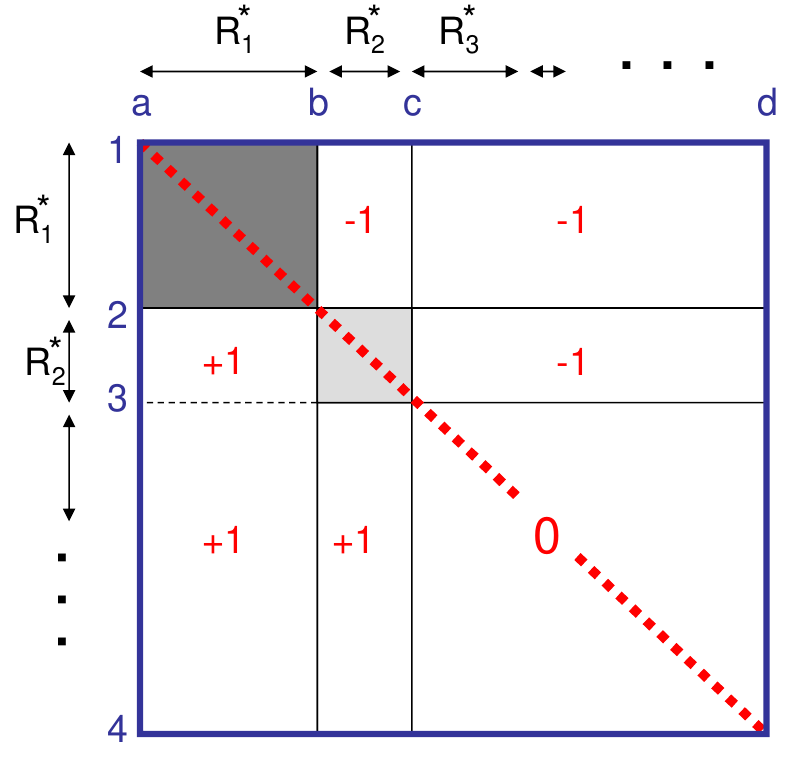}
\caption{Co-ranking matrix with elements ordered according to their current rankings.}
\label{fig8}

\end{figure}
The first (leading) set $\mathbb{G}_{1}$ is considered first. The
co-rankings $\bar{R}(\mathbb{G}_{1},\mathbb{G}_{1})$ and $\bar{R}^{\prime
\prime }(\mathbb{G}_{1},\mathbb{G}_{1}),$ which are determined by the sums
over the dark square region a1-b1-b2-a2 (Figure \ref{fig8}), are zeros as
the matrices $R_{ij}$ and $R_{ij}^{\prime \prime }$ are antisymmetric. Since 
$\bar{R}_{1}^{\ast }=\bar{R}(\mathbb{G}_{1},\mathbb{G}_{0})$ is the same as $%
\bar{R}_{1}^{^{\ast \ast }}=\bar{R}^{\prime \prime }(\mathbb{G}_{1},\mathbb{G%
}_{0})$ while $\bar{R}(\mathbb{G}_{1},\mathbb{G}_{1})=\bar{R}^{\prime \prime
}(\mathbb{G}_{1},\mathbb{G}_{1})=0,$ the co-rankings\ $\bar{R}(\mathbb{G}%
_{1},\mathbb{G}_{0}\mathbb{-G}_{1})$ and $\bar{R}^{\prime \prime }(\mathbb{G}%
_{1},\mathbb{G}_{0}\mathbb{-G}_{1})$ evaluated in terms the corresponding
sums over the rectangle a2-b2-b4-a4 must be the same. Hence $R_{ij}=1$ in
this rectangle since in any other case the sums $\bar{R}(\mathbb{G}_{1},%
\mathbb{G}_{0}\mathbb{-G}_{1})$ cannot coincide with $\bar{R}^{\prime \prime
}(\mathbb{G}_{1},\mathbb{G}_{0}\mathbb{-G}_{1})$.

Since the co-ranking matrices are antisymmetric, $R_{ij}=-1$ and $%
R_{ij}^{\prime \prime }=-1$ in the rectangle b1-c1-c2-b2. We take into
account that $\bar{R}(\mathbb{G}_{2},\mathbb{G}_{2})=\bar{R}^{\prime \prime
}(\mathbb{G}_{2},\mathbb{G}_{2})=0$ (the sums over the dark squares are
zeros) and reiterate our previous consideration for rectangle b3-c3-c4-b4,
where $R_{ij}^{\prime \prime }=1$ and the sums $\bar{R}(\mathbb{G}_{2},%
\mathbb{G}_{0}\mathbb{-G}_{1}-\mathbb{G}_{2})$ and $\bar{R}^{\prime \prime }(%
\mathbb{G}_{2},\mathbb{G}_{0}\mathbb{-G}_{1}-\mathbb{G}_{2})$ must be the
same. Hence, $R_{ij}=1$ in this rectangle. Continuing this consideration for
the remaining sets $q=3,4...,k$ proves that $R_{ij}=R_{ij}^{\prime \prime }$
provided $i$ and $j$ belong to different sets.

If $i$ and $j$ belong to the same set, then $R_{ij}^{\prime \prime }=0$
according to Equation~(\ref{A3_RAB2}) and either $R_{ij}=0$ or competition is
intransitive within the set (if the primary preferences within the set $q$
were transitive then, according to Proposition \ref{PA1c}, C$_{i}\succ $C$%
_{j}$ demands $R^{\ast }($C$_{i})>R^{\ast }($C$_{j}),$ which contradicts C$%
_{i},$C$_{j}\in \mathbb{G}_{q}$). The secondary preference represents a
coarsening of the primary preference since $R_{ij}^{\prime \prime }>0$
demands $R_{ij}>0$ (when $j$ and $i$ belong to different sets), while $%
R_{ij}^{\prime \prime }=0$ may correspond to $R_{ij}>0,$ $R_{ij}=0$ or $%
R_{ij}<0$ (when $j$ and $i$ belong to a common set).\ 

If current rankings of all elements are different, that is all $n_{q}=1$ for
all $q,$ then $R_{ij}=R_{ij}^{\prime\prime}$ since all $i\neq j$ always
belong to different sets. This means that the primary preference coincides
with the secondary preference based on primary current ranking and is
transitive.
\end{proof}

\subsection{Maps of Current Rankings}

The statements proven in the previous subsection suggest a relatively simple
method of analysing intransitivity in large systems. This method, which is
based on ranking maps (\textit{i.e.}, plots of primary current ranking $R^{\ast}$ of
the elements against their secondary current ranking $R^{^{\ast\ast}}),$
indicates the presence, intensity, extent and localisation of intransitivity
by deviations from the line $R^{\ast}=R^{^{\ast\ast}}$. We wish to stay
under the conditions of Theorem \ref{PA2}(c), and avoid complexities related
to statement b of this theorem since, in this case, equivalence between the
primary and secondary rankings ensures transitivity. However, in the case of 
$g($C$_{i})=1$, which perhaps is most common in practice, coincidences of
primary rankings for different elements $R^{\ast}($C$_{i})=R^{\ast}($C$_{j})$
for $i\neq j$ are likely due to the limited number of values spaced by $1/n$
that these rankings can take. The practical solution for this problems is
simple---consider $g($C$_{i})=(1+\varepsilon_{i})g_{0}($C$_{i}),$ where $%
\varepsilon_{i}$ represent small random values and $g_{0}($C$_{i})$ are the
original weights. Presence of these small values does not significantly
alter the maps but makes coincidences $R^{\ast}($C$_{i})=R^{\ast}($C$_{j})$
impossible unless all properties of the elements C$_{i}$ and C$_{j}$ are
identical, which is sufficient for our purposes.

Two (or more) sets (or groups) are said to be \textit{subject to a preference%
} when all possible selections of elements from these sets are compliant
with the preference.\ For example $\mathbb{G}_{q}\succ \succ \mathbb{G}_{p}$
implies that sets $\mathbb{G}_{q}$and $\mathbb{G}_{p}$ are subject to
preference $\succ \succ $ so that C$_{i(q)}\succ \succ $C$_{j(p)}$ for any C$%
_{i(q)}\in \mathbb{G}_{q}$ and any C$_{j(p)}\in \mathbb{G}_{p}.$ According
to these notations, the index $i(q)$ runs over all elements of set $\mathbb{G%
}_{q}$. In this subsection, we consider partition of all elements in the
system into \textit{range sets,} where each set is represented by a range of
primary current ranking (and consequently by a range of the secondary
current ranking). The range sets are non-overlapping and jointly cover all
elements.

We now turn to consideration of the $R^{\ast }$ versus $R^{^{\ast \ast }}$
maps. Consider an intransitive preference and its transitive closure. In
this closure, the elements C$_{1},...,$C$_{n}$ are divided into $k$
transitively ordered sets $\mathbb{G}_{1}\succ \succ \mathbb{G}_{2}\succ
\succ ...\succ \succ \mathbb{G}_{k}$ of elements that are transitively
equivalent within each set.\ That is for any C$_{i(q)}\in \mathbb{G}_{q}$
and any C$_{j(p)}\in \mathbb{G}_{p}$%
\begin{equation}
\begin{tabular}{ll}
$\text{C}_{i(q)}\approx \text{C}_{j(p)}$ & $\text{iff }q=p$ \\ 
$\text{C}_{i(q)}\succ \succ \text{C}_{j(p)}$ & $\text{iff }q<p$%
\end{tabular}%
\end{equation}%

Proposition \ref{PA1c} indicates that $R^{\ast }($C$_{i(q)})>R^{\ast }($C$%
_{j(p)})$ and $R^{^{\ast \ast }}($C$_{i(q)})>R^{^{\ast \ast }}($C$_{j(p)})$
when $q<p$, that is the ranges of current rankings of different sets do not
overlap. Hence sets $\mathbb{G}_{1},...,\mathbb{G}_{k}$ represent a set of
range sets. This implies equivalence of primary and secondary set
co-rankings 
\begin{equation}
\bar{R}(\mathbb{G}_{p},\mathbb{G}_{q})=\bar{R}^{\prime \prime }(\mathbb{G}%
_{p},\mathbb{G}_{q})
\end{equation}%
for all $p$ and $q$. Note that $\bar{R}(\mathbb{G}_{q},\mathbb{G}_{q})=\bar{R%
}^{\prime \prime }(\mathbb{G}_{q},\mathbb{G}_{q})=0$ for any $q$. This also
implies that 
\begin{equation}
\bar{R}^{\ast }(\mathbb{G}_{q})=\bar{R}^{^{\ast \ast }}(\mathbb{G}_{q})
\label{A_GG}
\end{equation}%
for any $q$. The property expressed by Equation~(\ref{A_GG}) is reflected in the
following proposition

\begin{proposition}
The primary and secondary average current rankings of range sets coincide if
and only if these sets are subject to a transitive primary preference
(although this preference may remain intransitive within each set).
\end{proposition}

First we note that current primary set rankings of different range sets
cannot coincide. The rest of the proof is similar to that of Theorem \ref%
{PA2}, where the leading set $\mathbb{G}_{1}$ is considered first. Since C$%
_{i(1)}\succ ^{\prime \prime }$C$_{i(q)}$ for all $q>1$, \ the equality $%
\bar{R}^{\ast }(\mathbb{G}_{1})=\bar{R}^{^{\ast \ast }}(\mathbb{G}_{1})$ is
achieved if and only if C$_{i(1)}\succ $C$_{i(q)}$ for all $q>1$. After
applying this consideration sequentially to sets $\mathbb{G}_{2},...,\mathbb{%
G}_{k}$, we conclude that C$_{i(p)}\succ $C$_{i(q)}$ for $p<q.$ Finally we
note that $\mathbb{G}_{1}\succ \mathbb{G}_{2}\succ ...\succ \mathbb{G}_{k}$
requires $\mathbb{G}_{1}\succ \succ \mathbb{G}_{2}\succ \succ ...\succ \succ 
\mathbb{G}_{k}$ since the preference $\succ ^{\prime }$ defined by C$%
_{i(p)}\succ ^{\prime }$C$_{i(q)}$ for $p<q$ and C$_{i(p)}\sim ^{\prime }$C$%
_{i(q)}$ for $p=q$ is transitive and is a coarsening of the primary
preference and therefore is a coarsening of its transitive closure.

Figure \ref{fig10} demonstrates a possible structure of the ranking map,
where primary current ranking is plotted versus secondary current ranking.
The elements are ordered according to their current rankings. The map in
Figure \ref{fig10} indicates that the preference is generally intransitive
(since the map deviates from the line specified by $R^{\ast }=R^{^{\ast \ast
}}$). The range sets, shown in the figure, are transitively ordered so that $%
\mathbb{G}_{1}\succ \succ \mathbb{G}_{2}\succ \succ \mathbb{G}_{3}\succ
\succ \mathbb{G}_{4}\succ \succ \mathbb{G}_{5}$. The large dots indicate
average set ranking, which is compliant with Equation~(\ref{A_GG}). The preferences
are transitive within $\mathbb{G}_{3}$ and $\mathbb{G}_{5}$ and intransitive
within $\mathbb{G}_{1},$ $\mathbb{G}_{2}$ and $\mathbb{G}_{4}$. Small
deviation from the line specified by $R^{\ast }=R^{^{\ast \ast }}$ within
set $\mathbb{G}_{4}$ indicate that intransitivity is present but not
frequent within this set. Two subsets $\mathbb{G}_{1a}$ and $\mathbb{G}_{1b}$
are distinguished within set $\mathbb{G}_{1}$. The preferences between these
sets are close to be transitive but some intransitive interference between
subsets is present as indicated by angle $\gamma >0$. The small dots show
current set rankings of $\mathbb{G}_{1a}$ and $\mathbb{G}_{1b},$ which are
not compliant with Equation~(\ref{A_GG}).
%================  13

\begin{figure}[H]
\centering
\includegraphics[width=9cm,page=9, clip ]{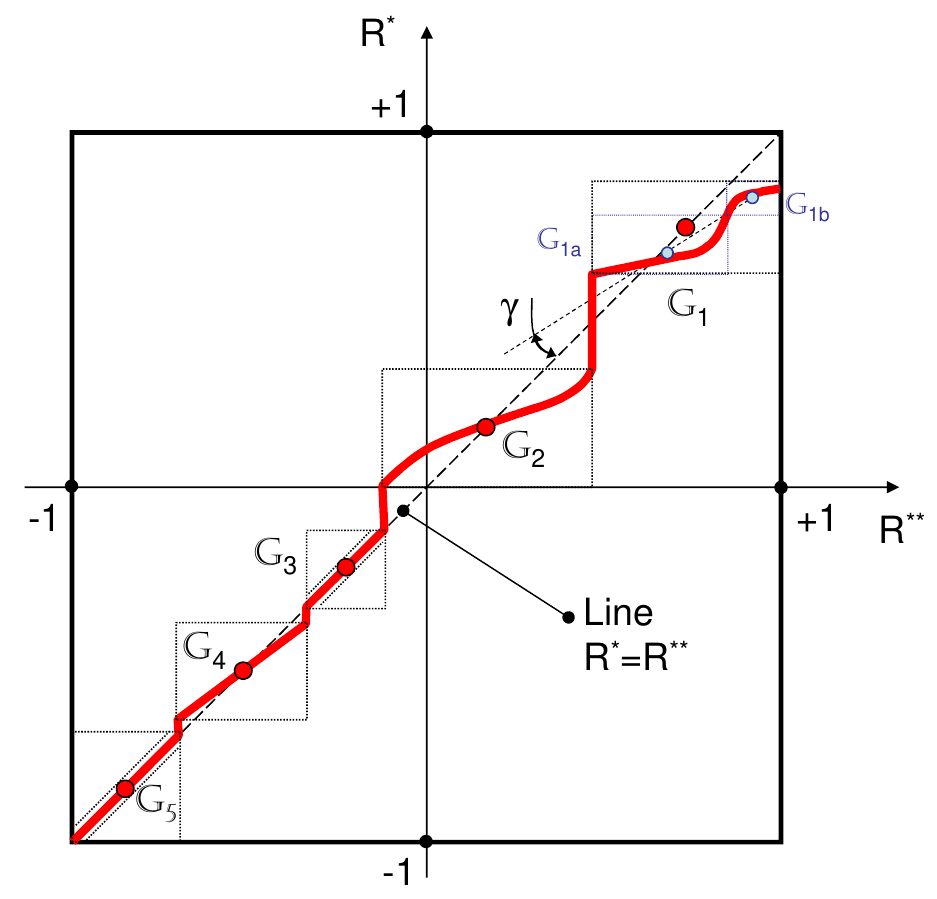}
\caption{Current ranking map (the thick red line shows primary vs secondary current ranking). 
The map is divided into range sets (groups) $\mathbb{G}_1,...,\mathbb{G}_5$ (group current 
rankings of the sets are shown by large dots) and subsets $\mathbb{G}_{1a}$ and $\mathbb{G}_{1b}$ (group current 
rankings of the subsets are shown by small dots). The black dashed line corresponds to $R^{*}=R^{^{**}}$. }
\label{fig10}

\end{figure}
\section{Evolutionary Intransitivity \label{AA3}}

\setcounter{figure}{0}
\renewcommand\thefigure{D\arabic{figure}}
\setcounter{equation}{0}
\renewcommand\theequation{D\arabic{equation}}
\setcounter{proposition}{0}
\renewcommand\theproposition{D\arabic{proposition}}
\setcounter{theorem}{0}
\renewcommand\thetheorem{D\arabic{theorem}}

This section considers characterisation of intransitivity in competitive
evolving systems (\textit{i.e.}, systems developing in time due to competition between
the elements \cite{K-PS2012,K-PT2013,K_Ent2014a}). The weight $g($C$_{i})>0$
is proportional to the probability of selecting element C$_{i}$ for
competition.

\subsection{Competitive Evolution}

The use of current rankings can be illustrated by evolution of competitive
systems, which is considered in this subsection. Consider a system of $n$
elements (particles) that evolve in time due to competition between these
elements and, consequently, have a preference defined for these elements
(which is referred to as the \textit{primary preference}). The focus of the
present consideration is competition steps (not involving mutations). In
this case evolution of the system due to competition is determined by the
current ranking according to the equation \cite{K-PS2012}%
\begin{equation}
\frac{df(\text{A})}{dt}=\lambda R^{\ast}(\text{A})f(\text{A})
\label{A1_dfdt}
\end{equation}
where $f($A$)$ is probability of having an element at location A, $\lambda$
is a constant depending on the time step and $R^{\ast}($A$)$ is the primary
current ranking defined with the weight $g($B$)=f($B$)\psi($B$)$, which is
proportional to the probability $\psi($B$)$ of selecting element B for
competitive mixing. 
(Generally, this selection weight can depend on both A and B, that is $%
\psi=\psi($A$,$B$)$ as considered in~\cite{K-PS2012} but this is not
needed in the present work). Since selection of elements for mixing is
stochastic, the process of competition requires a probabilistic description
(the fluctuations with respect to the averages are not considered). This is
achieved with the use of probability $f,$\ which can be initially set
according to the location of $n$ stochastic particles, which numerically
represent $f$ by a large set of delta-functions$,$ before the competition
steps. Mutation step, which is not considered here, would make this
distribution continuous. As noted above, isolated elements $g(...)=0$ are
not of interest and all elements are presumed to be connected by competition.

Let $\mathbb{G}_{t}$ denote the state of the system at time $t$ so that the
subscript \textquotedblleft$t$\textquotedblright\ relates quantities to this
state. The following proposition indicates the direction of the competition
steps:

\begin{proposition}
\label{PA0} Competition (transitive or intransitive) results in improved
competitiveness with respect to the current distribution $f$, that is for $%
t_{2}=t_{1}+dt$%
\begin{equation}
\bar{R}(\mathbb{G}_{t_{2}},\mathbb{G}_{t_{1}})=\frac{1}{G}\sum_{\text{C}%
_{i}}f_{t_{2}}(\text{C}_{i})\psi(\text{C}_{i})R_{t_{1}}^{\ast}(\text{C}%
_{i})\geq0  \label{A1_Rav}
\end{equation}
\end{proposition}

First, we note that $\bar{R}(\mathbb{G}_{t_{1}},\mathbb{G}_{t_{1}})=0$ due
to the antisymmetric properties of co-ranking $R($A$\mathbf{,}$C$)$. The
substitution of $f_{t_{2}}=f_{t_{1}}+df,$ where $df$ is determined from Equation~(\ref%
{A1_dfdt}), into Equation~(\ref{A1_Rav}) results in 
\begin{equation}
\bar{R}(\mathbb{G}_{t_{2}},\mathbb{G}_{t_{1}})=dt\frac{\lambda }{G}\sum_{%
\text{C}_{i}}f_{t_{1}}(\text{C}_{i})\psi ^{2}(\text{C}_{i})\left(
R_{t_{1}}^{\ast }(\text{C}_{i})\right) ^{2}+...\geq 0
\end{equation}%
since $f$, $\psi ^{2}$ and $(R^{\ast })^{2}$ are non-negative. Fluctuations
that can be present in the system due to random particle selection are not
considered here. Note that, in intransitive systems, competitiveness
increases only with respect to the current distribution and improvement is
not necessarily achieved when viewed from a different perspective, that is $%
\bar{R}(\mathbb{G}_{t},\mathbb{G}_{r})$ may decrease in time for some
reference groups $\mathbb{G}_{r}$. It is common to have a view attached to
current distribution---this, according to Proposition \ref{PA0}, produces
the impression that competitiveness is improved by the competition but, in
intransitive systems, the actual result might be different if viewed from a
different perspective (\textit{i.e.}, using a different reference set). The effect of
mutation steps (which is not considered here) on competitiveness is
typically negative.

\subsection{Evolutionary Intransitivity in Simulations of the Risk/benefit
Dilemma}

The standard deviation of the primary and secondary rankings 
\begin{equation}
\Omega=\left( \sum_{\text{C}_{i}}f(\text{C}_{i})\psi(\text{C}_{i})\left(
R^{\ast}(\text{C}_{i})-R^{^{\ast\ast}}(\text{C}_{i})\right) ^{2}\right)
^{1/2}
\end{equation}
is a measure of deviation from transitivity in evolution of competitive
systems specified by Equation~(\ref{A1_Rav}). Parameter $\Omega$ is thus referred to
as \textit{evolutionary intransitivity parameter; }$\Omega$ is plotted in
Figure \ref{fig2m} by dashed lines. In intransitive simulations, $\Omega$ is
small for most of the time and the evolution is quite close to being
transitive (see Figure \ref{fig3m}a--c). If $\Omega$ is close to
zero, the underlying preference is, effectively, currently transitive. These
rules Equation~(\ref{s_rho0}), however, are very different from the transitive rules
given by Equation~(\ref{s_rho0}) as shown in Figure \ref{fig3m}d--f. When the
system approaches a point of maximal benefit, the evolutionary
intransitivity parameter increases, indicating the approaching collapse of
the high-benefit/high risk state. Thus evolution of the system may seem
transitive most of the time but detection of increasing intransitivity warns
of imminent catastrophe resulting in collapse of the benefit.

\section{Summary of the Terms Characterising Intransitivity \label{AA4}}
\setcounter{figure}{0}
\renewcommand\thefigure{E\arabic{figure}}
\setcounter{equation}{0}
\renewcommand\theequation{E\arabic{equation}}
\setcounter{proposition}{0}
\renewcommand\theproposition{E\arabic{proposition}}
\setcounter{theorem}{0}
\renewcommand\thetheorem{E\arabic{theorem}}

We distinguish the following types of intransitivity:

\begin{itemize}
\item By strictness of the preference: 

\begin{itemize}
\item \textit{strict intransitivity}: there exists A$\succ$B$\succ$C$\succ $A%

\item \textit{semi-strict intransitivity}: there exists A$\succ$C$%
_{1}\succ...\succ$C$_{k}\succ$A but intransitivity is not strict

\item \textit{semi-weak intransitivity:} there exists A$\succ$B$\sim$C$\succ$%
A but intransitivity is not semi-strict

\item \textit{weak intransitivity:} there exists A$\sim$B$\sim$C$\succ$A but
intransitivity is not semi-weak
\end{itemize}

\item By localisation: 

\begin{itemize}
\item \textit{local intransitivity} (combined with global transitivity)

\item \textit{global intransitivity} (combined with local transitivity)

\item \textit{strong intransitivity} (strict and both local and global)
\end{itemize}

\item By explicit presence: 

\begin{itemize}
\item \textit{absolute transitivity}: any kind of intransitivity is
impossible under the given preference rules, implying existence of absolute
utility or ranking 

\item \textit{potential intransitivity}: intransitivity does not necessarily
show on the current set of elements but may appear when conditions are
changed

\item \textit{current intransitivity or current transitivity}: indicate the
properties of preferences on the current set of elements

\item \textit{near-transitive evolution}: evolution of an intransitive
system that, within a fixed time interval, can be reasonably approximated by
evolution of a transitive system
\end{itemize}
\end{itemize}

Weak and semi-weak intransitivity is not a prominent type of intransitivity
since the preference in this case can be represented as a coarsening of a
transitive preference (this representation can be achieved by transitively
extending preference ``$\succ$" to some of pairs that are originally
specified as equivalent). A~preference with strict or semi-strict
intransitivity can not be represented as a coarsening of any transitive~preference.

\section{Note on Quantum Preferences \label{AQP}}
\setcounter{figure}{0}
\renewcommand\thefigure{F\arabic{figure}}
\setcounter{equation}{0}
\renewcommand\theequation{F\arabic{equation}}
\setcounter{proposition}{0}
\renewcommand\theproposition{F\arabic{proposition}}
\setcounter{theorem}{0}
\renewcommand\thetheorem{F\arabic{theorem}}

Quantum games and quantum mechanisms of preferences and decision-making are
notable extensions of the respective classical approaches that have been
repeatedly researched in the recent years \cite%
{Qgame1999,Qcat2009M,Ent2015MPS}.

\subsection{Quantum Preferences and Co-Rankings}

While conventional strict preference implies that either A preferred over B
or vice versa, quantum preferences allow for A$\succ $B and A$\prec $B to be
valid at the same time. Here we refer to the superposition states which have
the quantum wave function $\Psi $ and the quantum density matrix $\mathbf{%
\hat{\rho}}$ specified by 
\begin{equation}
\left\vert \Psi \right\rangle =c_{1}\left\vert \text{A}\succ \text{B}%
\right\rangle +c_{2}\left\vert \text{B}\succ \text{A}\right\rangle ,\ \ \ \ 
\mathbf{\hat{\rho}}=\left[ 
\begin{array}{cc}
\bar{c}_{1}c_{1} & \bar{c}_{2}c_{1} \\ 
\bar{c}_{1}c_{2} & \bar{c}_{2}c_{2}%
\end{array}%
\right]  \label{QQS}
\end{equation}%
or to the mixed states with the wave function $\Psi $ and the quantum
density matrix $\mathbf{\hat{\rho}}$ are taking the forms%
\begin{equation}
\left\vert \Psi \right\rangle =c_{1}\left\vert \text{A}\succ \text{B}%
\right\rangle \left\vert \theta _{1}\right] +c_{1}\left\vert \text{B}\succ 
\text{A}\right\rangle \left\vert \theta _{2}\right] ,\ \ \ \ \mathbf{\hat{%
\rho}}=\left[ 
\begin{array}{cc}
\bar{c}_{1}c_{1} & 0 \\ 
0 & \bar{c}_{2}c_{2}%
\end{array}%
\right]  \label{QQM}
\end{equation}%

Here, we use the conventional quantum ket notations $\left\vert
...\right\rangle $ and imply normalisation condition \mbox{$\bar{c}_{1}c_{1}+\bar{c%
}_{2}c_{2}=1.$} The symbol $\mathbf{\hat{\rho}}$\ used here is not related to
co-ranking $\rho $. In this section, the overbar denotes the complex
conjugates. The strict quantum preferences can be treated as qubits.\ 

In order to distinguish superposition and mixed states in terms of wave
functions, we deploy random phases $\theta _{j}=e^{-i\omega _{j}}$ where $%
\omega _{j}$ ($j=1,...$) are random angles, which are uniformly and
independently distributed between $-\pi $ and $\pi $. A mixture of quantum
states is interpreted here as an entanglement with special quantum states $%
\left\vert \theta _{j}\right] $ representing the random phases. These states
have only one operation applicable to these states and resulting in a
physically measurable quantities---the scalar product: 
\begin{equation}
\left[ \theta _{i}\middle|\theta _{j}\right] =\left\langle e^{i\left( \omega
_{i}-\omega _{j}\right) }\right\rangle =\left\{ 
\begin{array}{cc}
1, & i=j \\ 
0, & i\neq j%
\end{array}%
\right\vert  \label{QQph}
\end{equation}%

The random phase notations are generally equivalent to the corresponding
notations using density matrices but random phases can be convenient for
specification of mixed states when using wave functions is preferred over
deploying density matrices. 
The random phases is only a convenient notation but not a physical theory
that answers the question whether a larger system (or the whole Universe) is
physically in a superposition state or a mixed state. The random phases do
not have energy and do not evolve in time (in unitary quantum mechanics).
The notations based on wave functions are often more explicit and
transparent than the equivalent notations based on density matrices (see Appendix of~\cite{Kaon2014K} for
details and discussion). 

In the case of superposition state, both the relative magnitudes $\left\vert
c_{2}\right\vert /\left\vert c_{1}\right\vert $ and phases $\arg
(c_{2})-\arg (c_{1})$ of complex amplitudes $c_{1}$ and $c_{2}$ affect the
quantum state. In the case of a mixed state, the phases are not important as
they are randomised by $\theta _{1}$ and $\theta _{2}$ but the relative
magnitudes $\left\vert c_{2}\right\vert /\left\vert c_{1}\right\vert $
determine the state: the effect of a mixed state is similar to classical
probabilities $P_{2}=\left\vert c_{2}\right\vert ^{2}$ and $P_{1}=\left\vert
c_{1}\right\vert ^{2}$, $P_{1}+P_{2}=1$.

The physical magnitude of the preference is determined by the Hermitian
preference operator $\mathbb{R}$, which is specified by the matrix $\hat{R}%
_{ij}=\left\langle i\middle|\mathbb{R}\middle|j\right\rangle $ taking the
form%
\begin{equation}
\mathbf{\hat{R}}=\left[ 
\begin{array}{cc}
+1 & 0 \\ 
0 & -1%
\end{array}%
\right]
\end{equation}%
This operator corresponds to the indicator co-ranking $R$ in conventional
preferences. For both superposition state Equation~(\ref{QQS}) and mixed state Equation~(\ref%
{QQM}), the magnitude of the preference is given by 
\begin{equation}
R=\left\langle \Psi \middle|\mathbb{R}\middle|\Psi \right\rangle =\left\vert
c_{1}\right\vert ^{2}-\left\vert c_{2}\right\vert ^{2}
\end{equation}

The quantum preference models the state of mind before a decision of
selecting A or B from the set of \{A,B\} takes place. The instance of
decision is marked by the collapse of the wave function converting the
system in one of the states, $\left\vert \text{A}\succ \text{B}\right\rangle 
$ or $\left\vert \text{B}\succ \text{A}\right\rangle $ with respective
probabilities $P_{1}=\left\vert c_{1}\right\vert ^{2}$ and $P_{2}=\left\vert
c_{2}\right\vert ^{2}.$ These probabilities are the same for superposition
states Equation~(\ref{QQS}) and mixed states Equation~(\ref{QQM}).

Demonstrating existence of quantum mechanisms in cognitive decision making
is closely linked to the fundamental possibility of distinguishing
superposition states and mixed states as various features of mixed states
can be interpreted in terms of classical probabilities and, generally, point
to probabilistic but not necessarily quantum effects in evaluation of
preferences by the human brain. Only quantum superpositions possess all
unique flavours of quantum mechanics. Distinguishing quantum superpositions
and quantum mixtures generally requires a measurable value $H$ that
corresponds to a Hermitian operator $\mathbb{H}$ with non-zero off-diagonal
elements. The same effect can be achieved by evolving the preference state
with a unitary operator $\mathbb{U}$ provided this operator is not trivial
and corresponds to a physically observable process. The evolved
state $\mathbb{U}^{\dag }\mathbb{RU}$ of the measurable preference operator $%
\mathbb{R}$ would generally acquire off-diagonal components. For mixed state
Equation~(\ref{QQM}), the value of $H$ is, as in the case of classical probabilities,
the sum of conditional values $\left\langle \text{A}\middle|\mathbb{H}\middle%
|\text{A}\right\rangle $ and $\left\langle \text{B}\middle|\mathbb{H}\middle|%
\text{B}\right\rangle $ multiplied by the respective probabilities $P_{1}$
and $P_{2}$. The bra/ket notations are abbreviated here to $\left\vert \text{%
A}\right\rangle =\left\vert \text{A}\succ \text{B}\right\rangle $ and $%
\left\vert \text{B}\right\rangle =\left\vert \text{B$\succ $A}\right\rangle
$.   In the superposition state Equation~(\ref{QQS}), the additional terms with
interferences $\left\langle \text{A}\middle|\mathbb{H}\middle|\text{B}%
\right\rangle $ and $\left\langle \text{B}\middle|\mathbb{H}\middle|\text{A}%
\right\rangle $ appear in the sum: conditional alternatives do interfere in
quantum mechanics. However, it is not completely clear what measurable
physical quantities operator $\mathbb{H}$ might represent, when the two
basic states $\left\vert \text{A}\right\rangle $ and $\left\vert \text{B}%
\right\rangle $ represent alternative preferences, and how experiments
distinguishing mixed cognitive states from superimposed cognitive states in
decision-making can be carried out.

In the case of weak preferences, which have three possible outcomes, the
superposition and mixed wave functions become 
\begin{equation}
\left\vert \Psi \right\rangle =c_{1}\left\vert \text{A}\succ \text{B}%
\right\rangle +c_{2}\left\vert \text{A}\sim \text{B}\right\rangle
+c_{3}\left\vert \text{A}\prec \text{B}\right\rangle
\end{equation}%
\begin{equation}
\left\vert \Psi \right\rangle =c_{1}\left\vert \text{A}\succ \text{B}%
\right\rangle \left\vert \theta _{1}\right] +c_{2}\left\vert \text{A}\sim 
\text{B}\right\rangle \left\vert \theta _{_{2}}\right] +c_{3}\left\vert 
\text{A}\prec \text{B}\right\rangle \left\vert \theta _{3}\right]
\end{equation}%
while the preference operator $\mathbb{R}$ is given by the matrix%
\begin{equation}
\mathbf{\hat{R}}=%
\begin{bmatrix}
+1 & 0 & 0 \\ 
0 & 0 & 0 \\ 
0 & 0 & -1%
\end{bmatrix}%
\end{equation}

\subsection{Absolute Ranking in Quantum Case}

Consider a set of absolute ranks $r_{1},...,r_{k}$. If these ranks are
sharp, we can assume that $r_{i}=i$ without loss of generality. The state of
a single quantum element is specified by the wave function 
\begin{equation}
\left\vert \Psi \right\rangle =\sum_{i=1}^{k}c_{i}\left\vert
r_{i}\right\rangle =\sum_{i=1}^{k}c_{i}\left\vert i\right\rangle
\end{equation}%
where the absolute rank of state $\left\vert i\right\rangle $ is $%
r(\left\vert i\right\rangle )=i$, \textit{i.e.}, each element is generally a
superposition of different states with different absolute ranks.

If we consider two elements, A and B, the state of the system is a
superposition of states 
\begin{equation}
\left\vert \Psi _{\text{AB}}\right\rangle =\sum_{i,j=1}^{k}c_{ij}\left\vert
r_{i}(\text{A})\right\rangle \left\vert r_{j}(\text{B})\right\rangle
=\sum_{i,j=1}^{k}c_{ij}\left\vert i,j\right\rangle
\end{equation}%
where $\left\vert i,j\right\rangle $ denotes $\left\vert r_{i}(\text{A}%
)\right\rangle \left\vert r_{j}(\text{B})\right\rangle $.

If the states $\left\vert i,j\right\rangle $ are ordered so that $%
k_{1}=(k^{2}-k)/2$ states $i>j$ are followed by $k$ states $i=j$ and then by 
$k_{1}$ states states $i<j$, then the matrix $\left\langle i^{\prime
},j^{\prime }\middle|\mathbb{R}\middle|i,j\right\rangle $ of the preference
co-ranking operator $\mathbb{R}$ takes the following form 
\begin{equation}
\mathbf{\hat{R}}=\underset{k^{2}}{\underleftrightarrow{\left[ %
\begin{tabular}{c:c:c} $\underset{k_{1}}{\underleftrightarrow{\left[
+{\mathbf{\hat{I}}}^{\ }\right] }}$ & & \\ \hdashline &
$\underleftrightarrow{_{k}}$ & \\ \hdashline & &
$\overset{k_{1}}{\overleftrightarrow{\left[ -{\mathbf{\hat{I}}}^{\ }\right]
}}$\end{tabular}\right] }}
\end{equation}%
where the size of the matrix is $k\times k$ and $\mathbf{\hat{I}}$ is the $%
k_{1}\times k_{1}$ unit matrix while the remaining elements are zeros. The
preference co-ranking value is given by 
\begin{equation}
R(\text{A,B})=\sum_{i>j}P_{ij}(\text{A,B})-\sum_{i<j}P_{ij}(\text{A,B}),\ \
\ \ P_{ij}(\text{A,B})=\bar{c}_{ij}c_{ij}
\end{equation}%

This expression is the same for superimposed and mixed states of the wave
function $\Psi $. We state that A$\succ $B when $R($A,B$)>0$.

If $n$ quantum elements C$_{1},...,$C$_{n}$ are simultaneously considered,
then the wave function takes the form 
\begin{equation}
\left\vert \Psi _{\text{C}_{1},...,\text{C}_{n}}\right\rangle
=\sum_{i_{1},...,i_{n}=1}^{k}c_{i_{1}...i_{n}}\left\vert r_{i}(\text{C}%
_{1})\right\rangle ...\left\vert r_{j}(\text{C}_{n})\right\rangle
=\sum_{i_{1},...,i_{n}=1}^{k}c_{i_{1}...i_{n}}\left\vert
i_{1...}i_{n}\right\rangle
\end{equation}%
while the co-ranking of two elements becomes 
\begin{equation}
R(\text{C}_{q}\text{,C}_{p})=\sum_{i>j}P_{ij}(\text{C}_{q}\text{,C}%
_{p})-\sum_{i<j}P_{ij}(\text{C}_{q}\text{,C}_{p}),\ \ \ \ P_{i_{q}i_{p}}(%
\text{C}_{q}\text{,C}_{p})=\sum_{\mathbb{S}_{qp}}\bar{c}%
_{i_{1}...i_{n}}c_{i_{1}...i_{n}}
\end{equation}%
where $\mathbb{S}_{qp}=\{i_{1},...,i_{n}\}-\{i_{q},i_{p}\}$ is the set of
indices that includes all indices with exception of the indices $i_{q}$ and $%
i_{p}$. The probabilities $P_{ij}($C$_{q}$,C$_{p})$ are joint probability
distributions of ranks for two elements C$_{q}$ and C$_{p}$. The case of
interest when these distributions are independent (separable in quantum
terminology) that is for any two elements A and B 
\begin{equation}
P_{ij}(\text{A,B})=P_{i}(\text{A})P_{j}(\text{B})
\end{equation}%
In this case each element in quantum formulations becomes very similar to a
group with $k$ classical deterministic elements and weights $g_{i}=P_{i}$.

The principal difference between deterministic and quantum formulations is
that, in the quantum case, existence of absolute ranking does not guarantee
transitivity as stipulated in the following proposition

\begin{proposition}
A quantum (or probabilistic) preference can be intransitive even if it has
an absolute~ranking.
\end{proposition}

Indeed, Figure \ref{figQ} illustrates the case of three elements A, B and C
with independent probability distributions $P_{i}($A$),$ $P_{i}($B$)$ and $%
P_{i}($C$)$ over $k=9$ ranks $i=1,...,9$. This case follows the dice example
shown in Figure \ref{fig0}b. It can be seen that 
\begin{equation}
\text{A}\prec \text{B}\prec \text{C}\prec \text{A}
\end{equation}
since $R($B,A$)=R($C,B$)=R($A,C$)=1/9$.
\begin{figure}[H]
\centering
\includegraphics[width=14cm,page=11, clip ]{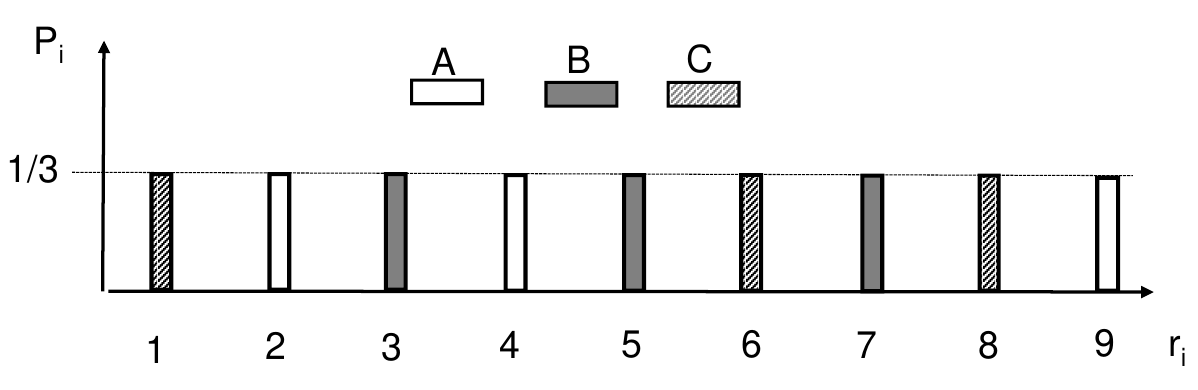}
\caption{Intransitive $\text{A}\prec\text{B}\prec\text{C}\prec\text{A}$ 
preference for quantum/probabilistic elements with absolute ranking $r_i$, where $P_i$ indicates probability of state $i$.   }
\label{figQ}

\end{figure}
%%%%%%%%%%%%%%%%%%%%%%%%%%%%%%%%%%%
%\input{figs.tex}
\vspace{-12pt}
%\pagebreak 
\conflictofinterests{Conflicts of Interest}

The author declares no conflict of interest.
% includegraphics [trim=left bottom right top, clip]

%trim=2.5cm 7cm 2cm 3cm,

%\begin{figure}[h!]
%\begin{center}
%\includegraphics[width=8cm,page=4,trim=6cm 5.5cm 5.5cm 4.5cm, clip ]{figppt.pdf}
%\caption{Explicit intransitivity of the choices for  {\it Economist} subscription with the use of s. 
%Notations are as in Figure \ref{fig3}.}
%\label{fig4}
%\end{center}
%\end{figure}

%%%%%%%%%%%%%%%%%%%%%%%%%%%%%%%%%%%
%%\bibliographystyle{mdpi}
%%\makeatletter
%%\renewcommand\@biblabel[1]{#1. }
%%\makeatother

%\bibliography{comp}

\begin{thebibliography}{----}
%\providecommand{\natexlab}[1]{#1}

\bibitem[Von~Neumann and Morgenstern(1953)]{VNM1953}
Von~Neumann, J.; Morgenstern, O.
\newblock {\em Theory of Games and Economic Behavior}; Princeton University
  Press: Princeton, NJ, USA, 1953.

\bibitem[Caratheodory(1909)]{Cara1909}
Caratheodory, C.
\newblock Studies in the foundation of thermodynamics.
\newblock {\em Math. Ann.} {\bf 1909},
\newblock {\em 67},~355--386.

\bibitem[Lieb and Yngvason(2003)]{EntOrd2003}
Lieb, E.H.; Yngvason, J.
\newblock The Entropy of Classical Thermodynamics. In {\em Entropy}; Greven,
  A., Keller, G., Warnecke, G., Eds.; Princeton University Press: Princeton, NY, USA, 2003;
\newblock Chapter~8, pp.~147--195.

\bibitem[Gyftopoulos and Beretta(1991)]{Beretta1991}
Gyftopoulos, E.; Beretta, G.
\newblock {\em Thermodynamics. Foundations and Applications}; Dover Publications: Mineola, NY,
  USA, 1991.

\bibitem[Klimenko(2012)]{K-PS2012}
Klimenko, A.Y.
\newblock Mixing, entropy and competition.
\newblock {\em Phys. Scr.} {\bf 2012},
\newblock {\em 85},~068201.

\bibitem[Klimenko(2013)]{K-PT2013}
Klimenko, A.Y.
\newblock Complex competitive systems and competitive thermodynamics.
\newblock {\em Phil. Trans. R. Soc. A} {\bf 2013}, {\em 371},
doi:10.1098/rsta.2012.0244.

\bibitem[Klimenko(2014)]{K_Ent2014a}
Klimenko, A.
\newblock Entropy and Equilibria in Competitive Systems.
\newblock {\em Entropy} {\bf 2014},
\newblock {\em 16},~1--22.

\bibitem[Allais(1953)]{Allais1953}
Allais, M.
\newblock Le Comportement de l'Homme Rationnel devant le Risque: Critique des
  Postulats et Axiomes de l'Ecole Americaine.
\newblock {\em Econometrica} {\bf 1953},
\newblock {\em 21},~503--546.

\bibitem[Edwards(1955)]{pref1955}
Edwards, W.
\newblock The prediction of decisions among bets.
\newblock {\em J. Exp. Psychol.} {\bf 1955},
\newblock {\em 50},~201--214.

\bibitem[Kahneman and Tversky(1982)]{pref1982T}
Kahneman, D.; Tversky, A.
\newblock The psychology of preferences.
\newblock  {\em Sci. Am.} {\bf 1982}, {\em 246}, 160--173.

\bibitem[Quiggin(1982)]{pref1982}
Quiggin, J.
\newblock A theory of anticipated utility.
\newblock {\em J. Econ. Behav. Organ.} {\bf 1982},
\newblock {\em 3},~323--343.

\bibitem[Fishburn(1983)]{pref1983}
Fishburn, P.C.
\newblock Transitive measurable utility.
\newblock {\em J. Econ. Theory} {\bf 1983},
\newblock {\em 31},~293--317.

\bibitem[Tversky and Kahneman(1992)]{pref1992}
Tversky, A.; Kahneman, D.
\newblock Advances in prospect theory: Cumulative representation of
  uncertainty.
\newblock {\em J. Risk Uncertain.} {\bf 1992},
\newblock {\em 5},~297--323.

\bibitem[Abe(2000)]{Abe1}
Abe, S.
\newblock Heat and generalized Clausius entropy of nonextensive systems. {\bf
  2000},  arXiv:cond-mat/0012115

\bibitem[Tsallis(2001)]{Tsallis1}
Tsallis, C. Nonextensive Statistical Mechanics and Thermodynamics: Historical
  Background and Present Status.
\newblock In {\em Nonextensive Statistical Mechanics and Its Applications};
  Abe, S., Okamoto, Y., Eds.; Springer: New York, NY, USA, 2001;
\newblock pp. 3--98.

\bibitem[Hanel and Thurner(2011)]{Thurner1}
Hanel, R.; Thurner, S.
\newblock When do generalized entropies apply? How phase space volume
  determines entropy.
\newblock {\em Europhys. Lett.} {\bf 2011},
\newblock {\em 96},~50003.

\bibitem[{de Condorcet}(1785)]{Cond1785}
{de Condorcet}, N.
\newblock {\em Essay on the Application of Analysis to the Probability of
  Majority Decisions}; De L'imprimerie Royale: Paris, France,
\newblock  1785.

\bibitem[Tversky(1969)]{Tversky1969}
Tversky, A.
\newblock Intransitivity of preferences.
\newblock {\em Psychol. Rev.} {\bf 1969},
\newblock {\em 76},~31--48.

\bibitem[Rubinstein(1988)]{Rubinstein1988}
Rubinstein, A.
\newblock Similarity and decision-making under risk: Is there a utility theory
  resolution to the Allais paradox?
\newblock {\em J. Econ. Theory} {\bf 1988},
\newblock {\em 1},~145--153.

\bibitem[Temkin(1996)]{Temkin1996}
Temkin, L.S.
\newblock A Continuum Argument for Intransitivity.
\newblock {\em Philos. Public Aff.} {\bf 1996},
\newblock {\em 25},~175--210.

\bibitem[Tullock(1964)]{Intrans1964}
Tullock, G.
\newblock The Irrationality of Intransitivity.
\newblock {\em Oxf. Econ. Pap.} {\bf 1964},
\newblock {\em 16},~401--406.

\bibitem[Anand(1993)]{Intrans1993}
Anand, P.
\newblock The philosophy of intransitive preference.
\newblock {\em Econ. J.} {\bf 1993},
\newblock {\em 103},~337--346.

\bibitem[Debreu(1964)]{Debreu1964}
Debreu, G.
\newblock Continuity properties of Paretian utility.
\newblock {\em Int. Econ. Rev.} {\bf 1964},
\newblock {\em 5},~285--293.

\bibitem[Nash(1950)]{Nash1950}
Nash, J.F.
\newblock Equilibrium Points in n-Person Games.
\newblock {\em Proc. Natl. Acad. Sci. USA} {\bf 1950},
\newblock {\em 36},~48--49.

\bibitem[Makowski(2009)]{Qcat2009M}
Makowski, M.
\newblock Transitivity \textit{vs.} intransitivity in decision making process---An
  example in quantum game theory.
\newblock {\em Phys. Lett. A} {\bf 2009},
\newblock {\em 373},~2125--2130.

\bibitem[Piotrowski and Makowski(2005)]{Icat2005PM}
Piotrowski, E.W.; Makowski, M.
\newblock Cat's dilemma  transitivity \textit{vs.} intransitivity.
\newblock {\em Fluct. Noise Lett.} {\bf 2005},
\newblock {\em 5}, L85--L95.

\bibitem[Makowski \em{et~al.}(2015)Makowski, Piotrowski, and
  Sladkowski]{Ent2015MPS}
Makowski, M.; Piotrowski, E.; Sladkowski, J.
\newblock Do Transitive Preferences Always Result in Indifferent Divisions?
\newblock {\em Entropy} {\bf 2015},
\newblock {\em 17},~968--983.

\bibitem[Makowski and Piotrowski(2011)]{Qelec2011MP}
Makowski, M.; Piotrowski, E.W.
\newblock Decisions in elections-transitive or intransitive quantum
  preferences.
\newblock {\em J. Phys. A} {\bf 2011},
\newblock {\em 44}, doi:10.1088/1751-8113/44/21/215303.

\bibitem[Persad(2010)]{SLR2010}
Persad, G.C.
\newblock Risk, Everyday Intuitions, and the Institutional Value of Tort Law.
\newblock {\em Stanf. Law Rev.} {\bf 2010},
\newblock {\em 62},~1445--1471.

\bibitem[Katz(2014)]{Katz2014}
Katz, L.
\newblock Rational Choice versus Lawful Choice.
\newblock {\em J. Inst. Theor. Econ.} {\bf 2014},
\newblock {\em 170},~105--121.

\bibitem[Feldman \em{et~al.}(2002)Feldman, Riley, Kerr, and Bohannan]{LV2002}
Feldman, M.W.; Riley, M.A.; Kerr, B.; Bohannan, B.J.M.
\newblock Local dispersal promotes biodiversity in a real-life game of
  rock-paper-scissors.
\newblock {\em Nature} {\bf 2002},
\newblock {\em 418},~171--174.

\bibitem[Reichenbach \em{et~al.}(2007)Reichenbach, Mobilia, and Frey]{LV2007}
Reichenbach, T.; Mobilia, M.; Frey, E.
\newblock Mobility promotes and jeopardizes biodiversity in rock-paper-scissors
  games.
\newblock {\em Nature} {\bf 2007},
\newblock {\em 448},~1046--1049.

\bibitem[Marshall \em{et~al.}(1982)Marshall, Webb, {Sepkoski,Jr}, and
  Raup]{GAI1982}
Marshall, L.; Webb, S.D.; {Sepkoski}, J.J., Jr.; Raup, D.M.
\newblock Mammalian Evolution and the Great American Interchange.
\newblock {\em Science} {\bf 1982},
\newblock {\em 215},~1351--1357.

\bibitem[Avelino \em{et~al.}(2014)Avelino, Bazeia, Menezes, and
  De~Oliveira]{LV2014}
Avelino, P.P.; Bazeia, D.; Menezes, J.; de~Oliveira, B.F.
\newblock String networks in ZN Lotka-Volterra competition models.
\newblock {\em Phys. Lett. A} {\bf 2014},
\newblock {\em 378},~393--397.

\bibitem[Lotka(1920)]{LV1920}
Lotka, A.J.
\newblock Undamped oscillations derived from the law of mass action.
\newblock {\em J. Am. Chem. Soc.} {\bf 1920},
\newblock {\em 42},~1595--1599.

\bibitem[Volterra(1939)]{LV1939}
Volterra, V.
\newblock The general equations of biological strife in the case of historical
  actions.
\newblock {\em Proc. Edinb. Math. Soc.} {\bf 1939},
\newblock {\em 6},~4--10.

\bibitem[Arrow(1951)]{Arrow}
Arrow, K.J.
\newblock {\em Social Choice and Individual Values}; Yale University Press: New Haven, CT, USA, 1951.

\bibitem[Ariely(2008)]{Dan2008}
Ariely, D.
\newblock {\em Predictably Irrational}; HarperCollins: New York, NY, USA,
\newblock  2008.

\bibitem[Ng(1977)]{Ng1977}
Ng, Y.K.
\newblock Sub-semiorder: A model of multidimensional choice with preference
  intransitivity.
\newblock {\em J.~Math. Psychol.} {\bf 1977},
\newblock {\em 16},~51--59.

\bibitem[Scholten and Read(2013)]{risk2013}
Scholten, M.; Read, D.
\newblock Prospect theory and the forgotten fourfold pattern of risk
  preferences.
\newblock {\em J.~Risk Uncertain.} {\bf 2013},
\newblock {\em 48},~67--83.

\bibitem[Leland(2002)]{Leland2002}
Leland, J.W.
\newblock Similarity judgements and anomalies in intertemporal choice.
\newblock {\em Econ. Inq.} {\bf 2002},
\newblock {\em 40},~574--581.

\bibitem[Lorentziadis(2013)]{Lorentz2013}
Lorentziadis, P.L.
\newblock Preference under risk in the presence of indistinguishable
  probabilities.
\newblock {\em Oper.~Res.} {\bf 2013},
\newblock {\em 13},~429--446.

\bibitem[Klimenko(2014)]{K_JSSSE}
Klimenko, A.Y.
\newblock Complexity and intransitivity in technological development.
\newblock {\em J. Syst. Sci. Syst.~Eng.} {\bf 2014},
\newblock {\em 23},~128--152.

\bibitem[Klimenko(2012)]{K-OTJ2012}
Klimenko, A.Y.
\newblock Teaching the third law of thermodynamics.
\newblock {\em Open Thermodyn. J.} {\bf 2012},
\newblock {\em 6},~1--14. 

\bibitem[Smith \em{et~al.}(1988)Smith, Depew, and Weber]{EIE1988}
Smith, J.D.; Depew, D.J.; Weber, B.H.
\newblock {\em Entropy, Information, and Evolution: New Perspectives on
  Physical and Biological Evolution}; MIT Press: Cambridge, MA, USA,
\newblock  1988.

\bibitem[Meyer(1999)]{Qgame1999}
Meyer, D.A.
\newblock Quantum Strategies.
\newblock {\em Phys. Rev. Lett.} {\bf 1999},
\newblock {\em 82},~1052--1055.

\bibitem[Klimenko(2014)]{Kaon2014K}
Klimenko, A.Y.
\newblock A note on invariant properties of a quantum system placed into
  thermodynamic environment.
\newblock {\em Physica A} {\bf
  2014},
\newblock {\em 398},~65--75.

\end{thebibliography}

\end{document}